\documentclass{siamart220329}

\usepackage{amsfonts}
\usepackage{amssymb}
\usepackage{mathtools}
\usepackage{algpseudocode}
\usepackage{soul}
\usepackage{xcolor}
\usepackage{caption}
\usepackage{subcaption}
\usepackage{graphicx}
\usepackage[margin=1in]{geometry}
\usepackage{float}
\usepackage{setspace}
\floatstyle{plaintop}
\usepackage{lipsum}

 \newsiamthm{assumption}{Assumption}

\newcommand{\E}[1]{\mathbb{E}\left(#1\right)}

\newcommand{\Var}[1]{\operatorname{Var}\bigl(#1\bigr)}
\newcommand{\V}[1]{\mathbb{V}\left(#1\right)}
\newcommand{\Vhat}[1]{\widehat{\mathbb{V}}\left(#1\right)}

\newcommand{\Vhatsub}[2]{\widehat{\mathbb{V}}_{#1}\left(#2\right)}
\newcommand{\Vbarsub}[2]{\bar{\mathbb{V}}_{#1}\left(#2\right)}
\newcommand{\DeltaVhatsub}[2]{\widehat{\Delta \mathbb{V}}_{#1}\left(#2\right)}
\newcommand{\Cov}[2]{\operatorname{Cov}\left(#1, #2\right)}
\newcommand{\Covs}[2]{\mathbb{C}_{M_{\ell}}\left(#1, #2\right)}
\newcommand{\Covsup}[3]{\operatorname{Cov}^{#1}\left(#2, #3\right)}
\newcommand{\Bias}[1]{\operatorname{Bias}\left(#1\right)}
\newcommand{\Msup}[2]{\mathbb{M}^{#1}\left(#2\right)}
\newcommand{\Esup}[2]{\left[\mathbb{E}\left(#2\right)\right]^{#1}}
\newcommand{\Varsup}[2]{\left[\operatorname{Var}\left(#2\right)\right]^{#1}}

\newcommand{\converge}[1]{\xrightarrow[]{~~#1~~}}
\newcommand{\bfone}[1]{\mathbf{1}\left\{ #1 \right\}}

\newcommand{\zhat}[1][\theta]{\widehat{Z}\left( #1 \right)}
\newcommand{\zhatsub}[3][\theta]{\widehat{Z}_{#2}\left( #1, M_{#3,\epsilon}, \varpi_{#3} \right)}
\newcommand{\deltaZhatsub}[3][\theta]{\widehat{\Delta{Z}}_{#2}\left( #1, M_{#3,\epsilon}, \varpi_{#3} \right)}
\newcommand{\deltaZhatsubshort}[4][\theta]{\widehat{\Delta{Z}}^{#2}_{#4}\left( #1, \varpi_{#3} \right)}

\DeclarePairedDelimiter{\ceil}{\lceil}{\rceil}
\DeclareMathOperator*{\argmax}{arg\,max}

\def\cA{\mathcal{A}}

\def\cL{\mathcal{L}}

\def\cN{\mathcal{N}}
\def\cO{\mathcal{O}}
\def\cP{\mathcal{P}}

\def\cT{\mathcal{T}}
\def\cV{\mathcal{V}}

\def\cY{\mathcal{Y}}

\def\cZ{\mathcal{Z}} 
\def\b1{{\mathbf 1}}

\makeatletter
\newcommand{\algmultiline}[1]{%
    \begin{tabular}[t]{@{}p{\dimexpr\linewidth-\ALG@thistlm}@{}}
        #1
    \end{tabular}
}
\makeatother

\title{Multilevel Monte Carlo Metamodeling for Variance Function Estimation\thanks{Submitted to the editors on September 1, 2024.\funding{This work was  supported by the National Science Foundation CAREER grant [CMMI-1846663].}}}

\author{Jingtao Zhang\thanks{Grado Department of Industrial and Systems Engineering, Virginia Tech, USA (\email{jingtaozhang@vt.edu}, \email{xchen6@vt.edu}).} \and  Xi Chen\footnotemark[2]}

\begin{document}  

\maketitle

\begin{abstract}
This work introduces a novel multilevel Monte Carlo (MLMC) metamodeling approach for variance function estimation. Although devising an efficient experimental design for simulation metamodeling can be elusive, the MLMC-based approach addresses this challenge by dynamically adjusting the number of design points and budget allocation at each level, thereby automatically creating an efficient design. Theoretical analyses show that, under mild conditions, the proposed MLMC metamodeling approach for variance function estimation can achieve superior computational efficiency compared to standard Monte Carlo metamodeling while achieving the desired level of accuracy. Additionally, this work establishes the asymptotic normality of the MLMC metamodeling estimator under certain sufficient conditions, providing valuable insights for uncertainty quantification. Finally, two MLMC metamodeling procedures are proposed for variance function estimation: one to achieve a target accuracy level and another to efficiently utilize a fixed computational budget. Numerical evaluations support the theoretical results and demonstrate the potential of the proposed approach in facilitating global sensitivity analysis.  
\end{abstract}



\section{Introduction}
The need to perform data analysis under heteroscedasticity, which means nonconstant variance across the input space, frequently arises in various fields. Examples abound, ranging from engineering and economics to medical and physical sciences: robust optimization utilizing mean and variance information \cite{dellino2012robust, iwazaki2021mean,postek2018robust,yanikouglu2016robust}, mean-variance portfolio optimization \cite{lai2011mean}, heteroscedastic regression analysis of e-commerce and cross-sectional data \cite{van2010semiparametric, weltz2023experimental}, volatility modeling and analysis of financial returns \cite{wu2014gaussian}, and variance-based global sensitivity analysis \cite{marrel2012global}. The prevalence of such a problem has given rise to considerable interest in developing novel design and analysis methods for variance function estimation. In this work, we propose a multilevel Monte Carlo (MLMC) metamodeling approach for variance function estimation in the stochastic simulation context.

Simulation metamodeling is a simulation modeling and analysis technique that involves developing a simplified mathematical or statistical model (a metamodel) to approximate the performance measure of interest in a complex simulation model as a function of the input variables. The metamodel is built using data from a designed simulation experiment and serves as a valuable tool for reducing computational costs and accelerating the decision-making process \cite{ankenman2010stochastic}. 
Variance function estimation through simulation metamodeling involves two key components: the design of simulation experiments and the methods for approximating the variance function. 
Various approaches are available for estimation purposes, including parametric methods \cite{atkinson1995d, goos2001optimal, vining1996experimental}, semi-parametric methods \cite{van2010semiparametric}, and nonparametric methods. Nonparametric techniques, such as splines, kernel smoothing, and Gaussian process regression, are widely used in practical applications due to their flexibility and robustness \cite{brown2007variance, cai2019simultaneous, dellino2012robust, liu2007smoothing, wang2018adaptive}.
	
While most simulation metamodeling literature focuses on mean function estimation, simulation metamodeling for variance function estimation remains relatively underdeveloped. In particular,  successful applications of simulation metamodeling depend heavily on meticulous experimental design, which remains an area of exploration for variance function estimation. This work addresses this gap by combining MLMC with metamodeling techniques, resulting in an efficient experimental design for variance function estimation via simulation metamodeling. Notably, our proposed MLMC metamodeling approach does not impose using a specific estimation method, offering the flexibility to employ suitable techniques for variance function estimation.

MLMC is an advanced computational technique that enhances the efficiency of standard Monte Carlo (SMC) methods for estimating quantities of interest in stochastic simulation. MLMC was initially devised for parametric integration and its associated applications \cite{heinrich1998monte, heinrich2000multilevel, heinrich2001multilevel,heinrich2006monte,heinrich1999monte}. Kerbaier (2005) extended its application to path simulations by introducing a two-level MC framework \cite{kebaier2005statistical}. Giles (2008) subsequently generalized this framework to accommodate multiple levels in MC path simulations, effectively reducing the computational complexity associated with estimating expected values stemming from stochastic differential equations \cite{giles2008multilevel}. 
Haji-Ali et al.\ (2016) proposed the multi-index Monte Carlo method, which integrates the sparse grid concept to extend MLMC, enabling more efficient handling of high-dimensional integration \cite{haji2016multi}.
The introduction of MLMC has led to  several related research endeavors;  see, e.g., \cite{barth2011multi} and \cite{cliffe2011multilevel}.   
While the initial focus of MLMC research  centered on estimating expectations, the field has since broadened its horizons to encompass a variety of statistical parameters, including nested expectation \cite{giles2019multilevel}, distribution functions  and densities \cite{giles2015multilevel}, failure probabilities \cite{elfverson2016multilevel, wagner2020multilevel}, variances and covariances \cite{bierig2015convergence, mycek2019multilevel}, and higher-order central moments \cite{bierig2016estimation}. 
Although the original MLMC framework was developed for point estimation, subsequent research has extended it to function estimation. For example, Krumscheid and Nobile (2018) investigated the estimation of parametric expectations using interpolation techniques  \cite{krumscheid2018multilevel}. Additionally, Chernov and Schetzke (2023) proposed a bias-free MLMC method for approximating the covariance functions of sufficiently regular random fields in tensor product Hilbert spaces. Their approach, based on the MLMC technique, achieves nearly optimal computational complexity \cite{chernov2023simple}.   
  The pioneering work by Rosenbaum and Staum (2017) introduced the concept of MLMC metamodeling, integrating MLMC into simulation metamodeling for mean function estimation \cite{rosenbaum2017multilevel}. It showcased superior computational efficiency compared to conventional simulation metamodeling methods relying on SMC. A defining strength of MLMC metamodeling lies in its ability to create an efficient experimental design by adaptively expanding design points and allocating the simulation budget. This dynamic approach, involving integrated design and analysis based on the already obtained simulation outputs, effectively balances computational efficiency and efficacy, potentially heralding a new paradigm for simulation metamodeling.

This work proposes an efficient MLMC metamodeling approach designed explicitly for variance function estimation, inspired by the methodologies in \cite{rosenbaum2017multilevel} and \cite{mycek2019multilevel}. The approach constructs a variance function estimator as a telescoping sum of metamodeling estimators, each corresponding to a specific level of accuracy. The accuracy of each estimator is determined by the level-specific experimental design used to run the simulation model and generate outputs. At each level, the design specifies the number of design points and the number of replications at each point required to run the simulation model. As the design level increases, both the accuracy of the corresponding level-specific metamodeling estimator and the computational cost to generate simulation outputs according to the specified design also increase. By strategically combining metamodeling estimators from different design levels, the proposed approach delivers an accurate variance function estimator with significant computational savings compared to SMC metamodeling.
We emphasize the key differences between our approach and the existing literature. 
Rosenbaum and Staum (2017) investigated MLMC metamodeling for mean function estimation, demonstrating significant computational efficiency underpinned by rigorous theoretical foundations \cite{rosenbaum2017multilevel}. Their approach uses sample means as point estimators at design points, supported by a carefully structured experimental design. However, this methodology does not directly extend to high-order moment functions due to its dependence on the linearity of sample means. Since this work focuses on variance function estimation, we conduct detailed analyses based on the properties of variance estimates. Our approach also diverges from the classical MLMC methods for point estimation of variance explored by Mycek and De Lozzo (2019) \cite{mycek2019multilevel}. Their method employs an MLMC point estimator with level-specific estimators that are sample variances obtained from running simulation models of varying levels of accuracy.
In contrast, our work employs a single simulation model and develops an MLMC metamodeling estimator that includes level-specific variance function estimators with varying statistical properties; these properties are determined by an experimental design tailored to each design level. Furthermore, while Mycek and De Lozzo (2019) based their theoretical analysis on the properties of sample variances, our approach adopts a function approximation perspective.

Our contribution encompasses three fundamental aspects.
First, we extend the MLMC metamodeling methodology to encompass variance function estimation and provide an explicit, novel MLMC metamodeling estimator tailored for this purpose. We broaden the theory related to the computational complexity of MLMC metamodeling, offering theoretical insights into its efficiency for variance function estimation.
Second, we conduct a comprehensive asymptotic analysis of the different components involved in the MLMC metamodeling estimator and establish their asymptotic normality under mild technical conditions. One key finding is that a computationally efficient MLMC metamodeling estimator exhibits asymptotic normality.
Third, guided by our theoretical findings, we propose two variance function estimation procedures via MLMC metamodeling. The first procedure aims to achieve a target mean integrated squared error level, while the second is suitable for scenarios with a fixed computational budget. Both procedures have versatile applications, and our numerical experiments demonstrate their superior efficiency and efficacy compared to SMC metamodeling.

The remainder of this work is organized as follows. Section \ref{sec:mlmc} provides an overview of MLMC metamodeling for variance function estimation. Section \ref{sec:theory} presents theoretical analyses of the proposed approach. Section \ref{sec:clt} establishes the asymptotic normality of the MLMC metamodeling estimator  under mild technical conditions. Section \ref{sec: mlmc_procedure} proposes two MLMC metamodeling procedures for different implementation purposes. Section~\ref{sec:exp} demonstrates the performance of MLMC metamodeling  for  variance function estimation in comparison with SMC metamodeling through numerical studies. Finally, Section \ref{sec:conclusion} concludes the paper.

\section{Multilevel Monte Carlo Metamodeling for Variance Function Estimation}\label{sec:mlmc}
\numberwithin{equation}{section}

This section introduces the concept underpinning MLMC metamodeling for variance function estimation.  Subsection \ref{subsec:setup} provides the basic setup, while Subsection \ref{subsec:MLMCM} presents MLMC metamodeling specifically for variance function estimation.

\subsection{Overview of the Basic Setup}\label{subsec:setup}
We consider a simulation model with an output $\cY(\theta, \omega): \Theta \times \Omega \to \mathbb{R}$, which is a measurable function in the probability space $(\Omega, \mathcal{F}, \mathbb{P})$. Here, $\theta$ represents the $d$-dimensional input vector within the input space $\Theta \subseteq \mathbb{R}^{d}$, while $\omega$ denotes a realization drawn from $\Omega$ using a random number stream $\varpi$, capturing the inherent randomness of the simulation model. For simplicity, we use the shorthand notation $\cY(\theta)$ to represent $\cY(\theta, \omega)$ when there is no risk of ambiguity. Define the variance of the simulation outputs at a  point $\theta \in \Theta$ as $\Var{\cY(\theta, \omega)} \coloneqq \int \left(\cY(\theta, \omega) - \E{\cY(\theta, \omega)} \right)^2\,d \mathbb{P}(\omega)$, where the mean of the outputs at $\theta$ is given by $\E{\cY(\theta, \omega)} \coloneqq \int \cY(\theta, \omega)\,d \mathbb{P}(\omega)$. Given a sample of $M$ outputs $\{\cY(\theta, \omega_1), \cY(\theta, \omega_2), \dots, \cY(\theta, \omega_M) \}$ generated at $\theta$, where $\omega_1, \omega_2, \dots, \omega_M$ are drawn using a random number stream $\varpi$,  the sample variance can be calculated as  $\cV(\theta, M, \varpi) \coloneqq (M-1)^{-1}\sum_{m=1}^{M}(\cY(\theta, \omega_m) - \bar{\cY}(\theta))^2$, with $\bar{\cY}(\theta) \coloneqq M^{-1}\sum_{m=1}^{M}\cY(\theta, \omega_m)$ denoting the sample average output at  $\theta$.  Our primary interest lies in estimating the variance function, $\V{\cdot}$, defined as $\V{\theta} \coloneqq \Var{\cY(\theta, \omega)}$ for any $\theta \in \Theta$, with its estimator denoted by $\Vhat{\cdot}$. 

In the remainder of this work, the following notation is consistently adopted. Define $[N] \coloneqq \{0, 1, \dots, N\}$ and $[N]^{+} \coloneqq [N] \setminus \{0\}$. Let $\lceil \cdot \rceil$ denote the ceiling function, and $|\mathcal{P}|$ denote the cardinality of  set $\mathcal{P}$. The function $\bfone{\cdot}$ denotes the indicator function, and $\|\cdot\|_p$ denotes the $p$-norm for $p \geq 1$. For the $2$-norm, we use the shorthand $\|\cdot\|$ when there is no risk of confusion. Additionally,  define $ \mbox{diam}(\Theta) := \sup\{\|\theta-\theta^\prime\|:  \theta,\theta^\prime \in \Theta\}$ as the maximum distance between two input points in $ \Theta $. We use $\Longrightarrow$ to denote convergence in distribution (page 116, \cite{durrett2019probability})  and $\converge{p}$ to denote convergence in probability (page 56, \cite{durrett2019probability}).

\subsection{Variance Function Estimation via Multilevel Monte Carlo Metamodeling}\label{subsec:MLMCM}

A brief summary of classical MLMC will facilitate our discussion of MLMC metamodeling. MLMC is a powerful computational tool, especially useful for point estimation of expectations or quantities of interest arising from complex stochastic systems.
The key idea behind MLMC is to estimate the quantity of interest by using multiple levels of computational (or simulation) models, each corresponding to a different level of discretization or resolution. Each level has its own computational cost and accuracy. Typically, finer levels provide more accurate estimators but are computationally expensive, while coarser levels are cheaper but less accurate. MLMC leverages the correlation between estimators at different levels, stemming from their shared underlying randomness. By strategically aggregating estimators from these levels, MLMC can achieve an accurate estimator with significant computational savings compared to standard Monte Carlo (SMC) methods.
In the same spirit, we explore MLMC metamodeling, which uses a hierarchy of metamodel-based estimators to estimate the variance function 
 $\V{\cdot}$. This hierarchy comprises metamodels at different levels of computational cost and accuracy.

To facilitate the understanding of MLMC metamodeling and highlight its advantages over SMC metamodeling, we first provide a concise review of the latter.
SMC metamodeling typically employs an experimental design that consists of a single level of design points. To estimate the variance function $\V{\cdot}$, consider an SMC metamodeling estimator taking the form of linear smoothers: $\Vhatsub{}{\theta, M, \varpi} = \sum_{i = 1}^{N}w_i(\theta)\cV(\theta_i, M, \varpi)$ for any $\theta \in \Theta$, where $\cV(\theta_i, M, \varpi)$ denotes the sample variance obtained from running $M$ independent simulation replications at design point $\theta_i$ using the random number stream $\varpi$, and $w_i(\theta)$ denotes the weight for $\cV(\theta_i, M, \varpi)$ for any $i\in [N]^{+}$.  
Constructing an SMC metamodel that adequately approximates $\V{\cdot}$ globally typically requires a large design-point set and a high number of replications (i.e., large $N$ and $M$), resulting in significant computational costs \cite{rosenbaum2017multilevel, wang2018adaptive}. This limitation underscores the necessity for a more efficient experimental design for variance function estimation through simulation metamodeling.

Instead of relying on a single level of design points,  MLMC metamodeling considers using $L+1$ levels of design-point sets. Specifically, on the $\ell$th level, the design-point set $\mathcal{T}_{\ell} $ contains $N_{\ell}:=|\cT_\ell|$ design points, for any $\ell \in [L]$. As the design level $\ell$ increases, the design points become increasingly dense in $\Theta$, i.e., $N_{\ell} > N_{\ell-1}$ for any $\ell \in [L]$, with $N_{-1} \coloneqq 0$ for consistency.
The MLMC metamodeling estimator is constructed from a collection of level-specific estimators which serve as  variance function estimators at different levels.
Specifically, the  $\ell$th level variance function estimator is an SMC metamodeling estimator,  derived by running $ M_{\ell}$ independent replications at each design point in $\mathcal{T}_{\ell}$ using the random number stream $\varpi_{\ell}$, for any $\ell \in [L]$.   \emph{This $\ell$th level variance function estimator} is expressed as follows:
\begin{equation}
	\label{eq:metamodel}
	\Vhatsub{\ell}{\theta, M_{\ell}, \varpi_{\ell}} = \sum_{i = 1}^{N_{\ell}}w_i^{\ell}(\theta)\cV(\theta_i^{\ell}, M_{\ell}, \varpi_{\ell}) \ ,
\end{equation}
where $\cV(\theta_i^{\ell}, M_{\ell}, \varpi_{\ell})$ denotes the sample variance obtained at design point $\theta_i^{\ell} \in \cT_\ell$,  and $w_i^{\ell}(\theta)$ denotes the weight assigned for $\cV(\theta_i^{\ell}, M_{\ell}, \varpi_{\ell})$, for any $i\in [N_{\ell}]^{+}$ and $\ell \in [L]$. 
Furthermore, for each level $\ell \in [L-1]$, we construct an auxiliary  estimator, $\Vhatsub{\ell}{\theta, M_{\ell+1}, \varpi_{\ell+1}}$, by performing $M_{\ell+1}$ simulation replications at each design point in $\mathcal{T}_{\ell}$
using the random number stream $\varpi_{\ell+1}$. Notice that  the  auxiliary  estimator $\Vhatsub{\ell}{\theta, M_{\ell+1}, \varpi_{\ell+1}}$ shares the same design-point set (i.e., $\mathcal{T}_{\ell}$) as
the $\ell$th  level  estimator 	$\Vhatsub{\ell}{\theta, M_{\ell}, \varpi_{\ell}}$. However, it is constructed using the same number of replications  (i.e., $ M_{\ell+1}$) and the same random number stream (i.e., $\varpi_{\ell+1}$)    as  those  used for constructing the $(\ell+1)$th  level estimator $\Vhatsub{\ell+1}{\theta, M_{\ell+1}, \varpi_{\ell+1}}$. 
MLMC metamodeling combines the level-specific estimators $\{\Vhatsub{\ell}{\theta, M_{\ell}, \varpi_{\ell}}\}_{\ell \in [L]}$ and their corresponding auxiliary estimators into the following variance function estimator:
\begin{equation}\label{eq:mlmc_metamodeling_estimator_origin}
	\Vhat{\theta} = \Vhatsub{L}{\theta, M_{L}, \varpi_{L}} + \sum_{\ell=0}^{L-1}\left(\Vhatsub{\ell}{\theta, M_{\ell}, \varpi_{\ell}} - \Vhatsub{\ell}{\theta, M_{\ell+1}, \varpi_{\ell+1}}\right) \ .
\end{equation}

The form of the MLMC metamodeling estimator given in \eqref{eq:mlmc_metamodeling_estimator_origin} is a combination of the finest level SMC metamodeling estimator $\Vhatsub{L}{\theta, M_{L}, \varpi_{L}}$ and a control variate \cite{nelson1990control, rosenbaum2017multilevel}, represented by the second term on the right-hand side of \eqref{eq:mlmc_metamodeling_estimator_origin}. This structure highlights two key attributes. First, the accuracy of $\Vhat{\theta}$ is determined by the finest level metamodeling estimator, $\Vhatsub{L}{\theta, M_{L}, \varpi_{L}}$. The matching bias between each level-specific variance estimator $\Vhatsub{\ell}{\theta, M_{\ell}, \varpi_{\ell}}$ and the corresponding auxiliary estimator $\Vhatsub{\ell}{\theta, M_{\ell+1}, \varpi_{\ell+1}}$, for any $\ell \in [L-1]$, ensures that this control variate introduces no additional bias. Second, this control variate exhibits a hierarchical structure that leverages the correlations  between the auxiliary estimators and the level-specific estimators at each successive level, thereby facilitating variance reduction for the MLMC metamodeling estimator $\Vhat{\theta}$.
We can rewerite the MLMC metamodeling estimator in \eqref{eq:mlmc_metamodeling_estimator_origin}   succinctly as follows, which facilitates subsequent analyses:
\begin{equation}
	\label{eq:mlmc_metamodeling_estimator}
	\Vhat{\theta} = \sum_{\ell=0}^{L} \DeltaVhatsub{\ell}{\theta, M_{\ell}, \varpi_{\ell}}  \  ,
\end{equation}
where $\DeltaVhatsub{\ell}{\theta, M_{\ell}, \varpi_{\ell}} := \Vhatsub{\ell}{\theta, M_{\ell}, \varpi_{\ell}}-\Vhatsub{\ell-1}{\theta, M_{\ell}, \varpi_{\ell}}$ denotes \emph{the $\ell$th level refinement estimator} for any $\ell \in [L]$, with $\Vhatsub{-1}{\theta, M_{0}, \varpi_{0}} :=0$.

We next outline the theoretical framework that supports the analysis of the MLMC metamodeling estimator in Sections~\ref{sec:theory} and~\ref{sec:clt}  and the development of efficient implementation procedures in Section~\ref{sec: mlmc_procedure}.

\section{Theoretical Framework for MLMC Metamodeling}\label{sec:theory}
\numberwithin{equation}{section}

This section presents a theoretical analysis of the MLMC metamodeling approach for variance function estimation. We begin by outlining the theoretical framework, followed by a detailed examination of the bias and variance of the proposed MLMC metamodeling estimator in Subsections~\ref{sec:bias-analysis} and~\ref{subsec:variance_analysis}, respectively. Building on this analysis, Subsection~\ref{sec:computation complexity} delves into the computational complexity assessment for MLMC metamodeling.
  
 In classical MLMC, the estimator for a quantity of interest is constructed as a telescoping sum of estimators computed at different levels. Each level corresponds to a different level of accuracy and computational cost. By combining these estimators appropriately, the overall variance of the MLMC estimator can be reduced compared to using a single-level SMC estimator \cite{giles2015multilevel}.
The theoretical framework of MLMC involves analyzing the error and computational cost associated with each level and devising strategies for optimally allocating computational resources across levels to minimize the overall cost while meeting a prescribed accuracy requirement. This typically involves striking a trade-off between bias and variance reduction at different levels.  
 
The MLMC metamodeling estimator given in \eqref{eq:mlmc_metamodeling_estimator} can be viewed as a telescoping sum of SMC metamodeling estimators constructed at different design levels.  The theoretical framework for MLMC metamodeling aligns with that of classical MLMC,  enabling an analysis of the bias and variance of   metamodeling estimators constructed at each level,  in relation to the expansion of the design-point set and the allocation of the computational budget as the design level increases.

We follow \cite{rosenbaum2017multilevel} to use the mean integrated squared error (MISE) for assessing the performance of $\Vhat{\cdot}$ in estimating $\V{\cdot}$:
\begin{align} \label{eq:mise}
\operatorname{MISE}(\widehat{\mathbb{V}}) 
&=\E{\|\widehat{\mathbb{V}}-\mathbb{V}\|_{2}^{2}}
=\int_{\Theta} \Var{\Vhat{\theta}} \mathrm{d} \theta+\int_{\Theta}\left(\E{\Vhat{\theta}}-\V{\theta}\right)^{2} \mathrm{d} \theta \nonumber \\
&=\left\|\operatorname{Var}(\widehat{\mathbb{V}})\right\|_{1}+\left\|\operatorname{Bias}(\widehat{\mathbb{V}})\right\|_{2}^{2}  \ . 
\end{align}
Equation~\eqref{eq:mise} indicates that an effective estimator should balance the variance and bias components to minimize the MISE. To facilitate the analysis, we introduce the following technical assumptions.

\begin{assumption}\label{asm:domain}
	The input space $ \Theta \subseteq \mathbb{R}^{d} $ is a compact set.
\end{assumption}

\begin{assumption}\label{asm:finite_moment}
 There exists some $\theta^{\prime} \in \Theta$ such that $ \Msup{4}{\cY(\theta^{\prime})} \coloneqq \mathbb{E}\left[\left(\cY(\theta^{\prime})  - \mathbb{E}\left[\cY(\theta^{\prime})\right]\right)^4\right] < \infty$.
\end{assumption}

\begin{assumption}\label{asm:LP_Y}
	There exists a random variable $ \kappa_{y} $ with $ \E{\kappa_{y}^{4}} < \infty$ such that
\[ \left|\cY\left(\theta_{1}, \omega\right)-\cY\left(\theta_{2}, \omega\right)\right| \leq \kappa_{y}(\omega)\left\|\theta_{1}-\theta_{2}\right\| \text{ almost surely}, \quad  \forall  \theta_{1}, \theta_{2} \in \Theta .\]
\end{assumption}

Assumptions \ref{asm:domain} and \ref{asm:finite_moment} impose some conditions on the input space and simulation outputs to ensure proper estimation of the variance function. Assumption \ref{asm:LP_Y} is a Lipschitz continuity condition on simulation outputs. Similar conditions are commonly stipulated in contexts such as distribution and density function estimation \cite{elfverson2016multilevel,giles2015multilevelfunction}  and mean function estimation \cite{rosenbaum2017multilevel}. Assumption  \ref{asm:LP_Y} is mild and relatively easy to satisfy. For an illustrative example, consider $\cY(\theta, \omega) = \theta \Phi^{-1}(\omega)$, where $\Phi$ denotes the cumulative distribution function of a standard normal random variable, with $\theta \in \Theta \subset \mathbb{R}$  and $\omega$ being a uniform random variable on $[0,1]$. In this case, setting $\kappa_{y} = |\Phi^{-1}(\omega)|$ satisfies Assumption~\ref{asm:LP_Y}. 

We first present the following two lemmas, based on Assumptions \ref{asm:domain} through \ref{asm:LP_Y},  to facilitate  subsequent analyses. Their proofs are deferred to Appendices \ref{app:lem_finite_moment} and \ref{app:lem_LP_V}, respectively. 

\begin{lemma}\label{lem:finite_moment}
	Under Assumptions \ref{asm:domain}, \ref{asm:finite_moment},  and \ref{asm:LP_Y},  $c_{\cY} \coloneqq \sup_{\theta \in \Theta} \Msup{4}{\cY(\theta)} < \infty$.
\end{lemma}

\begin{lemma}\label{lem:LP_V}
Under Assumptions \ref{asm:domain}, \ref{asm:finite_moment} and \ref{asm:LP_Y}, there exists a constant $ \kappa_{v} > 0 $ such that
\[ \left| \V{\theta_{1}} -\V{\theta_{2}} \right| \leq \kappa_v\left\|\theta_{1}-\theta_{2}\right\|, \quad  \forall  \theta_{1}, \theta_{2} \in \Theta \ .\]
\end{lemma}

\subsection{Bias}
\label{sec:bias-analysis}
This subsection first provides a bound on  the bias of the $\ell$th level variance function estimator $\Vhatsub{\ell}{\theta, M_{\ell}, \varpi_{\ell}}$ given in \eqref{eq:metamodel}, which is used to further bound the bias component $\left\|\Bias{\widehat{\mathbb{V}}}\right\|_{2}^{2}$ in \eqref{eq:mise}. 
Throughout this subsection, we use $\Vhatsub{\ell}{\theta}$ as  shorthand for $\Vhatsub{\ell}{\theta, M_{\ell}, \varpi_{\ell}}$ to simplify notation.
 
The bias of a metamodeling estimator is known to be influenced by both the function approximation method and the design-point set used. In particular, the fill distance of a given design-point set---$\max_{\theta \in \Theta} \min_{i \in [N]^+} \|\theta_i - \theta \|$, where $N$ denotes the number of design points---plays an essential role in determining statistical properties of many popular metamodel-based estimators \cite{rosenbaum2017multilevel, wendland:2005, wynne2021convergence}. Intuitively, if the prediction point $\theta$ is considerably distant from the design points, the available observations may not provide sufficient information for the metamodel to yield an accurate estimate. 
Therefore, 
to control the magnitude of the bias component $\left\|\operatorname{Bias}(\widehat{\mathbb{V}})\right\|_{2}^{2}$, it is crucial to manage the fill distance achieved at each design level. To this end, we introduce the following assumption.

\begin{assumption}\label{asm:decay_ratio}
 For each prediction point $ \theta \in \Theta $ and  any $\ell \in \mathbb{N}$, consider the $\ell$th level variance function estimator given in  \eqref{eq:metamodel}. Suppose that each weight $ w_{i}^{\ell}(\theta) $ associated with the sample variance $\cV(\theta_i^{\ell},M_{\ell},\varpi_{\ell})$ at design point $ \theta_{i}^{\ell} $ is non-negative. Let $ I^{\ell}(\theta ; r) \coloneqq \left\{i:\|\theta-\theta_{i}^{\ell}\| \leq r, i \in [N_{\ell}]^{+}\right \} $ represent the index set of those design points in $\cT_\ell$ that are within distance $r$ from $\theta$, where recall that $N_{\ell}= |\cT_\ell|$ is the size of the design-point set  on level $\ell$. There exist constants $ s>1 $, $ \alpha>0 $, and sequences $ \left\{p_{\ell}\right\}_{\ell \in \mathbb{N}} $ and $ \left\{r_{\ell}\right\}_{\ell \in \mathbb{N}} $ such that $ p_{\ell} $ and $ r_{\ell} $ are $ \mathcal{O}\left(s^{-\alpha \ell}\right) $, and  
\begin{equation*}
\sum_{i \notin I^{\ell}\left(\theta ; r_{\ell}\right)} w_{i}^{\ell}(\theta) \leq p_{\ell} \quad \text { and } \quad\bigg|1-\sum_{i \in I^{\ell}\left(\theta ; r_{\ell}\right)} w_{i}^{\ell}(\theta)\bigg| \leq p_{\ell}  \  .    \nonumber
\end{equation*}
\end{assumption}
Assumption \ref{asm:decay_ratio} stipulates that the $\ell$th level variance estimator $\Vhatsub{\ell}{\theta}$ given in \eqref{eq:metamodel} is a weighted sum that assigns a higher weight to design points that are closer to the given prediction point $\theta$. 
We note that Assumption \ref{asm:decay_ratio} is satisfied by some widely adopted function estimation approaches, such as piecewise linear interpolation, $k$ nearest-neighbor approximation (kNN), and kernel smoothing, when used with space-filling experimental designs. 
Taking kNN as an example, it assigns a weight of $1/k$ to each of the $k$ nearest design points relative to the prediction point $\theta$, where $k$ is predetermined. Assuming that the fill distance $\Delta^{\ell}$ for the level $\ell$ design-point set $\cT_\ell$ is $\mathcal{O}(s^{-\alpha \ell})$,  we have $r_{\ell} = 2k\Delta^{\ell}$ and $p_{\ell}=0$,  which satisfies Assumption \ref{asm:decay_ratio}. Other approaches, such as piecewise linear interpolation and kernel smoothing, can also be demonstrated to satisfy Assumption~\ref{asm:decay_ratio}. For a more detailed discussion, we refer the interested reader to  Assumption 4 and Section EC.1.1 of \cite{rosenbaum2017multilevel}.
 
We are now in a position to provide a bound on  
$\left\|\Bias{\widehat{\mathbb{V}}_{\ell}}\right\|_{2}^{2} $ for any $\ell \in \mathbb{N}$,  as detailed in Proposition \ref{prop:bias_bound} below.

\begin{proposition}\label{prop:bias_bound}
	Under Assumptions \ref{asm:domain}, \ref{asm:finite_moment}, \ref{asm:LP_Y}, and \ref{asm:decay_ratio}, the integrated squared bias $ \left\|\Bias{\widehat{\mathbb{V}}_{\ell}}\right\|_{2}^{2} $ is $ \mathcal{O}\left(s^{-2 \alpha \ell}\right)$.
\end{proposition}
\begin{proof}
For any given $\theta \in \Theta$, we have
	\begin{align}
		\label{eq:biasconver}
		&\left|\Bias{\widehat{\mathbb{V}}_{\ell}(\theta)} \right| = \left| \E{ \Vhatsub{\ell}{\theta}  - \V{\theta} } \right| = \left| \sum_{i=1}^{N_{\ell}} w_{i}^{\ell}(\theta)\V{\theta_{i}^{\ell}}  - \V{\theta}  \right| \nonumber \\
		=& \left| \sum_{i \in I^{\ell}(\theta; r_{\ell})} w_{i}^{\ell}(\theta) \Big(\V{\theta_{i}^{\ell}}  - \V{\theta}  \Big) 
		+  \sum_{i \notin I^{\ell}(\theta; r_{\ell})} w_{i}^{\ell}(\theta) \Big(\V{\theta_{i}^{\ell}}- \V{\theta}  \Big) 
		-   \left( 1 - \sum_{i=1}^{N_{\ell}} w_{i}^{\ell}  \right) \V{\theta} \right| \nonumber \\
		\leq & \sum_{i \in I^{\ell}(\theta; r_{\ell})} w_{i}^{\ell}(\theta) \left| \V{\theta_{i}^{\ell}}  - \V{\theta}  \right| 
		+ \sum_{i \notin I^{\ell}(\theta; r_{\ell})} w_{i}^{\ell}(\theta) \left| \V{\theta_{i}^{\ell}}  - \V{\theta} \right| 
		+ \left| 1 - \sum_{i=1}^{N_{\ell}} w_{i}^{\ell}  \right|  \V{\theta}   \nonumber \\
		\leq & (1 + p_{\ell}) k_v r_{\ell} + p_{\ell} k_v \mbox{diam}(\Theta) + 2p_{\ell} \sup_{\theta \in \Theta} \V{\theta} \ , 
	\end{align}
	where the last step follows from Lemma \ref{lem:LP_V} and Assumption \ref{asm:decay_ratio}. It follows from Assumption \ref{asm:domain} and Lemma \ref{lem:finite_moment} that $\mbox{diam}(\Theta)<\infty$ and $\sup_{\theta \in \Theta} \V{\theta} < \infty$. The proof is complete by noting that the bound in \eqref{eq:biasconver} is uniform over all $\theta \in \Theta$, and all terms involved  are $ \mathcal{O}\left(s^{-\alpha \ell}\right) $ under Assumption \ref{asm:decay_ratio}.
\end{proof}
Taking into account Proposition~\ref{prop:bias_bound} and the fact that $\left\|\operatorname{Bias}(\widehat{\mathbb{V}})\right\|_{2}^{2} =\left\|\Bias{\widehat{\mathbb{V}}_{L}}\right\|_{2}^{2}$ (recall from Subsection \ref{subsec:MLMCM}),  we can effectively control the magnitude of the bias component $\left\|\operatorname{Bias}(\widehat{\mathbb{V}})\right\|_{2}^{2}$ in \eqref{eq:mise} by appropriately determining the finest level $L$.

\subsection{Variance}\label{subsec:variance_analysis}
This subsection provides bounds on the integrated variance of the $\ell$th level estimator \sloppy $\left\|\Var{\Vhatsub{\ell}{\theta, M_{\ell},   \varpi_{\ell}}}\right\|_1$
and the integrated variance per replication in estimating the $\ell$th level refinement $\left\|(M_{\ell}-1)\Var{\DeltaVhatsub{\ell}{\theta, M_{\ell}, \varpi_{\ell}}}\right\|_1$. The results then help bound  the variance component $\left\|\operatorname{Var}(\widehat{\mathbb{V}})\right\|_{1}$ given in \eqref{eq:mise}. Hereinafter, we will use shorthand notation: $\cV(\theta, M_{\ell})$ for the sample variance $\cV(\theta, M_{\ell}, \varpi_{\ell})$ and $\Vhatsub{\ell}{\theta, M_{\ell}}$  for the single-level estimator  $\Vhatsub{\ell}{\theta, M_{\ell}, \varpi_{\ell}}$.

\begin{proposition}\label{prop:var_finite}
Suppose that Assumptions \ref{asm:domain}, \ref{asm:finite_moment}, \ref{asm:LP_Y}, and \ref{asm:decay_ratio} hold, and $ M_{\ell} \geq 2 $ replications are simulated at each design point in $\cT_\ell$. Then, there exists a constant $ \bar{v} > 0 $ such that $ \left\|\Var{\Vhatsub{\ell}{\cdot, M_{\ell}}}\right\|_{1} \leq \bar{v} / (M_{\ell} - 1) $  for any $ \ell \in \mathbb{N}$. 
\end{proposition}

\begin{proof}
Given any design point $\theta^\prime\in \Theta$,  the variance of the sample variance, $\Var{ \cV(\theta^\prime,M_{\ell}) }$,  can be expressed as follows (as shown in (2.28) of \cite{mycek2019multilevel}):  
	\[ \Var{ \cV(\theta^\prime,M_{\ell}) } = \frac{\Msup{4}{\cY(\theta^\prime) }}{M_{\ell}} -  \frac{(M_{\ell}-3)\Varsup{2}{\cY(\theta^\prime) }}{M_{\ell}(M_{\ell}-1)}\ , \quad \mbox{ when  $M_{\ell} \geq 2$}. 	
	\]
Hence, it follows that
\begin{equation}\label{eq:variance_proof_bound_1}
	\Var{ \cV(\theta^\prime,M_{\ell}) } \leq \frac{\Msup{4}{\cY(\theta^\prime) }}{M_{\ell}} \leq \frac{\Msup{4}{\cY(\theta^\prime) }}{M_{\ell}-1}\ , \quad \mbox{ when } M_{\ell} \geq 3 ; 
\end{equation}
moreover, since the kurtosis $\Msup{4}{\cY(\theta^\prime) }/\Varsup{2}{\cY(\theta^\prime)} \geq 1$,
	\begin{equation}\label{eq:variance_proof_bound_2}
		 \Var{ \cV(\theta^\prime,M_{\ell}) } = \frac{1}{2}\cdot\Msup{4}{\cY(\theta^\prime) } + \frac{1}{2}\cdot \Varsup{2}{\cY(\theta^\prime)} \leq  \Msup{4}{\cY(\theta^\prime) }\ , \mbox{ when } M_{\ell} =2. 
	\end{equation}
It follows from  \eqref{eq:variance_proof_bound_1} and \eqref{eq:variance_proof_bound_2} that 
\begin{equation}\label{eq:variance_proof_bound_3}
	\Var{ \cV(\theta^\prime,M_{\ell}) } \leq \frac{\Msup{4}{\cY(\theta^\prime) }}{M_{\ell}-1} \ , \quad \mbox{when }  M_{\ell} \geq 2 .
\end{equation}
For any $ \ell \in \mathbb{N}$ and  $\theta \in \Theta$, 
\begin{align}
	\label{eq:estimator_var_bound}
\Var{\Vhatsub{\ell}{\theta, M_{\ell}}} 
  \leq\left(\sum_{i=1}^{N_{\ell}} w_{i}^{\ell}(\theta)\right)^{2}  \sup _{\theta^{\prime} \in \Theta} \Var{ \cV(\theta^{\prime},M_{\ell}) }  
  \leq \left(\sum_{i=1}^{N_{\ell}} w_{i}^{\ell}(\theta)\right)^{2}  \frac{c_{\cY}}{M_{\ell}-1} \ ,   
\end{align}	
where the first inequality follows from \eqref{eq:metamodel}, the second one follows from \eqref{eq:variance_proof_bound_3}, and recall that $c_{\cY} \coloneqq \sup_{\theta \in \Theta} \Msup{4}{\cY(\theta)} < \infty$ (from Lemma \ref{lem:finite_moment}). By Assumption \ref{asm:decay_ratio}, the square of the sum of weights in \eqref{eq:estimator_var_bound} is bounded above by $ \left(1+2 p_{\ell}\right)^{2} $, and since $ p_{\ell} \longrightarrow 0$ as $\ell \rightarrow \infty$, $\left(1+2 p_{\ell}\right)^{2} $ is bounded. By noting that the bound in \eqref{eq:estimator_var_bound} is uniform over $\theta \in \Theta$, we have 
$\left\|\Var{\Vhatsub{\ell}{\cdot, M_\ell} }\right\|_{1} \leq \bar{v} / (M_\ell-1)$, 
where $\bar{v} \coloneqq  \mbox{diam}(\Theta) \sup _{\ell \in \mathbb{N}}\left(1+2 p_{\ell}\right)^{2} c_{\cY} $ is finite by Assumption \ref{asm:domain} and Lemma  \ref{lem:finite_moment}.
\end{proof}

To study the variance of the refinement estimator $\DeltaVhatsub{\ell}{\theta, M_{\ell}}$, we rely on Lemma \ref{lem:var_bound} below that upper bounds the difference in the sample variances obtained at any two  points $\theta_1$ and $\theta_2$ in $\Theta$. The proof of Lemma~\ref{lem:var_bound} is deferred to Appendix~\ref{app:lem_var_bound}.
\begin{lemma}\label{lem:var_bound}
	Under Assumptions \ref{asm:domain}, \ref{asm:finite_moment}, and \ref{asm:LP_Y}, for any $ \theta_1, \theta_2 \in \Theta$ and $M_{\ell} \geq 2$,
	   \[  \Var{\cV(\theta_1,M_{\ell})-\cV(\theta_2,M_{\ell})} \leq \frac{8 \|\theta_1 - \theta_2\|_2^2}{M_{\ell}-1}  \left(\E{ \kappa_{y}^4 } + 2\Esup{2}{ \kappa_{y}^2} + \Esup{4}{\kappa_{y}} \right)^{\frac{1}{2}} c_{\cY}^\frac{1}{2} \ , \]
	  where  recall that  $c_{\cY}\coloneqq \sup_{\theta \in \Theta} \Msup{4}{\cY(\theta)} < \infty$  (from Lemma \ref{lem:finite_moment}).
\end{lemma}

 Let $ v_{\ell}(\theta) \coloneqq (M_{\ell} - 1) \Var{\DeltaVhatsub{\ell}{\theta, M_{\ell}}} $ denote the variance per replication in estimating the $\ell$th level refinement.

\begin{proposition}\label{prop:bound_vl}
 If Assumptions \ref{asm:domain}, \ref{asm:finite_moment}, \ref{asm:LP_Y}, and \ref{asm:decay_ratio} hold,  then $\left\|v_{\ell}\right\|_{1} $ is $ \mathcal{O}\left(s^{-2 \alpha \ell}\right) $.
\end{proposition}

\begin{proof}
At any $\theta \in \Theta$ and $M_\ell \geq 2$,  we have  
\begin{align}
	\label{eq:v_l_inequality}
	v_{\ell}(\theta) = & (M_{\ell} - 1) \Var{\Vhatsub{\ell}{\theta, M_{\ell}} - \Vhatsub{\ell-1}{\theta, M_{\ell}}} \nonumber\\
	= & (M_{\ell} - 1) \Var{\Vhatsub{\ell}{\theta, M_{\ell}} - \cV(\theta,M_{\ell}) + \cV(\theta,M_{\ell}) - \Vhatsub{\ell-1}{\theta, M_{\ell}} } \nonumber\\
    \leq & 2(M_{\ell} - 1)\Var{ \Vhatsub{\ell}{\theta, M_{\ell}} - \cV(\theta,M_{\ell}) } + 2(M_{\ell} - 1)\Var{ \Vhatsub{\ell-1}{\theta, M_{\ell}} - \cV(\theta,M_{\ell}) } \ , 
\end{align}
where  $\cV(\theta,M_{\ell}) $ denotes the sample variance that would have been obtained using the same random number stream and number of replications as those used to calculate $\Vhatsub{\ell}{\theta, M_{\ell}}$ at $\theta$.   

  We next show that the first term on the right-hand side of \eqref{eq:v_l_inequality}  is $ \mathcal{O}\left(s^{-2 \alpha \ell}\right) $.  Since
	\begin{align}\label{eq:single_var_diff}
	     \Vhatsub{\ell}{\theta, M_{\ell}} - \cV(\theta,M_{\ell}) 
		= & \sum_{i \in I^{\ell}(\theta; r_{\ell})}w_{i}^{\ell}(\theta) ( \cV(\theta_i,M_{\ell})  - \cV(\theta,M_{\ell})) 
		+  \sum_{i \notin I^{\ell}(\theta; r_{\ell})}w_{i}^{\ell}(\theta) ( \cV(\theta_i,M_{\ell})  - \cV(\theta,M_{\ell}) )  \nonumber\\
		& - \bigg(1 - \sum_{i=1}^{N_{\ell}}w_{i}^{\ell}(\theta)\bigg) \cV(\theta,M_{\ell}) \ ,
	\end{align}
we have
	\begin{align*}
		& (M_{\ell} - 1)\Var{ \Vhatsub{\ell}{\theta, M_{\ell}} - \cV(\theta,M_{\ell}) }
		\leq \  \underbrace{3(M_{\ell} - 1) \Var{ \sum_{i \in I^{\ell}(\theta; r_{\ell})}w_{i}^{\ell}(\theta) \left( \cV(\theta_i,M_{\ell})  - \cV(\theta,M_{\ell}) \right) }}_{(i)} \\
		& + \underbrace{ 3(M_{\ell} - 1) \Var{ \sum_{i \notin I^{\ell}(\theta; r_{\ell})}w_{i}^{\ell}(\theta) \big( \cV(\theta_i,M_{\ell})  - \cV(\theta,M_{\ell}) \big) }}_{(ii)} 
	 + \underbrace{ 3(M_{\ell} - 1) \Var{ (1 - \sum_{i=1}^{N_{\ell}}w_{i}^{\ell}(\theta)) \cV(\theta,M_{\ell}) } }_{(iii)} \ .
	\end{align*}
	
		Hence, it suffices to show that the terms $(i)$, $(ii)$, and $(iii)$  are uniformly bounded over  $\theta\in \Theta$ and are $ \mathcal{O}\left(s^{-2 \alpha \ell}\right) $. Specifically,
	\begin{itemize}
		\item 
For the term $(i)$, we have \sloppy
$(M_\ell-1)\Var{ \sum_{i \in I^{\ell}(\theta; r_{\ell})}w_{i}^{\ell}(\theta) \left( \cV(\theta_i,M_{\ell})  - \cV(\theta,M_{\ell}) \right) } \leq (M_\ell-1)\sum_{i,j \in I^{\ell}(\theta; r_{\ell})}w_{i}^{\ell}(\theta)w_{j}^{\ell}(\theta)\Var{  \cV(\theta_j,M_{\ell})  - \cV(\theta,M_{\ell})}$, by expanding the variance of the weighted sum and further bounding the covariance terms.
The right-hand side can  further  be bounded by
		$ \left(1+p_{\ell}\right)^2 r_{\ell}^{2} \left(\E{\kappa_{y}^4} + 2\Esup{2}{\kappa_{y}^2} + \Esup{4}{ \kappa_{y} }  \right)^{1/2} \allowbreak c_{\cY}^{1/2}$,
	  which follows from Assumption \ref{asm:decay_ratio} and Lemma \ref{lem:var_bound}.
		\item  Similar to the term $(i)$,  based on Assumption \ref{asm:decay_ratio} and Lemma \ref{lem:var_bound}, the term $(ii)$ can be bounded by 
		$p_{\ell}^2 \mbox{diam}(\Theta)^4 \allowbreak \left(\E{\kappa_{y}^4} + 2\Esup{2}{\kappa_{y}^2} + \Esup{4}{ \kappa_{y} }  \right)^{1/2}  c_{\cY}^{1/2} $.
		\item For the term $(iii)$, since $(1 - \sum_{i=1}^{N_{\ell}}w_{i}^{\ell}(\theta))^2 \leq 2\bigl(1 - \sum_{i \notin I^{\ell}(\theta; r_{\ell})}w_{i}^{\ell}(\theta) \bigr)^2 + 2\bigl(\sum_{i \in I^{\ell}(\theta; r_{\ell})}w_{i}^{\ell}(\theta)\bigr)^2$, it can be bounded by $ 4 p_{\ell}^{2} \sup _{\theta \in \Theta} \mbox{Var}\big(\cV(\theta,M_{\ell})\big) $ based on Assumption  \ref{asm:decay_ratio}.
	\end{itemize}
	Given that the bounds for the terms $(i)$, $(ii)$, and $(iii)$ are uniform over  $\theta \in \Theta$ and are $ \mathcal{O}\left(s^{-2 \alpha \ell}\right) $  since $ p_{\ell} $ is $ \mathcal{O}\left(s^{-\alpha \ell}\right) $ by Assumption \ref{asm:decay_ratio}, it follows from  \eqref{eq:single_var_diff} that  the first term on the right-hand side of \eqref{eq:v_l_inequality}  has a bound that is uniform over  $\theta \in \Theta$ and is $ \mathcal{O}\left(s^{-2 \alpha \ell}\right) $.
	Similar derivations yield the same conclusion for the second term on the right-hand side of \eqref{eq:v_l_inequality}. 
The proof is complete by noting that $ v_{\ell}(\theta) $ is uniformly bounded over  $\theta \in \Theta$ and is $ \mathcal{O}\left(s^{-2 \alpha \ell}\right) $.
\end{proof}

\subsection{Computational Complexity}
\label{sec:computation complexity}
This subsection demonstrates that MLMC metamodeling is more computationally efficient than SMC metamodeling in achieving a target MISE level. 

We begin our discussion by outlining several conditions crucial for analyzing the computational complexity of MLMC and SMC metamodeling. These conditions address the decay rates of both bias and variance, as well as the growth rate of the number of design points across successive levels.
 
\begin{description}
	\item [Condition 1 \label{con:bias}] For all $\ell\in \mathbb{N}$, the integrated squared bias of the $\ell $th level variance function estimator satisfies $ \left\|\Bias{\widehat{\mathbb{V}}_{\ell}}\right\|_{2}^{2}  \leq\left(b s^{-\alpha \ell}\right)^{2} $ for some $b \geq 0$.
	\item [Condition 2\label{con:var}] For all $\ell\in \mathbb{N}$ and all $M_\ell \geq 2$, the integrated variance of the $ \ell $th level variance function estimator satisfies $ \left\|\Var{\Vhatsub{\ell}{\cdot, M_{\ell}}}\right\|_{1} $ $ \leq \bar{v} / (M_{\ell}-1) $. 
	\item [Condition 3\label{con:vl}] For all $\ell\in \mathbb{N}$, the integrated variance per replication in estimating the $\ell$th level refinement  satisfies $ \left\|v_{\ell}\right\|_{1} \leq \tau^{2} s^{-2 \alpha \ell} $ for some $\tau^2 > 0$.
	\item [Condition 4\label{con:comp}] For all $\ell\in \mathbb{N}$, the number of design points $ N_{\ell} $ in the design-point set  on level $\ell$,  $ \mathcal{T}_{\ell}$, satisfies $N_{\ell} \leq c s^{\gamma \ell} $  for some $c , \gamma > 0$.
\end{description}
It is worth noting that Propositions~\ref{prop:bias_bound}, \ref{prop:var_finite}, and \ref{prop:bound_vl}  have shown that  Conditions 1 to 3 are satisfied. Condition 4 can be met by appropriately configuring the number of design points   $N_{\ell}$ in $ \mathcal{T}_{\ell}$.  
 
Define $\phi := \gamma / (2 \alpha)$,  where $\gamma$ denotes the growth rate of the number of design points in Condition 4, and $\alpha$ characterizes the diminishing rate of the bias and variance components  in Conditions 1 and 3.
Theorems~\ref{thm:computation} and \ref{thm:computation_SMC} quantify the computational complexity of MLMC and SMC metamodeling for variance function estimation, respectively.
 \begin{theorem}\label{thm:computation}	
	 Fix $\epsilon>0$. Under Conditions 1 to 4, for MLMC metamodeling, the computational budget required, in terms of the total number of simulation replications needed to achieve $\operatorname{MISE} < \epsilon^{2} $,  satisfies
	\begin{itemize}
		\item $ \mathcal{O}\left(\epsilon^{-2}\right) $ if $ \phi<1 $,
		\item $ \mathcal{O}\left(\left(\epsilon^{-1}\left(\log \epsilon^{-1}\right)\right)^{2}\right) $ if $ \phi=1 $,
		\item $ \mathcal{O}\left(\epsilon^{-2 \phi}\right) $ if $ \phi>1 $.
	\end{itemize}
\end{theorem}
 \begin{theorem}\label{thm:computation_SMC}
	 Fix $\epsilon>0$. Under Condition 2, for SMC metamodeling, the computational budget required, in terms of the total number of simulation replications needed to achieve $\operatorname{MISE} < \epsilon^{2} $, satisfies $ \mathcal{O}\left(\epsilon^{-2(1+\phi)}\right) $.
\end{theorem}

The proofs of Theorems \ref{thm:computation} and \ref{thm:computation_SMC} are provided in Appendices~\ref{app:thm_proof} and \ref{app:proof_smc}, respectively.
We highlight key insights into MLMC metamodeling derived from the proof of Theorem \ref{thm:computation} as follows. 
Adding a new design level (say, level $\ell$) and conducting the corresponding simulation runs requires an additional budget of $M_{\ell}N_{\ell} = \mathcal{O}(s^{(\gamma - 2\alpha)\ell/2})$, since  the number of replications $M_{\ell}$ is $\mathcal{O}(s^{-(2\alpha + \gamma)\ell/2})$ and the number of design points $N_{\ell}$ is $\mathcal{O}(s^{\gamma \ell})$ at level $\ell$. 
We can understand why the ratio $\phi= \gamma / (2 \alpha)$ determines the computational complexity as $\epsilon^2$ becomes small in the following way. If $\phi < 1$ (i.e., $\gamma - 2\alpha < 0$), then the budget required to add an extra level is negligible as $\ell$ increases, and the total computational budget to attain $\text{MISE}<\epsilon^{2} $ is $ \mathcal{O}\left(\epsilon^{-2}\right)$, which is comparable to the typical  computational complexity of estimating the expectation of a random variable using MC methods.
If $\phi = 1$ (i.e., $\gamma - 2\alpha = 0$), the budget required to add a new level is $\mathcal{O}(1)$, leading to a total computational budget of order $ \mathcal{O}\left(\left(\epsilon^{-1}\left(\log \epsilon^{-1}\right)\right)^{2}\right) $. This is slightly worse than the case when $\phi < 1$.  
If $\phi > 1$ (i.e., $\gamma - 2\alpha > 0$), the budget required to add a new level grows exponentially with the design level $\ell$. Consequently, the total computational budget to achieve $\text{MISE}<\epsilon^{2} $ is $ \mathcal{O}\left(\epsilon^{-2 \phi}\right) $, which is the most demanding among the three cases.

We note that the parameter $\phi= \gamma / (2 \alpha)$ is determined by the characteristics of the design-point generation scheme (e.g., Sobol' sequences) and the function approximation method (e.g., kernel smoothing). The primary flexibility available to the user is in adjusting the growth rate of the number of design points, $\gamma$. 
 In practice, selecting a high value for $\gamma$ results in fewer design levels, which may reduce the effectiveness of variance reduction. Conversely, selecting an excessively low value for $\gamma$ slows bias reduction, requiring a large number of design levels. Therefore, carefully selecting $\gamma$ is essential for optimizing the efficiency of MLMC metamodeling. However, the optimal choice depends on specific examples.

\section{Asymptotic Normality of the MLMC Metamodeling Estimator}\label{sec:clt}

In this section, we begin by examining the asymptotic properties of both the single-level and refinement estimators. We then establish the asymptotic normality of the MLMC metamodeling estimator, as given in \eqref{eq:mlmc_metamodeling_estimator}, for estimating the variance function.

An important property of MLMC estimators, investigated in various studies, is their asymptotic normality under suitable conditions. Alaya and Kebaier (2015) demonstrated the applicability of the Lindeberg-Feller central limit theorem (CLT) to  the MLMC method associated with the Euler discretization scheme \cite{ben2015central}. Collier et al.\ (2015) established the asymptotic normality of their proposed MLMC estimator for the expected value of a bounded linear or Lipschitz functional of the solution to stochastic differential equations \cite{collier2015continuation}. Dereich and Li (2016) explored the asymptotic normality of the MLMC estimator in the context of stochastic differential equations driven by L\'evy processes \cite{dereich2016multilevel}. 
Giorgi et al.\ (2017) investigated the asymptotic normality of both MLMC and weighted MLMC estimators, applying their theoretical findings to discretization schemes of diffusions and nested Monte Carlo methods \cite{giorgi2017limit}. Notably, most studies have typically assumed uniform integrability. However, Hoel and Krumscheid (2019) demonstrated that this condition is not necessary for CLTs, providing near-optimal weaker conditions under which CLTs can be established \cite{hoel2019central}.
While the asymptotic normality of MLMC estimators has been extensively studied, research on the asymptotic normality of MLMC metamodeling estimators remains relatively scarce. This work addresses this gap, offering valuable insights into inference and uncertainty quantification for MLMC metamodeling estimators.

 We first introduce a few key definitions to facilitate the analysis in this section. An MLMC metamodeling estimator that meets Conditions 1 to 4 specified in Section~\ref{sec:theory} and has an MISE less than a given value of $\epsilon^2$ (where $\epsilon > 0$) is defined as an \emph{$\epsilon^2$-estimator.} Define $L_{\epsilon}$ as the number of design levels required for constructing an $\epsilon^2$-estimator  and $M_{\ell, \epsilon}$ as  the number of simulation replications expended at each design point on level $\ell$ for any $\ell \in [L_{\epsilon}]$. Their respective forms can be given as follows:  
  \[L_\epsilon \coloneqq \ceil[\Big]{\frac{\log_s\left(\sqrt{2}b \epsilon^{-1}\right)}{\alpha} } \quad \mbox{and} \quad M_{\ell, \epsilon} \coloneqq \ceil[\Big]{ 2 \epsilon^{-2} \sqrt{ V_{\ell,\epsilon} / N_{\ell}} S_{L_{\epsilon}} } \ , \]
 where $V_{\ell,\epsilon} \coloneqq \left\| v_{\ell,\epsilon} \right\|_{1}$ and $S_{L_{\epsilon}} \coloneqq \sum_{\ell=0}^{L_{\epsilon}} \sqrt{ V_{\ell,\epsilon}N_{\ell}}$,  
 with $v_{\ell,\epsilon}$  defined similarly to $v_{\ell}$ (see its definition  above Proposition \ref{prop:bound_vl} in Subsection \ref{subsec:variance_analysis}), with $M_{\ell}$ replaced by $M_{\ell, \epsilon}$. 
 We see that $L_\epsilon$ and $M_{\ell, \epsilon}$ can be regarded as functions of $\epsilon$, with $M_{\ell, \epsilon} \rightarrow \infty$ for any $\ell \in [L_{\epsilon}]$ and $L_{\epsilon} \rightarrow \infty$,  as $\epsilon \rightarrow 0$.   
  In this section, we investigate the asymptotic normality of the MLMC metamodeling estimator and its different components as $\epsilon \rightarrow 0$. The established asymptotic normality holds pointwise for each prediction point $\theta \in \Theta$.

\subsection{Asymptotic Analysis of Single-level and Refinement Estimators}
\label{sec:CLT-single-refinement}
This subsection examines the asymptotic properties of the single-level estimator $\Vhatsub{\ell}{\theta, M_{\ell, \epsilon}, \varpi_{\ell}}$ given in  \eqref{eq:metamodel} and  the refinement estimator  $\DeltaVhatsub{\ell}{\theta, M_{\ell, \epsilon}, \varpi_{\ell}} $
for $\ell \in [L_{\epsilon}]$ given in   \eqref{eq:mlmc_metamodeling_estimator} as $\epsilon\rightarrow 0$.  

For ease of exposition, in this section, we abuse the notation slightly by writing the simulation output $\cY(\theta, \omega_m)$ as $\cY_{m}(\theta, \varpi_{\ell})$, which highlights that the simulation outputs are generated using the random number stream $\varpi_{\ell}$. That is, $\cY_{m}(\theta, \varpi_{\ell})$ denotes the $m$th random output obtained at $\theta$ using the random number stream $\varpi_{\ell}$, for  $m \in [M_{\ell, \epsilon}]^{+}$.  Define $\cZ_m(\theta, \varpi_{\ell}) \coloneqq \cY_m(\theta, \varpi_{\ell}) - \E{\cY_m(\theta, \varpi_{\ell})}$ as the centralized version of $\cY_{m}(\theta, \varpi_{\ell}) $, and let $\overline{Z^2}(\theta_{i}^{\ell}, \varpi_{\ell}) \coloneqq \sum_{m=1}^{M_{\ell, \epsilon}} \cZ_{m}^2(\theta_{i}^{\ell}, \varpi_{\ell}) /  M_{\ell, \epsilon}$. We note that the quantities defined from centralized outputs will facilitate the analysis of the asymptotic normality of the single-level and refinement estimators in this subsection, as well as the metamodeling estimator in the following subsection. 
  
 We first  establish  the asymptotic normality of $\Vhatsub{\ell}{\theta, M_{\ell, \epsilon}, \varpi_{\ell}} $ as $\epsilon\rightarrow 0$, with the proof provided in Appendix \ref{proof:clt_single_level}. 

\begin{proposition}\label{prop:clt_single_level}
Suppose that $\Vhat{\theta}$ is an $\epsilon^2$-estimator.   For all $\ell \in [L_{\epsilon}]$, it holds that the corresponding $\ell$th level estimator $\Vhatsub{\ell}{\theta, M_{\ell, \epsilon}, \varpi_{\ell}}$  is asymptotically normal:
	 \[\sqrt{M_{\ell,\epsilon}}\left(\Vhatsub{\ell}{\theta, M_{\ell, \epsilon}, \varpi_{\ell}} - \E{\Vhatsub{\ell}{\theta, M_{\ell, \epsilon}, \varpi_{\ell}}}\right) \Longrightarrow \mathcal{N}\left(0, \Var{\sum_{i=1}^{N_{\ell}}w_{i}^{\ell}(\theta) \cY_1^2(\theta_i^{\ell}, \varpi_{\ell})} \right) \text{ as } \epsilon \rightarrow 0.\]
\end{proposition}
Building upon Proposition \ref{prop:clt_single_level}, we  establish the asymptotic normality of the refinement estimator $\DeltaVhatsub{\ell}{\theta, M_{\ell, \epsilon}, \varpi_{\ell}}$ as $\epsilon\rightarrow 0$ in  Proposition \ref{prop:clt_refine} below. The proof is provided in Appendix \ref{proof:clt_refine}.

\begin{proposition}\label{prop:clt_refine}
Suppose that $\Vhat{\theta}$ is an $\epsilon^2$-estimator.
For all $\ell \in [L_{\epsilon}]$, it holds that the corresponding $\ell$th level refinement estimator $\DeltaVhatsub{\ell}{\theta, M_{\ell, \epsilon}, \varpi_{\ell}}$   is asymptotically normal:
	{	
	\begin{align*}
 	&\sqrt{M_{\ell, \epsilon}} \left( \DeltaVhatsub{\ell}{\theta, M_{\ell,\epsilon}, \varpi_{\ell}}  - \E{\DeltaVhatsub{\ell}{\theta, M_{\ell,\epsilon}, \varpi_{\ell}} } \right) \\
 	&\qquad \Longrightarrow \mathcal{N}\left(0, \Var{\sum_{i=1}^{N_{\ell}}w_{i}^{\ell}(\theta) \cY_1^2(\theta_i^{\ell}, \varpi_{\ell}) - \sum_{i=1}^{N_{\ell-1}}w_{i}^{\ell-1}(\theta) \cY_1^2(\theta_i^{\ell-1}, \varpi_{\ell}) } \right), \quad  \mbox{ as } \epsilon \rightarrow 0,   
\end{align*}
}
where $N_{-1}:= 0$. 
\end{proposition}

While the asymptotic normality of the single-level and refinement estimators can be derived using the classical central limit theorem (CLT), the analysis of the MLMC metamodeling estimator 
$\Vhat{\theta}$ requires the Lindeberg-Feller CLT and more extensive effort, as detailed in the next subsection.

\subsection{Asymptotic Normality of the MLMC Metamodeling Estimator}
\label{sec:CLT-MLMC-metamodel}
In this subsection, we  establish the asymptotic normality of the MLMC metamodeling estimator $\Vhat{\theta}$ in \eqref{eq:mlmc_metamodeling_estimator}.

 To facilitate our analysis, we first introduce the MLMC metamodeling estimator built on the centralized observations $\cZ_m(\theta, \varpi)$'s  (recall the definition from the beginning of Subsection~\ref{sec:CLT-single-refinement}).  Specifically, for any $\ell \in [L_{\epsilon}]$, define the corresponding single-level estimators utilizing these observations   as $\zhatsub{\ell}{\ell}  \coloneqq \sum_{i=1}^{N_{\ell}}w_{i}^{\ell}(\theta) \overline{Z^2}(\theta_i^{\ell}, \varpi_{\ell})$,  and the $\ell$th level refinement estimator  as $\deltaZhatsub{\ell}{\ell} \coloneqq \zhatsub{\ell}{\ell} - \zhatsub{\ell-1}{\ell}$, where $\zhatsub{-1}{\ell} := 0$. The MLMC metamodeling estimator built on centralized observations  is then given by
 \[\zhat \coloneqq \sum_{\ell=0}^{L_{\epsilon}}\deltaZhatsub{\ell}{\ell}.\]
 
Thanks to the following result, which indicates that the asymptotic normality of $\zhat$ implies that of the original MLMC metamodeling estimator $\Vhat{\theta}$, it suffices to study $\zhat$ instead.
 \begin{proposition}\label{prop:clt_z_v}
	If $\left( \zhat - \E{\zhat} \right) / \sqrt{\Var{\zhat}}  \Longrightarrow \mathcal{N}(0, 1)$ as $\epsilon \rightarrow 0$, then
	\[ \left(\Vhat{\theta} - \E{\Vhat{\theta}}\right) / \sqrt{\Var{\zhat}}  \Longrightarrow  \mathcal{N}(0, 1)  \text{ as } \epsilon \rightarrow 0. \]
\end{proposition}
The proof of Proposition \ref{prop:clt_z_v} is provided in Appendix \ref{proof:clt_z_v}.   Notice that  the $\ell$th level refinement estimator $\deltaZhatsub{\ell}{\ell}$ can be rewritten as follows: 
\[ \deltaZhatsub{\ell}{\ell} =   M_{\ell, \epsilon}^{-1} \sum_{m=1}^{M_{\ell, \epsilon}}\deltaZhatsubshort{(m)}{\ell}{\ell} , \]
where $ \deltaZhatsubshort{(m)}{\ell}{\ell} = \sum_{i=1}^{N_{\ell}}w_{i}^{\ell}(\theta) (\cZ_{m}(\theta_i^{\ell}, \varpi_{\ell}))^2 - \sum_{i=1}^{N_{\ell-1}}w_{i}^{\ell-1}(\theta) (\cZ_{m}(\theta_i^{\ell-1}, \varpi_{\ell}))^2$ are independent random variables for any $\ell \in [L_{\epsilon}]$ and $m \in [M_{\ell, \epsilon}]^{+}$. As a result, the Lindeberg-Feller CLT can be conveniently applied to   investigate the asymptotic normality of $\zhat$. Define $M_{\epsilon} := \sum_{\ell=0}^{L_{\epsilon}}M_{\ell, \epsilon}$  and 
a sequence of random variables $\{Z_{\epsilon, n}, n \in [M_{\epsilon}]^{+}\}$ as follows: 
\begin{equation}\label{eq:triangle}
	Z_{\epsilon, n} \coloneqq \left\{
	\begin{array}{l}
	\frac{\deltaZhatsubshort{(n)}{0}{0}  -\E{\deltaZhatsubshort{(n)}{0}{0}}}{M_{0,\epsilon}\sqrt{\Var{\zhat}}}, \quad 1 \leq n \leq M_{0,\epsilon} \ , \\
	\frac{\deltaZhatsubshort{(n-M_{0,\epsilon})}{1}{1}  -\E{\deltaZhatsubshort{(n-M_{0,\epsilon})}{1}{1}}}{M_{1,\epsilon}\sqrt{\Var{\zhat}}}, \quad M_{0,\epsilon} < n \leq M_{0,\epsilon} + M_{1,\epsilon}\ ,\\
	\qquad \qquad \vdots \\
	\frac{\deltaZhatsubshort{(n- \sum_{\ell=0}^{L_{\epsilon}-1}M_{\ell,\epsilon})}{L_{\epsilon}}{L_{\epsilon}}  -\E{\deltaZhatsubshort{(n- \sum_{\ell=0}^{L_{\epsilon}-1}M_{\ell,\epsilon})}{L_{\epsilon}}{L_{\epsilon}}}}{M_{L_{\epsilon},\epsilon}\sqrt{\Var{\zhat}}}, \quad  \sum_{\ell=0}^{L_{\epsilon}-1}M_{\ell,\epsilon} <  n \leq M_{\epsilon} \ .
	\end{array}\right.
\end{equation}
Based on \eqref{eq:triangle}, we can  scale and center $\zhat$ and obtain its equivalent form as follows:
 \begin{equation}\label{eq:norm_var_estimator}
	\frac{\zhat - \E{\zhatsub{L_{\epsilon}}{L_{\epsilon}}{}}}{\sqrt{\Var{\zhat}}} =	\frac{\zhat - \E{\zhat}}{\sqrt{\Var{\zhat}}} = \sum_{n=1}^{M_{\epsilon}} Z_{\epsilon, n} \ , 
	\end{equation}  
where the first equality follows from the definition of $\deltaZhatsub{\ell}{\ell}$  and  the fact that  $\E{ \widehat{Z}_{\ell}\left( \theta, M_{\ell}, \varpi_{\ell} \right) - \widehat{Z}_{\ell}\left( \theta, M_{\ell+1}, \varpi_{\ell+1} \right) }=0$ for  $\ell \in [L_{\epsilon}-1]$. 
	 It is evident from \eqref{eq:norm_var_estimator} that   $\E{\sum_{n=1}^{M_{\epsilon}} Z_{\epsilon, n}}=0$ and $\Var{\sum_{n=1}^{M_{\epsilon}} Z_{\epsilon, n}}=1$. We  are now in a position to  analyze the right-hand side of \eqref{eq:norm_var_estimator} via the Lindeberg-Feller CLT (Page 148, \cite{durrett2019probability}), with the notation adapted to our setting.
 
\begin{theorem}[Lindeberg-Feller CLT]\label{thm:linderberg_clt}
Let $M_{\epsilon} \rightarrow \infty$ as $\epsilon \rightarrow 0$ and $Z_{\epsilon, n}$ be independent random variables with $\E{Z_{\epsilon, n}} = 0$ for $n \in [M_{\epsilon}]^+$ and $\sum_{n=1}^{M_{\epsilon}} \E{Z_{\epsilon, n}^2} = 1$ as define in \eqref{eq:triangle}.    
Suppose for all $\nu > 0$,
\begin{equation}\label{eq:linder_origin}
	\lim_{\epsilon \rightarrow 0} \sum_{n=1}^{M_{\epsilon}} \E{|Z_{\epsilon, n}|^2 \bfone{|Z_{\epsilon, n}|> \nu}} = 0.
\end{equation}
Then  $\sum_{n=1}^{M_{\epsilon}} Z_{\epsilon, n} \Longrightarrow \cN(0, 1)$ as $\epsilon \rightarrow 0$.
\end{theorem}

Theorem \ref{thm:linderberg_clt} indicates that $\sum_{n=1}^{M_{\epsilon}} Z_{\epsilon, n}$ is asymptotically normal if the Lindeberg's condition in \eqref{eq:linder_origin} is fulfilled. By \eqref{eq:norm_var_estimator}, the same conclusion holds for $\zhat$.  As verifying \eqref{eq:linder_origin} can be challenging, we  reformulate it from the perspective of $\zhat$, as shown in Proposition \ref{thm:clt} below, to facilitate our subsequent analysis. The corresponding proof is provided in Appendix \ref{proof:clt_corollary}.

\begin{proposition}\label{thm:clt}
	 Suppose that $\Var{\zhat} > 0$ for any MISE target level $\epsilon^2 > 0$. The Lindeberg's condition in \eqref{eq:linder_origin} holds if, for any $\nu > 0$, the following condition is satisfied:
	 	{\small
 \begin{align}\label{eq:lindeberg}
	\lim_{\epsilon \rightarrow 0}\sum_{\ell=0}^{L_{\epsilon}} \frac{V_{\ell,\epsilon}(\theta)}{\Var{\zhat}M_{\ell,\epsilon}} &\E{\frac{\left|\deltaZhatsubshort{(1)}{\ell}{\ell} - \E{\deltaZhatsubshort{(1)}{\ell}{\ell}} \right|^2}{V_{\ell,\epsilon}(\theta)} \times \right. \nonumber \\ &\left.  \bfone{ \frac{\left|\deltaZhatsubshort{(1)}{\ell}{\ell} - \E{\deltaZhatsubshort{(1)}{\ell}{\ell}} \right|^2}{V_{\ell,\epsilon}(\theta)} > \frac{\Var{\zhat}M_{\ell,\epsilon}^2}{V_{\ell,\epsilon}(\theta)} \nu } }=0 .
\end{align} 	
}
\end{proposition}

 We also introduce a set of sufficient conditions to ensure that the condition in \eqref{eq:lindeberg} is fulfilled.
 Specifically, we consider the following two cases in terms of $S_{L_{\epsilon}}$ (recall its definition from the beginning of Section~\ref{sec:clt}): $\lim_{\epsilon \rightarrow 0} S_{L_{\epsilon}} < \infty$ and   $\lim_{\epsilon \rightarrow 0} S_{L_{\epsilon}} = \infty$. Intuitively, establishing sufficient conditions for the condition in \eqref{eq:lindeberg} relies on analyzing convergence rate of $\Var{\zhat}$. 
This analysis is facilitated by some knowledge of $S_{L_{\epsilon}}$, which is crucial for determining the magnitude of the number of replications $M_{\ell,\epsilon}$ for $\ell \in [L_{\epsilon}]$, as defined at the beginning of Section~\ref{sec:clt}.

\paragraph{\textbf{Case $\lim_{\epsilon \rightarrow 0} S_{L_{\epsilon}} <\infty$}}
The convergence rate of $\Var{\zhat}$ can be lower bounded by $\epsilon^2$ in this case. 
Hence, the condition in \eqref{eq:lindeberg} can be verified directly, as indicated by Proposition~\ref{thm:linde_finite} below.  Its proof is provided in Appendix \ref{proof:linde_finite}.

\begin{proposition}\label{thm:linde_finite}
If $\lim_{\epsilon \rightarrow 0} S_{L_{\epsilon}} < \infty$, the condition in \eqref{eq:lindeberg} holds. 
\end{proposition}

\paragraph{\textbf{Case $\lim_{\epsilon \rightarrow 0} S_{L_{\epsilon}} = \infty$}}
 Additional assumptions  are required  to verify the condition in \eqref{eq:lindeberg}.

\begin{assumption}\label{asm:basic_convergence}
$\lim_{\epsilon \rightarrow 0} \epsilon^{-2}\Var{\zhat}  > 0$.
\end{assumption}

\begin{assumption}\label{asm:Vl_Sl_order}
	$\lim_{\epsilon \rightarrow 0} S_{L_{\epsilon}} \cdot \epsilon^{\frac{\gamma}{2\alpha} - 2} > 0$. 
\end{assumption}

Under either Assumption \ref{asm:basic_convergence} or \ref{asm:Vl_Sl_order}, we can verify that the condition specified in \eqref{eq:lindeberg} is satisfied, as stated in Proposition \ref{thm:linde_infinite} below. The corresponding proof is given in Appendix \ref{proof:linde_infinite}.

\begin{proposition}\label{thm:linde_infinite}
	Suppose that $\lim_{\epsilon \rightarrow 0} S_{L_{\epsilon}} = \infty$. Under Assumption \ref{asm:basic_convergence} or \ref{asm:Vl_Sl_order}, the condition in \eqref{eq:lindeberg} holds. 
\end{proposition}

 We have several remarks. In Theorem \ref{thm:computation}, we observe that the relationship between $\gamma$ and $2\alpha$ impacts the computational efficiency of MLMC metamodeling. Recall  that $\gamma$ represents the growth rate of the number of design points with design levels, and $\alpha$ is associated with the diminishing rate of the integrated bias and variance components. Specifically, the condition $\gamma \leq 2\alpha$ results in a more efficient MLMC metamodeling estimator than $\gamma > 2\alpha$. Moreover, if $\gamma \leq 2\alpha$, Assumption \ref{asm:Vl_Sl_order} is automatically satisfied, which ensures  the condition in \eqref{eq:lindeberg}, and consequently, the Lindeberg's condition in \eqref{eq:linder_origin} are met, regardless of the value of $S_{L_{\epsilon}}$.

Finally, it is worth noting that as stated by Theorem \ref{thm:linderberg_clt}, satisfying the Lindeberg's condition in \eqref{eq:linder_origin} ensures the asymptotic normality of the MLMC metamodeling estimator $\zhat$. According to Proposition \ref{prop:clt_z_v}, the MLMC metamodeling estimator $\Vhat{\theta}$ is also asymptotically normal  under these same sufficient assumptions. In particular, a computationally efficient MLMC metamodeling estimator with $\gamma \leq 2\alpha$ is also asymptotically normal as $\epsilon\rightarrow 0$.

The asymptotic normality of the MLMC metamodeling estimator can be valuable for uncertainty quantification. Constructing a confidence interval for $\V{\theta}$ based on the proven CLT requires knowledge of the asymptotic variance $ \Var{\widehat{Z}(\theta)}$, which can be challenging to estimate. Classical nonparametric techniques for asymptotic variance and interval estimation in the simulation output analysis literature, such as batching \cite{ChenKim:2016, SchmeiserBatchSize:1982, seila1982batching}, can be leveraged alongside MLMC metamodeling to address this challenge. To conserve space, we do not provide a detailed discussion.

\section{Procedures for Multilevel Monte Carlo Metamodeling}
\label{sec: mlmc_procedure}
This section introduces two computational procedures for variance function estimation via MLMC metamodeling. Subsection \ref{subseq:target_procedure} details the first procedure, which focuses on attaining a target MISE level, while Subsection \ref{subseq:fixed_budget} elaborates on the second, designed to utilize a fixed computational budget.

\subsection{Procedure for Achieving a Target Accuracy Level}\label{subseq:target_procedure}
This subsection presents an MLMC procedure designed to achieve a target MISE level of $\epsilon^2$ for variance function estimation. 
The basic idea is to ensure that both the integrated squared bias and the integrated variance, which are the two components of the MISE (refer to Equation \eqref{eq:mise}), are each less than $\epsilon^2 / 2$. To achieve this, the design levels should be increased to reduce the bias, and additional simulation replications should be added to decrease the variance.
Our procedure builds upon the foundational work of MLMC metamodeling for mean function estimation \cite{rosenbaum2017multilevel} and MLMC pointwise variance estimation \cite{mycek2019multilevel}. Notably, we emphasize our contribution in addressing the challenges associated with high-order moment estimation in variance function estimation.

To begin with, the parameters of the target-accuracy MLMC procedure include the following: the target MISE level  $ \epsilon^{2} $, an initial number of replications $ M^{0} $ to be applied at each design point at each added design level, a prediction-point set $\cP$, a sequence of design-point sets $ \left\{\mathcal{T}_{\ell}\right\}_{\ell \in \mathbb{N}} $ with the corresponding size $ N_{\ell} $ increasing by a factor of roughly $ s^{\gamma} $, and the parameter $ \alpha > 0 $ for which $\left\|\Bias{\widehat{\mathbb{V}}_{\ell}}\right\|_{2}^{2}$ and $\left\|v_{\ell}\right\|_{1}$ are $ \mathcal{O}\left(s^{-2 \alpha \ell}\right) $ for $\forall \ell \in \mathbb{N}$ (refer to Conditions 1 and 3 in Subsection \ref{sec:computation complexity}).

The proposed MLMC procedure progresses iteratively. In each iteration, we add one additional design level and adjust the required number of replications at each existing level based on the updated design-point set structure, ensuring that the variance component $\left\|\Var{\widehat{\mathbb{V}}}\right\|_{1} < \epsilon^2/2$ is met. As design levels are gradually added, the bias component diminishes. The procedure terminates when the bias term $\left\|\Bias{\widehat{\mathbb{V}}}\right\|_{2}^{2}$ drops below $\epsilon^2/2$, which naturally determines the number of design levels adopted. We now discuss the stopping criterion and the approach for determining the required number of replications in detail.

The stopping criterion leverages the bound on bias (Condition 1 in Subsection \ref{sec:computation complexity}) to assess whether the integrated squared bias of the MLMC metamodeling estimator  with $ \ell $ levels, i.e., $\left\|\Bias{\widehat{\mathbb{V}}_{\ell}}\right\|_{2}^{2}$, meets the target. Assuming that the bound $\left\|\Bias{\widehat{\mathbb{V}}_{\ell}}\right\|_{2}^{2} \leq b^2s^{-2\alpha \ell}$ is tight, we have
\[ \left\|\E{\widehat{\Delta \mathbb{V}}_{\ell}}\right\|_{2}^{2}=\left\|\E{\widehat{\mathbb{V}}_{\ell}-\widehat{\mathbb{V}}_{\ell-1}}\right\|_{2}^{2} \geq b^{2} s^{-2 \alpha \ell}\left(s^{2 \alpha}-1\right)=\left(s^{2 \alpha}-1\right)\left\|\Bias{\widehat{\mathbb{V}}_{\ell}}\right\|_{2}^{2}. \]
To meet the criterion $ \left\|\Bias{\widehat{\mathbb{V}}_{\ell}}\right\|_{2}^{2} < \epsilon^{2} / 2 $, a sufficient condition is $\left\|\E{\widehat{\Delta \mathbb{V}}_{\ell}}\right\|_{2}^{2}  < \left(s^{2 \alpha}-1\right) \epsilon^{2} / 2$. This can be implemented as $\left\|\widehat{\Delta \mathbb{V}}_{\ell}\right\|_{2}^{2} < \left(s^{2 \alpha}-1\right) \epsilon^{2} / 2$. However, this approach can be conservative since $ \E{\left\|\widehat{\Delta \mathbb{V}}_{\ell}\right\|_{2}^{2}} \geq \left\|\E{\widehat{\Delta \mathbb{V}}_{\ell}}\right\|_{2}^{2} $. Alternatively, the implementation can include $ \left\|\E{\widehat{\Delta \mathbb{V}}_{\ell}}\right\|_{2}^{2} \approx s^{-2 \alpha}\left\|\widehat{\Delta \mathbb{V}}_{\ell-1}\right\|_{2}^{2} $.  Hence, we propose the following stopping criterion:
\begin{equation}\label{eq:stop}
	\max \left\{\left\|\widehat{\Delta \mathbb{V}}_{\ell}\right\|_{2}^{2}, s^{-2 \alpha}\left\|\widehat{\Delta \mathbb{V}}_{\ell-1}\right\|_{2}^{2}\right\} < \left(s^{2 \alpha}-1\right)  \epsilon^{2}/2 \ .
\end{equation}
Therefore,  the finest design level $L$ is  determined by identifying the smallest value of $\ell$ that meets the stopping criterion given in \eqref{eq:stop}.

We next examine how to make the variance component less than $\epsilon^2/2$. Recall that $ v_{\ell}(\theta) = (M_{\ell} - 1) \Var{\DeltaVhatsub{\ell}{\theta, M_{\ell}}}$ denotes the variance per replication in estimating the $\ell$th level refinement for $\ell \in [L]$. The variance component can be expressed as
\begin{align*}
    \left\|\Var{\widehat{\mathbb{V}}}\right\|_{1} &= \left\|\sum_{\ell=0}^{L} \Var{ \DeltaVhatsub{\ell}{\cdot, M_{\ell}}  } \right\|_{1} = \left\|\sum_{\ell=0}^{L} \frac{v_{\ell}}{M_{\ell}-1} \right\|_{1} \leq \sum_{\ell=0}^{L} \frac{\left\| v_{\ell} \right\|_{1}}{M_{\ell}-1} \ .
\end{align*}

Hence, a sufficient condition to ensure that the variance component is below the desired level is: $\sum_{\ell=0}^{L} \left\| v_{\ell} \right\|_{1}/(M_{\ell}-1) < \epsilon^2/2$. To determine the optimal number of replications  at each level $\ell \in [L]$, denoted as $M_{\ell}^{\ast}$, with the goal of minimizing the total number of replications, we formulate the following optimization problem:
\begin{align*}
    \mbox{min} &\sum_{\ell=0}^{L}M_{\ell}N_{\ell}   \\
    \text{s.t. } &\sum_{\ell=0}^{L} \frac{\left\| v_{\ell} \right\|_{1}}{M_{\ell}-1}    < \epsilon^2/2   \\
    & M_{\ell}   \geq 1, \quad \forall \ell \in [L]   \ .
\end{align*}
The optimal solution to this problem has a closed-form expression given by
\begin{equation}\label{eq:M_star}
	M_{\ell}^{\ast} = \ceil[\Big]{ 2 \epsilon^{-2} \sqrt{ \left\| v_{\ell} \right\|_{1} / N_{\ell}} \sum_{\ell^{\prime}=0}^{L} \sqrt{ \left\| v_{\ell^{\prime}} \right\|_{1}N_{\ell^{\prime}}} +1}, \quad \forall \ell \in [L].
\end{equation} 

Controlling the bias and variance components, as discussed above, is crucial for the MLMC metamodeling procedure to achieve the target accuracy level. Based on these discussions, we propose the MLMC metamodeling procedure outlined in Algorithm \ref{alg:mlmc}. We highlight a few key steps below. Specifically, we use a nested sequence of  design-point sets  in Steps 8 and 9. This approach allows for the reuse of simulation outputs, resulting in cost savings. Additionally, reusing simulation outputs establishes correlations between the estimators $\Vhatsub{\ell}{\theta, M_{\ell}, \varpi_{\ell}}$ and $\Vhatsub{\ell-1}{\theta, M_{\ell}, \varpi_{\ell}}$ for $\ell \in [L]$, which efficiently reduces the variance of the refinement estimators.
In Step 14,  we estimate  the variance of the  0th level metamodel estimator $\Vhatsub{0}{\theta, M_{0}, \varpi_{0}}$ and that of the higher-level refinement estimators $\DeltaVhatsub{\ell}{\theta, M_{\ell}, \varpi_{\ell}}$ for $\theta \in \Theta$ and $\ell \in [L]$ via bootstrapping \cite{barton1993uniform, cheng1995bootstrap}, as detailed in Algorithm~\ref{alg:bootstrap} provided in Appendix \ref{app:bootstrap}. Notice that the bootstrap estimator is biased when the number of replications $M_{\ell}$ is small, but it becomes approximately unbiased as $M_{\ell}$ becomes large (Page 271, \cite{efron1994introduction}).
Alternative methods, such as sectioning or jackknife \cite{asmussen2007stochastic}, can also be adopted.  In \cite{mycek2019multilevel}, a heuristic method was employed, assuming $\left\|v_{\ell}\right\|_{1} = s^{-2\alpha} \left\|v_{\ell-1}\right\|_{1}$. This method only requires estimating $\left\|v_{0}\right\|_{1}$, as $\left\|v_{\ell}\right\|_{1}$ can be estimated iteratively for $\ell \geq 1$. However, we recommend bootstrapping due to its superior effectiveness and robustness in implementation compared to these alternative methods. Steps 15 and 19 are to add additional replications if needed to ensure that the integrated variance of the estimator $\left\|\Var{\widehat{\mathbb{V}}}\right\|_{1}$ is less than $\epsilon^2/2$.

\begin{algorithm}
\caption{The MLMC metamodeling procedure for achieving a target accuracy level}\label{alg:mlmc}
\begin{algorithmic}[1]

\State \textbf{Input:} Parameters $\alpha, \gamma, \epsilon, s > 0$, the initial number of replications $M^{0}$ at each level, and the prediction-point set $\mathcal{P}$
\State \textbf{Output:} The MLMC metamodeling variance function estimator $\Vhat{\cdot}$

\State $L \gets 0$; \Comment{Initialize the design level index}
\State $\mathcal{T}_{-1} \gets \emptyset$, $N_{-1} \gets 0$, $d_{-1}^{2} \gets 0$;  \Comment{Initialize the design level index at ``$-1$'' to prevent illegal operations}
\State $d_{L}^{2} \gets 0$;  \Comment{Initialize the estimation of the integrated squared bias}

\While{$L < 2$ \textbf{or} $ \max \left\{d_{L}^{2}, s^{-2 \alpha} d_{L-1}^{2}\right\}>\left(s^{2 \alpha}-1\right) \epsilon^{2} / 2 $
}
\State  $N_{L} \gets s^{\gamma L}$;  \Comment{Set the number of design points on the current finest level}
\State  Generate an additional design-point set $\cA_{L}$ of size $(N_{L} - N_{L - 1})$;
\State  Construct the design-point set at level $L$: $\mathcal{T}_{L} \gets \cA_{L} \cup \mathcal{T}_{L-1}$;
\State  $M_{L} \gets M^{0}$;   \Comment{Initialize the number of replications}
\State  \algmultiline{For $\forall \theta^{L} \in \mathcal{T}_{L}$ and $\forall \theta^{L-1} \in \mathcal{T}_{L-1}$, simulate $M_{L}$ replications using the random number stream $\varpi_L$ and get simulation outputs $\{\cY(\theta^L, \omega_m), \theta^L \in \mathcal{T}_{L},m \in [M_{L}]^{+}\}$ and $\{\cY(\theta^{L-1}, \omega_m), \theta^{L-1} \in \mathcal{T}_{L-1}, m \in [M_{L}]^{+}\}$;}
\State  Build the metamodel-based estimators $\Vhatsub{L}{\theta, M_{L}, \varpi_{L}}$ and $\Vhatsub{L-1}{\theta, M_{L}, \varpi_{L}}$ according to \eqref{eq:metamodel};
\State  Calculate $\DeltaVhatsub{L}{\theta, M_{L}, \varpi_{L}} = \Vhatsub{L}{\theta, M_{L}, \varpi_{L}} - \Vhatsub{L-1}{\theta, M_{L}, \varpi_{L}}$ for $\forall \theta \in \mathcal{P}$;
\State \algmultiline{Estimate $ \Var{\DeltaVhatsub{L}{\theta, M_{L}, \varpi_{L}}} $ via Algorithm \ref{alg:bootstrap} in Appendix \ref{app:bootstrap} and estimate $\left\|v_{L}\right\|_{1}$ by $V_{L} \coloneqq \sum_{\theta \in \mathcal{P}} (M_{L}-1)\Var{\DeltaVhatsub{L}{\theta, M_{L}, \varpi_{L}}} / |\mathcal{P}|$;}
\State  \algmultiline{For $\ell \in [L]$, calculate $ \Delta_{\ell} = \ceil[\Big]{ 2 \epsilon^{-2} \sqrt{V_{\ell}/N_{\ell}} \sum_{\ell^{\prime}=0}^{L} \sqrt{V_{\ell^{\prime}}N_{\ell^{\prime}}} +1}  - M_{\ell}$. If $\Delta_{\ell} > 0$, $M_{\ell} \gets M_{\ell} + \Delta_{\ell}$ and simulate additional $\Delta_{\ell}$ number of replications, obtain the outputs, and update the statistics;}
\State  \algmultiline{For $\forall \theta \in \mathcal{P}$ and $\ell \in [L]$, calculate $ \left( \DeltaVhatsub{\ell}{\theta, M_{\ell}, \varpi_{\ell}}\right)^{2} $ and estimate $\left\|\widehat{\Delta \mathbb{V}}_{\ell}\right\|_2^2$ by $d_{\ell}^{2} \coloneqq \sum_{\theta \in \mathcal{P}} \left( \DeltaVhatsub{\ell}{\theta, M_{\ell}, \varpi_{\ell}}\right)^{2} / |\mathcal{P}|$;}
\State  $L \gets L + 1$; \Comment{Update the index of the finest level}	
\EndWhile
\State Calculate $ \Delta_{\ell} = \ceil[\Big]{ 2 \epsilon^{-2} \sqrt{\frac{V_{\ell}}{N_{\ell}}} \sum_{\ell^{\prime}=0}^{L} \sqrt{V_{\ell^{\prime}}N_{\ell^{\prime}}} +1}  - M_{\ell}$ for $\ell \in [L]$. If $\Delta_{\ell} > 0$, $M_{\ell} \gets M_{\ell} + \Delta_{\ell}$, simulate an additional $\Delta_{\ell}$ number of replications and update corresponding statistics;
\State Obtain the MLMC  metamodeling estimator $\Vhat{\theta} = \sum_{\ell=0}^{L}\DeltaVhatsub{\ell}{\theta, M_{\ell},\varpi_{\ell}}$ for $\forall \theta \in \mathcal{P}$.
\end{algorithmic}
\end{algorithm}

\subsection{Procedure for Expending a Fixed Simulation Budget}\label{subseq:fixed_budget}
In practice, one may aim to achieve the lowest possible MISE given a fixed simulation budget. This subsection presents a procedure to accomplish this goal.
It is important to address two crucial questions when using a fixed budget: (1) Should more design levels be added, or should additional replications be allocated to each design point on existing levels?
(2) If the latter, which existing design level should receive more replications?

To address the first question, a good approach is to strike a balance between the integrated squared bias $\left\|\Bias{\widehat{\mathbb{V}}}\right\|_{2}^{2}$ and the integrated variance $\left\|\Var{\widehat{\mathbb{V}}}\right\|_{1}$ of the estimator $\widehat{\mathbb{V}}$ to efficiently decrease the MISE. Given the current finest level $L$, recall that Algorithm \ref{alg:mlmc} requires
	\[
		\left(s^{2 \alpha}-1\right)^{-1} \max \left\{\left\|\widehat{\Delta \mathbb{V}}_{L}\right\|_{2}^{2}, s^{-2 \alpha}\left\|\widehat{\Delta \mathbb{V}}_{L-1}\right\|_{2}^{2}\right\} < \epsilon^2/2 \quad \text{and} \quad \sum_{\ell=0}^{L} \frac{\left\|v_{\ell}\right\|_{1}}{M_{\ell}-1} < \epsilon^2/2 \ .
	\]
To balance $\left\|\Bias{\widehat{\mathbb{V}}}\right\|_{2}^{2}$ and $\left\|\Var{\widehat{\mathbb{V}}}\right\|_{1}$, we propose adding one more design level if the following condition is met:
\begin{align*}
	\left(s^{2 \alpha}-1\right)^{-1} \max \left\{\left\|\widehat{\Delta \mathbb{V}}_{L}\right\|_{2}^{2}, s^{-2 \alpha}\left\|\widehat{\Delta \mathbb{V}}_{L-1}\right\|_{2}^{2}\right\} \geq \sum_{\ell=0}^{L}  \left\|v_{\ell}\right\|_{1}/(M_{\ell}-1) \ ;
\end{align*}
otherwise, additional replications will be allocated to a selected existing level.

To determine which existing design level receives additional replications, we aim to achieve the maximum variance reduction from adding a given number of replications. Let $A$ be the number of replications to be added at each design point on the chosen level. The optimal design level $\ell^{\ast}$ is selected by maximizing the variance reduction as follows:
\begin{equation}\label{eq:repadd_criterion}
    \ell^{\ast} \coloneqq  \argmax_{\ell \in [L]} \left(\left\|v_{\ell}\right\|_{1}/\left(M_{\ell}-1\right) - \left\|v_{\ell}\right\|_{1}/\left(M_{\ell}-1 + A\right) \right) / \left(N_{\ell}\cdot A\right)   \  ,
\end{equation}
where $\left\|v_{\ell}\right\|_{1}/(M_{\ell}-1) - \left\|v_{\ell}\right\|_{1}/(M_{\ell} -1+ A) $ represents the total variance reduction achieved by adding $A$ replications to each design point on level $\ell$, and $N_{\ell}\cdot A$ denotes the total number of additional replications to be added on level $\ell$. Therefore, the criterion in \eqref{eq:repadd_criterion} selects the design level that achieves the maximum variance reduction per design point per additional replication.

The MLMC metamodeling procedure for expending a fixed budget is detailed in Algorithm \ref{alg:mlmc_budget}. Below, we highlight a few key steps.
Steps 9 and 11 check the budget constraint. Specifically, Step 9 assesses whether the remaining budget is sufficient for adding the minimum number of replications required, and Step 11 examines if there is enough budget to add a new design level.  Step 12 determines whether to add a new design level or allocate additional replications to existing levels. Steps 16 through 20 and 23 through 27 choose the optimal design level to add additional replications. Similar to Algorithm \ref{alg:mlmc}, we employ Algorithm \ref{alg:bootstrap} in Appendix~\ref{app:bootstrap} to compute $V_{\ell}$, which represents the estimated integrated variance for $\ell \in [L]$. The calculations of $d_{\ell}^2$, $\Vhatsub{\ell}{\theta, M_{\ell}, \varpi_{\ell}}$, and $\DeltaVhatsub{\ell}{\theta, M_{\ell}, \varpi_{\ell}}$ remain consistent with those in Algorithm \ref{alg:mlmc}.

\begin{algorithm}
\caption{The MLMC metamodeling procedure for expending a fixed budget}\label{alg:mlmc_budget}
\begin{algorithmic}[1]
\State \textbf{Input:} Parameters $\alpha, \gamma, s, A > 0$, an initial number of replications $M^{0}$, the prediction-point set $\mathcal{P}$, and a given total budget $T$
\State \textbf{Output:} The MLMC metamodeling variance function estimator $\Vhat{\cdot}$

\State $L \gets 0$; \Comment{Initialize the design level index}
\State $N_0 \gets s^{\gamma}$; \Comment{Set \# of design points on the initial level}
\State Construct $0$th level metamodeling estimator $\Vhatsub{0}{\theta, M_{0}, \varpi_{0}}$;
\State Estimate $\left\|\widehat{\Delta \mathbb{V}}_{0}\right\|_2^2$ by $d_{0}^{2}$;
\State Estimate $\|v_{0}\|_1$ by $V_{0}$;
\State The cumulated total budget spent $t \gets 0;$  \Comment{Initialize the cumulative cost}
\While{$t + N_{0}A \leq T$ }\Comment{Continue until the budget is depleted}
\State $N_{L+1} \gets N_0s^{\gamma(L+1)}$;  \Comment{Set \# of design points on the next level}
  \If{$t + N_{L+1}M^{0} \leq T$ }
  \If{$\left(s^{2 \alpha}-1\right)^{-1} \max\left\{d_{L}^2, s^{-2\alpha}d_{L-1}^2\right\} \geq \sum_{\ell=0}^{L}\frac{V_{\ell}}{M_{\ell}-1}$} \Comment{If bias component dominates, add a new level}
  	\State  \algmultiline{Add a new level: $L \gets L+1$, build the metamodel-based estimators $\Vhatsub{L}{\theta, M_{L}, \varpi_{L}}$ and $\Vhatsub{L-1}{\theta, M_{L}, \varpi_{L}}$ according to \eqref{eq:metamodel} and calculate $\DeltaVhatsub{L}{\theta, M_{L}, \varpi_{L}}$;}
    \State  	   Calculate $d_{L}^2$ and $V_{L}$;
  \Else
   \Comment{If variance component dominates, add more replications}
 \State	  $\cL \gets \left\{\ell:  t + N_{\ell}A \leq T, \ell \in [L] \right\}$;
 \State  	 Choose level $\ell^{\ast} \coloneqq \argmax_{\ell \in \cL} V_{\ell}((M_{\ell}-1)(M_{\ell} - 1 + A)N_{\ell})^{-1}$;
 \State  	 Add $A$ replications at each design point on level $\ell^{\ast}$;
 \State  	  $M_{\ell^{\ast}} \gets M_{\ell^{\ast}} + A$;
 \State  	  Obtain simulation outputs and update corresponding statistics;
   \EndIf
  \Else
  \Comment{Add more replications}
  \State $\cL \gets \left\{\ell:  t + N_{\ell}A \leq T, \ell \in [L] \right\}$;
  \State Choose level $\ell^{\ast} \coloneqq \argmax_{\ell \in [L]} V_{\ell}((M_{\ell}-1)(M_{\ell} - 1 + A)N_{\ell})^{-1}$;
  \State Add $A$ replications at each design point on level $\ell^{\ast}$;
  \State  $M_{\ell^{\ast}} \gets M_{\ell^{\ast}} + A$;
  \State Obtain simulation outputs and update corresponding statistics;
  \EndIf
\EndWhile
\State Obtain the MLMC metamodeling estimator $\Vhat{\theta} =  \sum_{\ell=0}^{L}\DeltaVhatsub{\ell}{\theta, M_{\ell},\varpi_{\ell}}$ for $\forall \theta \in \mathcal{P}$.  
\end{algorithmic}
\end{algorithm}

\section{Numerical Experiments}\label{sec:exp}
\numberwithin{equation}{section}
This section presents numerical evaluations of the two MLMC metamodeling procedures proposed in Section \ref{sec: mlmc_procedure}.  Subsection \ref{sec:example} demonstrates the effectiveness of the procedure designed to achieve a target accuracy level.  Subsection \ref{sec:sobol case} applies the fixed-budget procedure to variance function estimation in the context of global sensitivity analysis.

\subsection{Numerical Evaluations of the Procedure for Achieving a Target Accuracy Level}\label{sec:example}

This subsection focuses on evaluating MLMC metamodeling  versus SMC metamodeling in achieving a target MISE level using two examples: the initial value problem and the Griewank function.

In each example, we vary the target MISE level $\epsilon^2$ and compare the computational efficiency of the two approaches. 
To implement MLMC metamodeling for a given value of $\epsilon^2$, we follow the procedure outlined in Algorithm~\ref{alg:mlmc},  which, upon termination, determines the computational cost required by MLMC metamodeling.
Specifically, we set the initial number of replications $M^0 = 4$ and tailor the parameters $\alpha$ and $\gamma$ to each specific example. A prediction-point set comprising 256 points from  a Sobol' sequence is used for estimating the bias and variance components achieved by  MLMC metamodeling. 
The sequence of design-point sets is generated by using Sobol' sequences with a random shift \cite{Lemieux2009}.

To implement SMC metamodeling and facilitate comparisons with MLMC metamodeling, we need to provide a suitable experimental design and estimate the computational cost required to achieve a given MISE level $\epsilon^2$. Following \cite{rosenbaum2017multilevel}, we construct an SMC metamodeling estimator using the design-point set from the finest level (denoted as level $L$),  $\cT_L$, formed by MLMC metamodeling, which consists of $2^{\gamma L}$ design points. This SMC metamodeling estimator has the same bias as the MLMC metamodeling estimator, which is less than $\epsilon^2/2$. To ensure that the integrated variance of the SMC metamodeling estimator is also less than $\epsilon^2/2$, we first allocate 1,000 initial replications to each design point to estimate the integrated variance of the estimator $V_{L}$ using Algorithm \ref{alg:bootstrap}. We then set the total number of replications $M_{L}^{*} = \lceil 2 \epsilon^{-2} V_{L} + 1 \rceil$ according to \eqref{eq:M_star}. Consequently, the computational cost required by SMC metamodeling to achieve the target MISE level $\epsilon^2$ is $M_{L}^{*} \cdot 2^{\gamma L}$.

For both MLMC and SMC metamodeling, we apply kernel smoothing with a Gaussian kernel as the function approximation method, selecting the bandwidth via leave-one-out cross-validation \cite{bowman1984alternative}. Kernel smoothing satisfies Assumption \ref{asm:decay_ratio}, and we refer the interested reader to Section EC.1.1 of \cite{rosenbaum2017multilevel} for further details.
In fact, we mention, without showing details, that Assumptions \ref{asm:domain} through \ref{asm:LP_Y} and Assumption \ref{asm:decay_ratio} can be verified to hold in both examples.
	For a given value of $\epsilon^2$, we perform the comparisons using 100 independent macro-replications and evaluate the computational efficiency of MLMC and SMC metamodeling by  averaging the computational costs recorded across these macro-replications.

\subsubsection{The Initial Value Problem}
Consider the variance function estimation example based on the initial value problem from \cite{mycek2019multilevel}, which involves the ordinary differential equation and the initial condition as follows:
\begin{align}
	\left\{\begin{array}{l}\frac{\mathrm{d} u(t)}{\mathrm{~d} t}=\lambda u(t), \quad t \in(0,1] \\u(0)=u_{0},\end{array}\right. \label{eq:de}
\end{align}
where the growth coefficient and initial condition $ \lambda$, $ u_{0}   \in \mathbb{R} $. The solution to \eqref{eq:de} is given by $ u(t)=u_{0} e^{\lambda t} $. Now consider the function $ F: \mathbb{R}^{2} \times[0,1] \rightarrow \mathbb{R}$, given by $F\left(u_{0}, \lambda, t\right) = u_{0} e^{\lambda t}$, where $u_0$ and $\lambda$ are realizations of independent random variables $U_{0}$ and $\Lambda$. For any $t \in [0, 1]$, the variance of $F(U_{0}, \Lambda, t)$ as a function of $t$, denoted as $\V{t}$, is given by
\begin{align*}
	\V{t} = \Var{F\left(U_{0}, \Lambda, t\right)}=\E{U_{0}^{2}} \E{e^{2 \Lambda t}}-\Esup{2}{U_{0}} \Esup{2}{e^{\Lambda t}}, \ t \in [0, 1].
\end{align*}
Assuming that $U_0$ and $\Lambda$ are independent normal random variables, i.e., $ U_{0} \sim \mathcal{N}\left(\mu_{0}, \sigma_{0}^{2}\right) $ and $ \Lambda \sim \mathcal{N}\left(\mu, \sigma^{2}\right) $, we can obtain the closed-form expression of the variance function as
\begin{align*}
	\V{t}=e^{2 \mu t+\sigma^{2} t^{2}}\left[\sigma_{0}^{2} e^{\sigma^{2} t^{2}}+\mu_{0}^{2}\left(e^{\sigma^{2} t^{2}}-1\right)\right],\ t \in [0, 1].
\end{align*}

For numerical evaluations, we set the parameters of the distributions  as $\mu_{0}=10$, $\sigma_{0}=2$, $\mu=-1$, and $\sigma=0.25$, respectively. Figure~\ref{fig:variance_function_example}(a) illustrates the true variance function $\V{t}$ for $t \in [0, 1]$. To implement MLMC metamodeling (Algorithm \ref{alg:mlmc}) and SMC metamodeling as described at the beginning of Subsection \ref{sec:example}, we select $\alpha = 1.5$ and $\gamma = 1$ as the parameters for the target accuracy procedure. The target MISE level $\epsilon^2$ is varied in $\left\{10^{-2}, 10^{-3}, 10^{-4}, 10^{-5}\right\}$.

Figure~\ref{fig:result_slope}(a) presents a comparison of the average computational cost across 100 macro-replications for MLMC metamodeling versus SMC metamodeling, targeting various MISE levels $\epsilon^2$.
The log-log plot shows the average computational cost corresponding to each target $\epsilon$ value. The empirical slopes, approximately $-1.3$ for SMC metamodeling and $-1.04$ for MLMC metamodeling, corroborate Theorem \ref{thm:computation}'s theoretical values of $-1.33$ and $-1$, respectively,  with $\phi=\gamma / 2\alpha = 0.33$. This agreement highlights the superior computational efficiency of MLMC metamodeling  compared to SMC  metamodeling.

\begin{figure}
     \centering
     \begin{subfigure}[b]{0.45\textwidth}
         \centering
         \includegraphics[width=\textwidth]{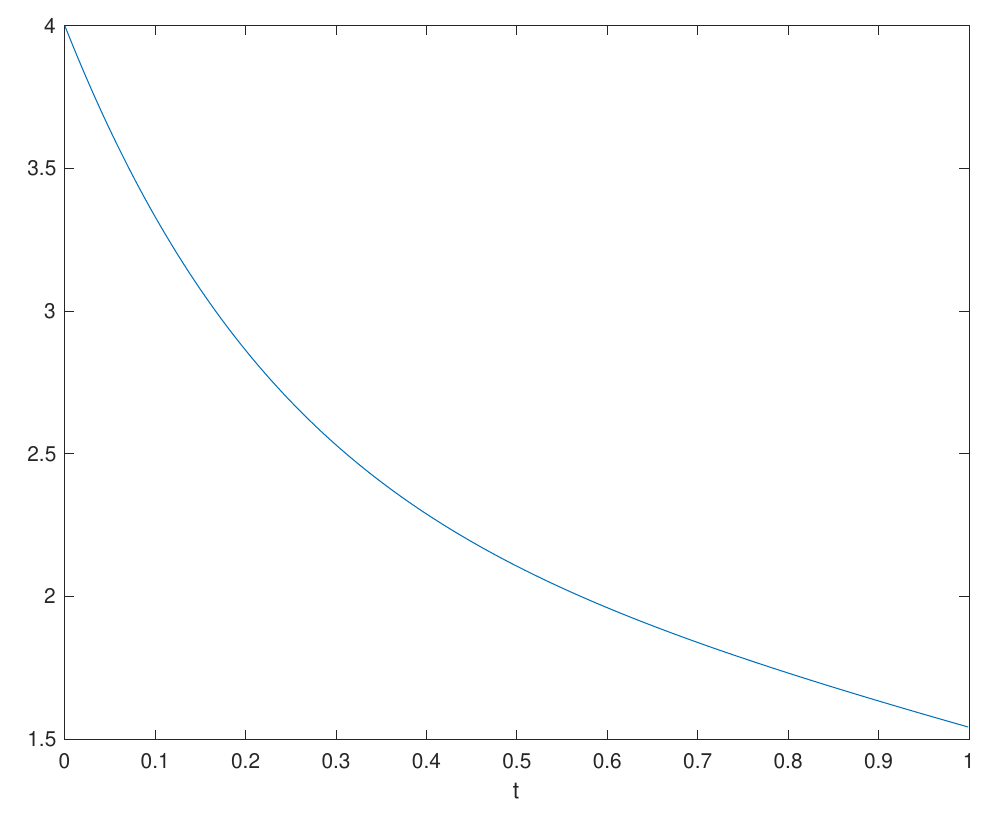}
         \caption{The initial value problem}
         \label{fig:variance_de}
     \end{subfigure}
     \hfill
     \begin{subfigure}[b]{0.48\textwidth}
         \centering
         \includegraphics[width=\textwidth]{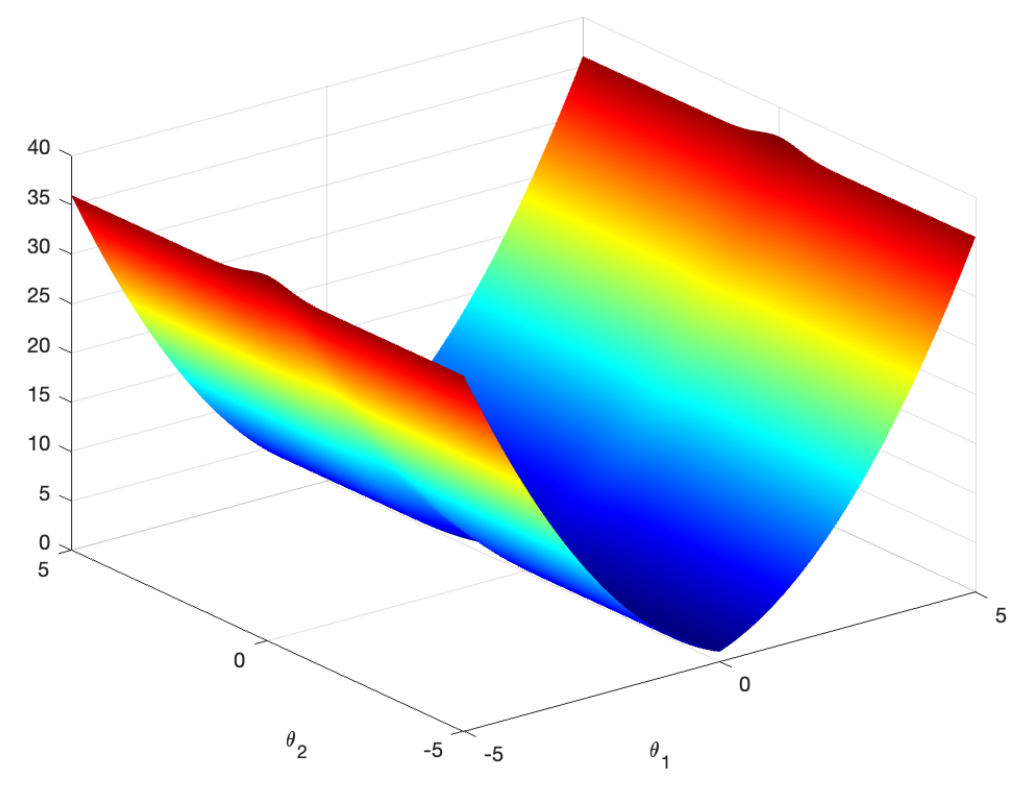}
         \caption{The Griewank function example}
         \label{fig:variance_Griewank}
     \end{subfigure}
        \caption{The true variance functions for the two examples in Subsection \ref{sec:example}.}
        \label{fig:variance_function_example}
\end{figure}

\begin{figure}
     \centering
     \begin{subfigure}[b]{0.48\textwidth}
         \centering
         \includegraphics[width=\textwidth]{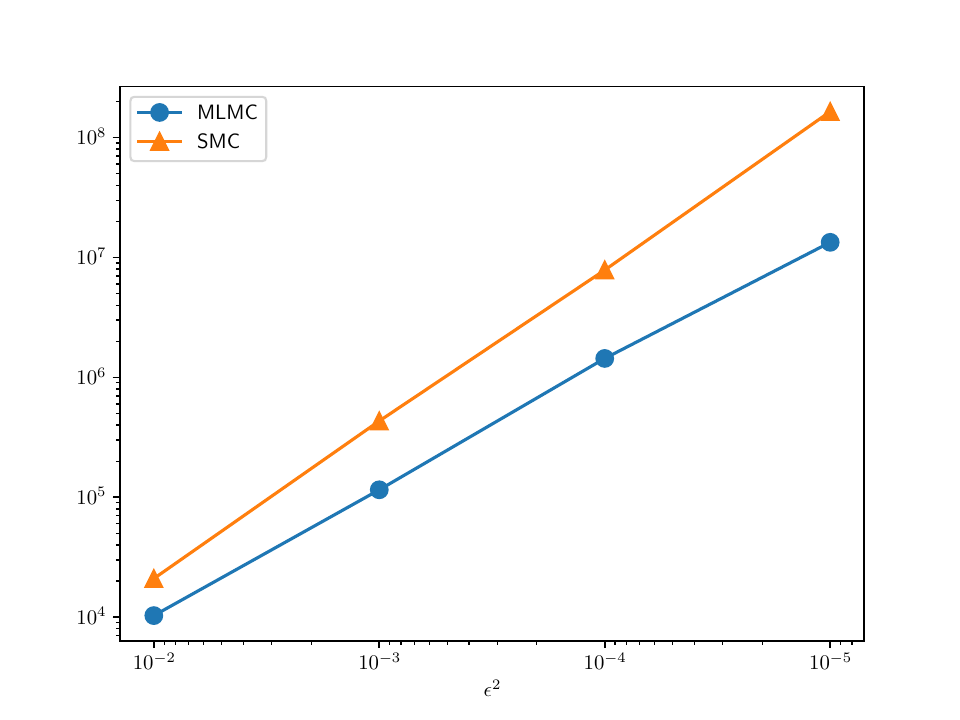}
         \caption{The initial value problem}
     \end{subfigure}
     \hfill
     \begin{subfigure}[b]{0.48\textwidth}
         \centering
         \includegraphics[width=\textwidth]{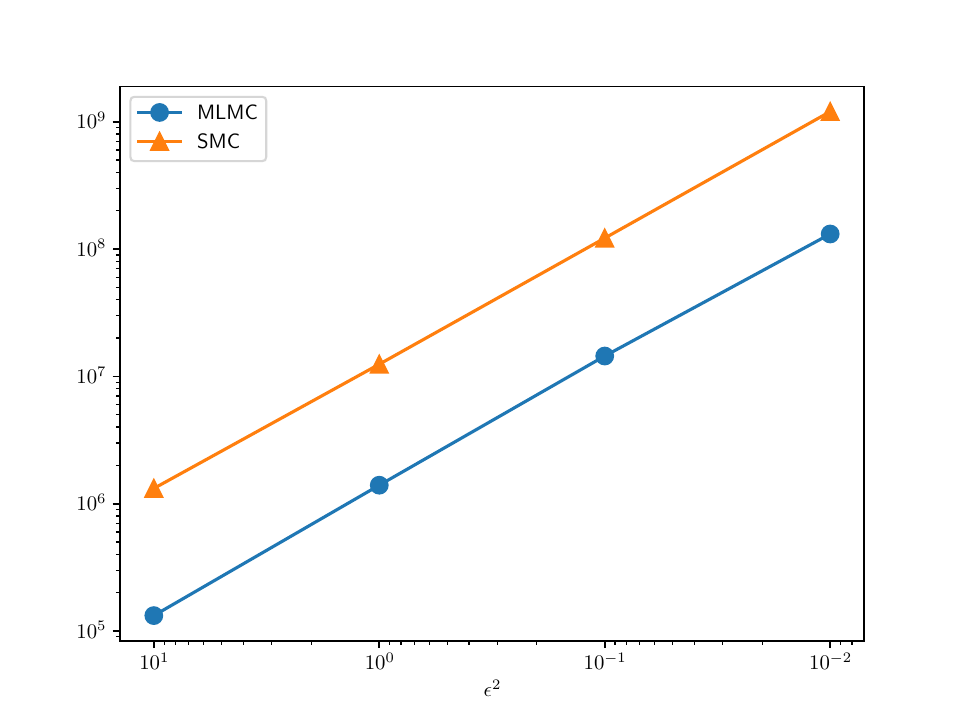}
         \caption{The Griewank function example}
     \end{subfigure}
        \caption{Log-log plot of the average computational cost ($y$-axis) against the target MISE level $\epsilon^2$ ($x$-axis) for the two examples in Subsection \ref{sec:example}, with both axes on a logarithmic scale.}
        \label{fig:result_slope}
\end{figure}

\subsubsection{The Griewank Function Example}
The following two-dimensional example is based on the Griewank  function considered in \cite{rosenbaum2017multilevel}, where the focus is on estimating the mean function. 
The simulation output is $\cY(\theta_1, \theta_2) = \mu(\theta_1, \theta_2) + \left(1+\left|\theta_{1}\right|\right) Z_{1}+\exp \left(-\theta_{2}^{2}\right) Z_{2} $ for $(\theta_1, \theta_2) \in [-5, 5]^2$, where $ \mu(\theta_1, \theta_2)=1+\left(\theta_{1}^{2}+\theta_{2}^{2}\right) / 4000-\cos \left(\theta_{1}\right) \cos \left(\theta_{2} / \sqrt{2}\right) $, with  $ Z_{1} $ and $ Z_{2} $ being independent standard normal random variables. 
The true variance function of the output is given by
\[\V{\theta_1, \theta_2} = \left(1 + |\theta_1|\right)^2 + \exp\left(-2\theta_2^2\right),  \ \theta \in  [-5, 5]^2,
\]
as illustrated in Figure \ref{fig:variance_function_example}(b).
 For implementing MLMC and SMC metamodeling approaches in this example, we use $\alpha=2$, $\gamma=2$, and vary the target MISE level $\epsilon^2$  in $\{10^{1}, 10^{0}, 10^{-1}, 10^{-2}\}$. 

Figure~\ref{fig:result_slope}(b) compares the average computational costs of MLMC and SMC metamodeling across 100 macro-replications for different MISE levels $\epsilon^2$. For MLMC metamodeling, we observe a slope close to $-1$, consistent with Theorem \ref{thm:computation}, where $\phi=\gamma / 2\alpha = 0.5$. In contrast, the slope for SMC metamodeling is also approximately $-1$, deviating from the expected value of $-1.5$ suggested by Theorem \ref{thm:computation_SMC}. The discrepancy can be attributed to two main reasons. First, Theorem \ref{thm:computation_SMC} provides an upper bound on the computational cost, which may not always be tight. Specifically, for an MISE level $\epsilon^2 < 1$, Theorem \ref{thm:computation_SMC} states that the upper bound of the computational cost for SMC metamodeling is $\cO(\epsilon^{-3})$, implying that the slope should not be steeper than -1.5. The observed slope for SMC metamodeling is -1, which aligns with this bound. 
The looseness of the observed bound may be attributed to the strong performance of SMC metamodeling in this example, which combines kernel smoothing with a suitable experimental design to effectively estimate the underlying variance function.
	Second, since the total number of design levels,  $L+1$, is discrete,  and the value of $L$ can remain unchanged for different MISE levels  $\epsilon^2$ specified, particularly when  $\epsilon^2$ is relatively large. This discreteness can cause the slope for SMC metamodeling to be less steep than the theoretical prediction. Similar observations and their underlying reasons were discussed by Rosenbaum and Staum (2017) in \cite{rosenbaum2017multilevel}.

We conclude this subsection with a summary of the performance of MLMC metamodeling for achieving a target MISE level. Table \ref{tab:reduction_cost} shows the ratio of the average computational cost of SMC metamodeling to that of MLMC metamodeling across all macro-replications. Compared with SMC metamodeling, MLMC metamodeling demonstrates substantial computational savings. For example, in the  initial value problem, MLMC metamodeling reduces the computational cost by half for the largest value of $\epsilon^2$ and by 92\% for the smallest value of $\epsilon^2$ considered. In the Griewank function example, MLMC metamodeling achieves nearly 90\% savings in computational cost across different values of $\epsilon^2$. Overall, MLMC metamodeling provides significant computational savings while delivering the prescribed estimation accuracy compared to SMC metamodeling.

\begin{table}
    \centering
    \caption{Reduction in computational cost achieved by MLMC metamodeling compared to SMC metamodeling for the examples in Subsection~\ref{sec:example}.}
    \begin{tabular}{|c|c|c|}
    \hline
        $\log_{10}\epsilon$ & Initial value & Griewank  \\ \hline
        0.5 & ~ & 9.97 \\
        0 & ~ & 8.92 \\ 
        -0.5 & ~ & 8.44  \\
        -1 & 2.03 & 9.15 \\
        -1.5 & 3.75 & ~ \\ 
        -2 &  5.46 & ~  \\ 
        -2.5 & 12.30 & ~  \\ \hline
    \end{tabular}
    
	\label{tab:reduction_cost}
\end{table}

\subsection{Application of MLMC Metamodeling to Global Sensitivity Analysis}\label{sec:sobol case}

This subsection examines the performance of the MLMC metamodeling procedure in utilizing a fixed computational budget for global sensitivity analysis. Subsection \ref{subsec:sobol} provides a brief review of Sobol' indices in global sensitivity analysis and existing estimation methods.  Subsection \ref{subsec:sobol_exp} presents numerical results demonstrating the application of MLMC metamodeling to Sobol' index estimation.

\subsubsection{Review of Sobol' Indices in Global Sensitivity Analysis}\label{subsec:sobol}
Global sensitivity analysis seeks to identify the most influential inputs---those for which a small variation results in a significant change in the model output. Sobol' indices are popular global sensitivity measures established based on the functional ANOVA decomposition \cite{saltelli2008global,im1993sensitivity}. They  quantify the impact of each input variable on the output of interest, rendering input-space dimensionality reduction by screening out input variables with low impacts \cite{lamboni2013derivative, saltelli2010variance}. 

Consider a simulation model with a $d$-dimensional  input vector $\boldsymbol{\theta} = (\theta_1, \theta_2, \dots, \theta_d) \in \Theta \subseteq \mathbb{R}^{d}$, and  $\cY(\boldsymbol{\theta})$ denotes the corresponding  output at $\boldsymbol{\theta}$.
Understanding the impact of inputs on the model's output is crucial for users. One approach to address this issue is to treat the inputs as random variables, which in turn makes the model output a random variable as well.
For an arbitrary non-empty index set $\mathbf{u}  \subseteq [d]$, define $\boldsymbol{\theta}_{\mathbf{u}}$ as the subset of entries in $\boldsymbol{\theta}$ with indices in  $\mathbf{u}$; for example, if $\mathbf{u}=\{1,2\}$, then $\boldsymbol{\theta}_{\mathbf{u}} = (\theta_1, \theta_2)$. 
The Sobol' index of $\boldsymbol{\theta}_{\mathbf{u}}$, representing the share of total variance of the output  that is due to the uncertainty in the set of input variables $\boldsymbol{\theta}_{\mathbf{u}}$,  is defined as
	\begin{equation}\label{eq:sobol_def}
		S_{\mathbf{u}} := \frac{\Var{\E{\cY\mid \boldsymbol{\theta}_{\mathbf{u}}}}}{\Var{\cY}} \ .
	\end{equation}
In particular, when $\boldsymbol{\theta}_{\mathbf{u}}$ consists of a single input, $S_{\mathbf{u}}$ is known as the first-order Sobol' index of $\boldsymbol{\theta}_{\mathbf{u}}$ \cite{mycek2019multilevel, sobol2001global}. By \eqref{eq:sobol_def}, all Sobol' indices take values between 0 and 1. The higher the value, the more significant the contribution of the partial input vector $\boldsymbol{\theta}_{\mathbf{u}}$ to the output $\cY$.

For Sobol' index estimation, commonly adopted approaches involve separately  estimating the denominator and  numerator in \eqref{eq:sobol_def} using simple Monte Carlo methods, including the well known pick-freeze scheme \cite{im1993sensitivity, sobol2001global}.
While estimating the denominator in \eqref{eq:sobol_def} is  straightforward, estimating the numerator $\Var{\E{\cY\mid \boldsymbol{\theta}_{\mathbf{u}}}}$ can be challenging. 
One method to address this is to rewrite its nested form into a covariance form as shown in Lemma 2.2 of \cite{janon2014asymptotic}:
\begin{equation}\label{eq:nest_mc_transform}
	\Var{\E{\cY\mid \boldsymbol{\theta}_{\mathbf{u}}}} = \Cov{\cY((\boldsymbol{\theta}_{\mathbf{u}},  \boldsymbol{\theta}_{\mathbf{-u}}))}{\cY((\boldsymbol{\theta}_{\mathbf{u}},  \boldsymbol{\theta}_{\mathbf{-u}}^{\prime}))},
\end{equation}
where $\boldsymbol{\theta}_{\mathbf{-u}}^{\prime}$ and $\boldsymbol{\theta}_{\mathbf{-u}}$ are independent  and identfically distributed. The computational efficiency of estimating $\Var{\E{\cY\mid \boldsymbol{\theta}_{\mathbf{u}}}}$ is improved by leveraging \eqref{eq:nest_mc_transform}, which eliminates the need for nested simulation. Recent advancements have further enhanced this efficiency, with studies incorporating techniques such as MLMC \cite{mycek2019multilevel} and multifidelity MC \cite{qian2018multifidelity} to boost the performance of MC-based estimation.

 When running the simulation model is computationally expensive, incorporating metamodeling techniques can significantly enhance efficiency for Sobol' index estimation \cite{castellan2020non, hart2017efficient, marrel2012global}. 
Previous research has explored applying metamodeling to estimate the numerator $\Var{\E{\cY\mid \boldsymbol{\theta}_{\mathbf{u}}}}$ in \eqref{eq:sobol_def},  while estimating the denominator $\Var{\cY}$ via SMC. For instance, Castellan et al.\ (2020) developed a metamodel $\widehat{Y}_m(\boldsymbol{\theta}_{\mathbf{u}})$ to approximate  $\E{\cY\mid \boldsymbol{\theta}_{\mathbf{u}}}$ \cite{castellan2020non}. They then estimated  $\Var{\E{\cY\mid \boldsymbol{\theta}_{\mathbf{u}}}}$ using
\begin{align*}
	\frac{1}{N}\sum_{i=1}^{N}\widehat{Y}_{m}^2(\boldsymbol{\theta}_{\mathbf{u}, i}) - \left(\frac{1}{N}\sum_{i=1}^{N}\widehat{Y}_{m}(\boldsymbol{\theta}_{\mathbf{u}, i})\right)^2,
\end{align*}
where $N$ denotes the MC sample size of $\boldsymbol{\theta}_{\mathbf{u}}$, and  $\boldsymbol{\theta}_{\mathbf{u}, i}$ represents the $i$th random observation of $\boldsymbol{\theta}_{\mathbf{u}}$ for $i \in [N]^{+}$. 

Consistent with the aforementioned approach, to leverage the proposed MLMC metamodeling approach for variance function estimation, we express $\Var{\E{\cY\mid \boldsymbol{\theta}_{\mathbf{u}}}}$ as $\Var{\cY} - \E{\Var{\cY \mid \boldsymbol{\theta}_{\mathbf{u}} }}$ and construct a  metamodel $\Vhat{\boldsymbol{\theta}_{\mathbf{u}}}$ as specified in \eqref{eq:mlmc_metamodeling_estimator} to approximate $\Var{\cY \mid \boldsymbol{\theta}_{\mathbf{u}} }$. The Sobol' index $S_{\mathbf{u}}$ in \eqref{eq:sobol_def} can be estimated by
 \begin{equation}\label{eq:sobol_estimator}
 	\widehat{S}_{\mathbf{u}} = \frac{M^{-1}\sum_{i=1}^{M}(\cY_i-\bar{\cY})^2 - N^{-1}\sum_{i=1}^{N} \Vhat{\boldsymbol{\theta}_{\mathbf{u}, i}}}{M^{-1}\sum_{i=1}^{M}(\cY_i-\bar{\cY})^2} \ ,
 \end{equation}
where $M$ denotes the sample size of the simulation outputs for estimating $\Var{\cY}$. Each $\cY_i$ in \eqref{eq:sobol_estimator} denotes an independent simulation output, and $\bar{\cY}$ is their sample average.

\subsubsection{Application to Global Sensitivity Analysis}\label{subsec:sobol_exp}
This subsection demonstrates the MLMC metamodeling procedure for allocating a fixed budget, as outlined in Algorithm \ref{alg:mlmc_budget}, to estimate the first-order Sobol' indices. Consider the Ishigami function, a widely used example for evaluating global sensitivity analysis approaches \cite{ishigami1990importance}: 
\begin{align*}
	\cY=\sin \left(X_{1}\right)+7 \sin \left(X_{2}\right)^{2}+0.1 X_{3}^{4} \sin \left(X_{1}\right) \ ,
\end{align*}
where $X_i$'s are independent and uniformly distributed in $[-\pi, \pi]$. We are interested in estimating the first-order Sobol' indices associated with $X_i$,  given by
\begin{align*}
	S_i = \frac{\Var{\cY} - \E{\Var{\cY \mid X_i }}}{\Var{\cY}}, \quad i =1,2, 3 \  .
\end{align*}
The true values are known in this case \cite{castellan2020non}, and are given by $ S_1 = 0.3139, S_2 = 0.4424$, and $S_3=0$.

To construct the Sobol' index estimator given in \eqref{eq:sobol_estimator}, we first estimate the following conditional variance functions:
\begin{align*}
	&\V{x_1} = \Var{\cY \mid X_1 = x_1} = \frac{49}{8}+\frac{4 \pi^{8}}{5625} \sin \left(x_1\right)^{2}, & x_1 \in [-\pi, \pi] \ , \\
	&\V{x_2} = \Var{\cY \mid X_2 = x_2} = \frac{1}{2} + \frac{\pi^4}{50} + \frac{\pi^8}{1800},  & x_2 \in [-\pi, \pi] \ , \\
\text{and } & \V{x_3} = \Var{\cY \mid X_3 = x_3} = \frac{49}{8} + \frac{1}{2}(1 + 0.1 x_3^4)^2, & x_3 \in [-\pi, \pi] \ ,
\end{align*}
and obtain their estimators $\Vhat{x_1}, \Vhat{x_2}$, and $\Vhat{x_3}$ as defined at the end of Subsection \ref{subsec:sobol} via MLMC metamodeling.

In our numerical implementation, a fixed budget of $10,000$ is allocated for constructing  $\Vhat{x_i}$ for $i = 1, 2, 3$ using MLMC metamodeling  with Algorithm \ref{alg:mlmc_budget}, where the parameters $\alpha$, $\gamma$, $s$, and $A$ are all set to $2$. 
The SMC metamodeling estimator is constructed using the same design as that for constructing the finest level's estimator  $\Vhatsub{L}{\cdot}$ as described in \eqref{eq:metamodel}. 
The parameters for the function approximation method (i.e., kernel smoothing using the Gaussian kernel) are determined through a leave-one-out cross-validation procedure. Finally, to construct the Sobol' index estimator $\widehat{S}_i$, we use a sample size of $N=10,000$ for randomly sampling  $X_i$ ($i=1,2,3$) and an additional sample size of $M=10,000$ for generating outputs to estimate $\Var{\cY}$, as described in \eqref{eq:sobol_estimator}. We mention, without showing details, that Assumptions \ref{asm:domain} through \ref{asm:LP_Y} and Assumption \ref{asm:decay_ratio} can be verified to hold in this study.

For performance evaluation, we adopt Nadaraya-Watson and wavelets estimators considered in \cite{castellan2020non} as benchmarking approaches. The experimental settings, including the budgets  for constructing metamodels and for estimating $\Var{\cY}$, are consistent with those described above. We compare the mean squared error (MSE) of the Sobol' index estimator  $\widehat{S}_i$ for $i=1,2,3$. The total number of independent macro-replications is set to $R=1,000$. For each macro-replication $j$, the estimator for $S_i$ is denoted  $\widehat{S}_{i}^{(j)}$, and the MSE of $\widehat{S}_i$ for a given method is calculated as $R^{-1}\sum_{j=1}^{R}\left(\widehat{S}_{i}^{(j)} - S_i \right)^2 $.

We first examine the estimated conditional variance functions using MLMC and SMC metamodeling in comparison with the true functions, as shown in Figure \ref{fig:compare ishigami conditional variance}. Both MLMC and SMC metamodeling methods provide adequate estimates of the variance functions, with MLMC metamodeling outperforming SMC metamodeling.

The resulting MSEs for the Sobol' index estimators are shown in Table \ref{tab:sobol}. We observe that MLMC metamodeling excels in estimating $S_1$ and $S_2$. Although MLMC metamodeling performs slightly worse than the Nadaraya-Watson and wavelet estimators for $S_3$, it is worth noting that the numerator in \eqref{eq:sobol_def}, associated with small Sobol' indices (such as $S_3=0$ in this example), is inherently challenging to estimate,  as discussed in \cite{owen2013better}. Despite this, MLMC metamodeling performs adequately in ranking the significance of all inputs and demonstrates competitive performance overall.

\begin{figure}[ht]
     \centering
\begin{subfigure}[b]{0.47\textwidth}
    \centering
    \includegraphics[width=\textwidth]{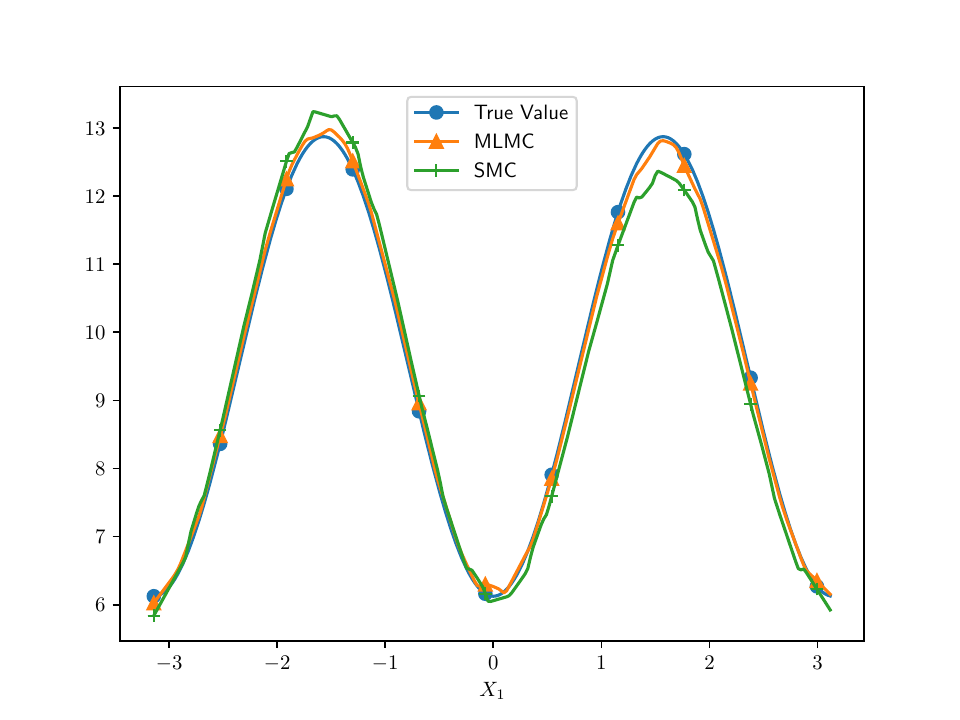}
    \caption{$\Var{\cY \mid X_1=x_1}$}
    \label{fig:true ishigami conditional variance x1}
\end{subfigure}
\hfill
\begin{subfigure}[b]{0.47\textwidth}
    \centering
    \includegraphics[width=\textwidth]{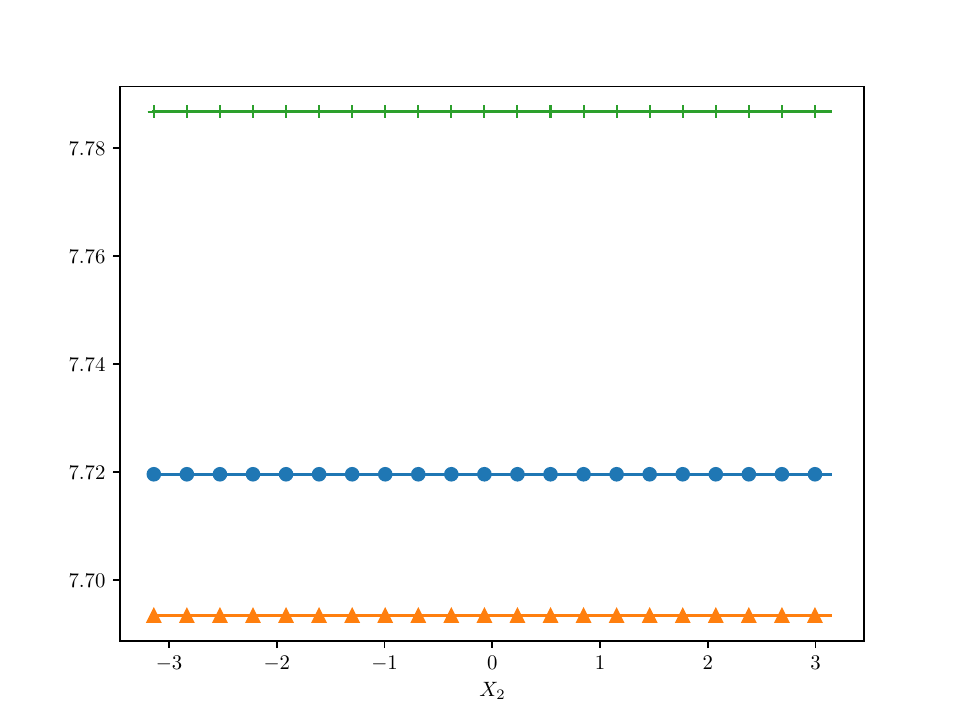}
    \caption{$\Var{\cY \mid X_2=x_2}$}
    \label{fig:true ishigami conditional variance x2}
\end{subfigure}
\hfill
\begin{subfigure}[b]{0.47\textwidth}
    \centering
    \includegraphics[width=\textwidth]{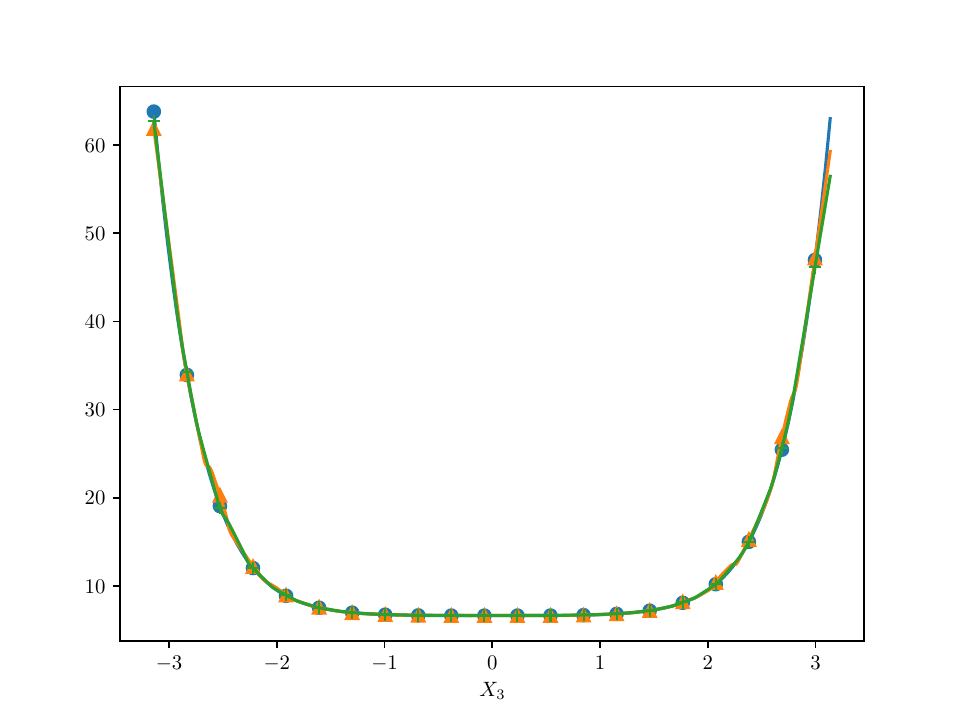}
    \caption{$\Var{\cY \mid X_3=x_3}$}
    \label{fig:true ishigami conditional variance x3}
\end{subfigure}
\caption{True conditional variance functions of the Ishigami function and their estimates obtained using MLMC and SMC metamodeling methods on an arbitrary macro-replication.}
\label{fig:compare ishigami conditional variance}
\end{figure}

\begin{table}
    \centering
    \caption{Summary of the resulting MSEs for the Sobol' index estimators obtained using MLMC metamodeling, SMC metamodeling, Nadaraya-Watson, and wavelet methods. Results for the Nadaraya-Watson and wavelet estimators are directly cited from \cite{castellan2020non}.
    }
    \begin{tabular}{|c|c|c|c|}
    \hline
        Method & $\widehat{S}_1$ & $\widehat{S}_2$ & $\widehat{S}_3$ \\ \hline
        MLMC metamodeling & $8.22\times 10^{-6}$ & $2.08\times 10^{-5}$ & $6.98\times 10^{-5}$\\
        SMC metamodeling  & $1.85\times 10^{-5}$ & $8.03\times 10^{-5}$ & $5.29\times 10^{-4}$ \\ 
        Nadaraya-Watson &  $1.30\times 10^{-5}$ & $1.10\times 10^{-4}$ & $1.10\times 10^{-7}$ \\
        Wavelets & $4.00\times 10^{-5}$ & $6.60\times 10^{-5}$ & $2.20\times 10^{-5}$ \\ \hline
    \end{tabular}
	\label{tab:sobol}
\end{table}

\section{Conclusion}\label{sec:conclusion}
In this work, we proposed a multilevel Monte Carlo (MLMC) metamodeling approach for variance function estimation. Under mild assumptions, we demonstrated that the proposed method achieves the same order of computational cost as the MLMC metamodeling approach for mean function estimation \cite{rosenbaum2017multilevel} while meeting a prescribed MISE level. We also conducted asymptotic analyses of the proposed MLMC metamodeling estimator. Additionally, we presented two MLMC procedures for practical implementation: one for achieving a target MISE level and another for utilizing a fixed computational budget. Numerical examples illustrated the efficiency and efficacy of the proposed variance function estimation approach and validated the theoretical results. An application in global sensitivity analysis showed that MLMC metamodeling performs competitively in supporting Sobol' index estimation.

Several directions for future research can be explored. Firstly, developing a dual MLMC metamodeling framework for simultaneous estimation of mean and variance functions could be valuable, potentially benefiting mean-variance analysis and robust simulation optimization. Secondly, while this work builds on classical MLMC (also known as geometric MLMC \cite{giles2015multilevel}), integrating other MLMC variants, such as randomized MLMC, may further enhance performance.
One limitation of the proposed MLMC metamodeling procedures is the curse of dimensionality in high-dimensional input spaces. Improving computational efficiency may involve advanced sampling and design strategies, such as sparse grids. The multi-index Monte Carlo method \cite{haji2016multi}, which integrates sparse grids to extend MLMC for efficient high-dimensional integration, also offers a promising approach for variance function estimation.

\appendix

\section{Proofs in Section \ref{sec:theory}}
\label{app:proofs-Section3}

\subsection{Proof of Lemma \ref{lem:finite_moment}}\label{app:lem_finite_moment}
\begin{proof}
 By Assumption \ref{asm:finite_moment}, there exists some $\theta^{\prime} \in \Theta$ such that $\Msup{4}{\cY(\theta^{\prime})}< \infty$. For any $\theta \in \Theta$, we have
	\begin{align}
		& \Msup{4}{\cY(\theta)} =  \mathbb{E}\left[\left(\cY(\theta)  - \mathbb{E}\left[\cY(\theta)\right]\right)^4\right] \nonumber \\
		 = & \ \mathbb{E}\left[\left(\cY(\theta) - \cY(\theta^\prime ) + \cY(\theta^\prime ) - \E{\cY(\theta^\prime )} + \E{\cY(\theta^\prime )} - \E{\cY(\theta)}\right)^4\right] \nonumber \\
		\leq & \ 27 \E{(\cY(\theta) - \cY(\theta^\prime ))^4} + 27 \E{(\cY(\theta^\prime ) - \E{\cY(\theta^\prime )})^4} + 27(\E{\cY(\theta^\prime )} - \E{\cY(\theta)})^4 \nonumber \\
		\leq & \ 54 \E{(\cY(\theta) - \cY(\theta^\prime ))^4} + 27 \E{(\cY(\theta^\prime ) - \E{\cY(\theta^\prime )})^4} \nonumber \\
		\leq & \ 54 \E{k_y^4} \|\theta - \theta^\prime \|^4 + 27 \Msup{4}{\cY(\theta^\prime ) } \ , \label{eq:A3-2-update-pf}
	\end{align}
	where  the first inequality follows from the power mean inequality  and the last inequality follows from Assumption \ref{asm:LP_Y}. Given that $\Msup{4}{\cY(\theta^\prime) } < \infty$ and Assumption \ref{asm:domain} hold, it follows from \eqref{eq:A3-2-update-pf} that  $\Msup{4}{\cY(\theta)}  < \infty$ for $\forall \theta \in \Theta$.
\end{proof}

\subsection{Proof of Lemma \ref{lem:LP_V}}\label{app:lem_LP_V}

\begin{proof}
		
For any $\theta_1, \theta_2 \in \Theta$,
we  observe that
\begin{align*}
			\left| \V{\theta_1} - \V{\theta_2} \right| & = \left| \E{\cY^2(\theta_1)} -\Esup{2}{\cY(\theta_1)} - \E{\cY^2(\theta_2)}  + \Esup{2}{\cY(\theta_2)} \right| \\
			& \leq \underbrace{\left| \E{\cY^2(\theta_1) - \cY^2(\theta_2)} \right|}_{(i)} + \underbrace{ \left| \Esup{2}{\cY(\theta_1)} - \Esup{2}{\cY(\theta_2)} \right| }_{(ii)} \ .
		\end{align*}
        For the term $(i)$, we have
        \begin{align*}
        	 \left| \E{\cY^2(\theta_1) - \cY^2(\theta_2)} \right| \nonumber = & \left| \E{ \left( \cY(\theta_1) - \cY(\theta_2) \right) \left( \cY(\theta_1) + \cY(\theta_2) \right) } \right|  \\
        	\leq & \Esup{1/2}{\left| \cY(\theta_1) + \cY(\theta_2) \right|^2} \Esup{1/2}{\left| \cY(\theta_1) - \cY(\theta_2) \right|^2}   \\
        	\leq &  \underbrace{ \Esup{1/2}{\left| \cY(\theta_1) + \cY(\theta_2) \right|^2} }_{(iii)} \cdot \Esup{1/2}{\kappa_y^2} \cdot \left\|\theta_1 - \theta_2 \right\| \ , 
        \end{align*}
 		where the first inequality on the right-hand side follows from the Cauchy-Schwarz inequality, and the second inequality follows from Assumption \ref{asm:LP_Y}.
 		
 		Notice that Lemma \ref{lem:finite_moment} implies the existence of a constant $C_1 > 0$ such that the term $(iii) \leq C_1$. Hence, it follows that the term $(i)$ can be bounded as follows:
 		\begin{equation}\label{eq:expectation_term1_bound}
 			\left| \E{\cY^2(\theta_1) - \cY^2(\theta_2)} \right| \leq C_1 \cdot  \Esup{1/2}{\kappa_y^2} \cdot \left\|\theta_1 - \theta_2 \right\| \ .
 		\end{equation} 
Similarly, by Assumption \ref{asm:LP_Y} and Lemma \ref{lem:finite_moment}, we can demonstrate that the term $(ii)$ can be bounded as follows:
 		\begin{equation}\label{eq:expectation_term2_bound}
 			\left| \Esup{2}{\cY(\theta_1)} - \Esup{2}{\cY(\theta_2)} \right| \leq C_2 \cdot \E{\kappa_y} \cdot \left\|\theta_1 - \theta_2 \right\| \ ,
 		\end{equation}
 		where $C_2 > 0$ is a  constant. By combining \eqref{eq:expectation_term1_bound} and \eqref{eq:expectation_term2_bound}, we obtain the following bound:
 		\[\left| \V{\theta_1} - \V{\theta_2} \right| \leq  \left(C_1\Esup{1/2}{\kappa_y^2} + C_2\E{\kappa_y}\right) \cdot  \left\|\theta_1 - \theta_2 \right\| \ .\]
 The proof is complete by setting $\kappa_v = C_1\Esup{1/2}{\kappa_y^2} + C_2\E{\kappa_y}$.

\end{proof}

\subsection{Proof of Lemma \ref{lem:var_bound}}\label{app:lem_var_bound}

\begin{proof}
The proof is in a similar spirit to the proof of the upper bound given in Equation (2.29) of \cite{mycek2019multilevel}.
Recall that $\V{\theta}$ denotes the variance of $\cY(\theta, \omega)$ and $\cV(\theta,M_{\ell})$ denotes the corresponding sample variance based on $M_{\ell}$ replications run at design point $\theta \in \Theta$. Given two random variables $W$ and $W^{\prime}$, recall that $ \Msup{4}{W} \coloneqq \E{(W-\E{W})^{4}} $ and define $ \Msup{4}{W, W^{\prime}} \coloneqq \E{(W-\E{W})^{2}(W^{\prime}-\E{W^{\prime}})^{2}} $. Let $\Cov{\cdot}{\cdot}$ and $\Covs{\cdot}{\cdot}$ denote the covariance and the sample covariance calculated with a Monte Carlo sample size of $M_{\ell}$, respectively. Given the realizations $\{W_m\}_{m \in [M_{\ell}]^{+}}$ and $\{W^{\prime}_m\}_{m \in [M_{\ell}]^{+}}$  for $W$ and $W^{\prime}$, the sample covariances of $W$ and $W^{\prime}$ are computed as follows: 

\[  \Covs{W}{W^{\prime}}  = \frac{M_{\ell}}{M_{\ell}-1}\left(\frac{1}{M_{\ell}}\sum_{m=1}^{M_{\ell}}W_m W^{\prime}_m - \frac{1}{M_{\ell}}\sum_{m=1}^{M_{\ell}}W_m \cdot\frac{1}{M_{\ell}}\sum_{m=1}^{M_{\ell}}W^{\prime}_m\right) \ . \] 
 For $M_{\ell} \geq 2$, we have
\begin{align}\label{eq:variance_of_covariance}
	\Var{\Covs{W}{W^{\prime}}} =&\frac{\Msup{4}{W, W^{\prime}}}{M_{\ell}}-\frac{(M_{\ell}-2) \Covsup{2}{W}{W^{\prime}}-\Var{W} \Var{W^{\prime}}}{M_{\ell}(M_{\ell}-1)} \nonumber \\ \leq & \frac{\Msup{4}{W, W^{\prime}}}{M_{\ell}} + \frac{\Var{W} \Var{W^{\prime}}}{M_{\ell}(M_{\ell}-1)} \ .
\end{align}
By the Cauchy-Schwarz inequality, we have 
\begin{equation}\label{eq:m4_inequaltity}
	\Msup{4}{W, W^{\prime}} \leq \sqrt{\Msup{4}{W} \Msup{4}{W^{\prime}}} \ .
\end{equation} 
Using Jensen’s inequality, we have 
\begin{equation}\label{eq:variance_square_bound}
	\left(\Var{W} \Var{W^{\prime}}\right)^{2}=\Esup{2}{\left( W - \E{W} \right)^{2}} \Esup{2}{\left( W^{\prime} - \E{W^{\prime}} \right)^{2}} \leq \Msup{4}{W} \Msup{4}{W^{\prime}} \ . 
\end{equation}
It follows from \eqref{eq:variance_of_covariance} to \eqref{eq:variance_square_bound} that $ \Var{\Covs{W}{W^{\prime}}} \leq (M_{\ell}-1)^{-1} \sqrt{\Msup{4}{W} \Msup{4}{W^{\prime}}}$.
Given any $\theta_1, \theta_2 \in \Theta$, define $\Delta_{\cY} := \cY(\theta_1, \omega) - \cY(\theta_2, \omega)$ and  $\Sigma_{\cY} := \cY(\theta_1, \omega) + \cY(\theta_2, \omega)$. It is easy to show that $\V{\theta_1} - \V{\theta_2} = \Cov{\Delta_{\cY}}{\Sigma_{\cY}}$ and $\cV(\theta_1,M_{\ell}) -\cV(\theta_2,M_{\ell}) = \Covs{\Delta_{\cY}}{\Sigma_{\cY}}$. Hence,
\begin{align*}
	\Var{\cV(\theta_1,M_{\ell}) -\cV(\theta_2,M_{\ell}) } &=  \E{\left( \cV(\theta_1,M_{\ell}) - \cV(\theta_2,M_{\ell})  - \left( \V{\theta_1} - \V{\theta_2}  \right) \right)^2 } \\
	&=  \E{ \left( \Covs{\Delta_{\cY}}{\Sigma_{\cY}} - \Cov{\Delta_{\cY}}{\Sigma_{\cY}} \right)^2 } \\
	&=  \Var{ \Covs{\Delta_{\cY}}{\Sigma_{\cY}} } \leq  \frac{1}{M_{\ell}-1} \sqrt{\Msup{4}{\Delta_{\cY}} \Msup{4}{\Sigma_{\cY}}} \ .
\end{align*}
It follows  that
\begin{align*}
	\Msup{4}{\Delta_{\cY}} &= \E{ \left( \Delta_{\cY} - \E{ \Delta_{\cY} } \right)^4} \\
	&\leq \E{ \left( 2 \Delta_{\cY}^2 + 2\Esup{2}{\Delta_{\cY}} \right)^2 } \\
	&\leq 4 \E{ \left( \kappa_{y}^2 \|\theta_1 - \theta_2\|^2 + \Esup{2}{\kappa_{y}} \|\theta_1 - \theta_2\|^2 \right)^2 } \\
	&= 4 \left(\E{ \kappa_{y}^4 } + 2\Esup{2}{\kappa_{y}^2} + \Esup{4}{\kappa_{y}} \right) \|\theta_1 - \theta_2\|^4 \ ,
\end{align*}
where the second inequality on the right-hand side follows from Assumption \ref{asm:LP_Y}. On the other hand,
\begin{align*}
	\Msup{4}{\Sigma_{\cY}} &= \E{ \left( \Sigma_{\cY} - \E{\Sigma_{\cY}} \right)^4 } \\
	&\leq \E{ 8 \left( \cY(\theta_1, \omega) - \E{\cY(\theta_1, \omega)} \right)^4 + 8 \left( \cY(\theta_2, \omega)  - \E{ \cY(\theta_2, \omega) } \right)^4  } \\
	&= 8 \left( \Msup{4}{\cY(\theta_1)} + \Msup{4}{\cY(\theta_2)} \right) 
	\leq 16 c_{\cY} \ ,
\end{align*}
which is finite according to Lemma \ref{lem:finite_moment}.
Hence,
\begin{align*}
	\Var{ \cV(\theta_1,M_{\ell}) -\cV(\theta_2,M_{\ell}) }
	\leq & \frac{1}{M_{\ell}-1} \sqrt{\Msup{4}{\Delta_{\cY}} \Msup{4}{\Sigma_{\cY}}} \\
	\leq & \frac{8}{M_{\ell}-1}\|\theta_1 - \theta_2\|^2 \sqrt{\left(\E{ \kappa_{y}^4 } + 2\Esup{2}{ \kappa_{y}^2 } + \Esup{4}{ \kappa_{y} } \right) c_{\cY}} \ .  
	\end{align*}
	
\end{proof}

\subsection{Proof of Theorem \ref{thm:computation}}\label{app:thm_proof}

\begin{proof}
	The proof is in the same vein as that of Theorem 2 in \cite{rosenbaum2017multilevel}.
    To achieve a target MISE level $\epsilon^2$, we consider the computational cost required to ensure that both the bias and variance components in \eqref{eq:mise} are less than $\epsilon^2 / 2$. 
    
    For the MLMC metamodeling estimator, as the bias component
    $\left\|\Bias{\widehat{\mathbb{V}}}\right\|_{2}^{2} \leq b^2 s^{-2 \alpha L} < \epsilon^2 / 2$ by Condition 1, the index of the finest design level $L$ satisfies
     \begin{equation}\label{eq:level_num}
     	L = \ceil[\Big]{\left(\log_s\left(\sqrt{2}b \epsilon^{-1}\right) \right) / \alpha } \ .
     \end{equation}
     Regarding the variance component of the estimator, consider the following problem to minimize the total computational cost $C$:
    \begin{align*}
    \mbox{ min } \ &C = \sum_{\ell=0}^{L}M_{\ell}N_{\ell} \\
    \text{s.t. } &\sum_{\ell=0}^{L} \frac{\left\| v_{\ell} \right\|_{1}}{M_{\ell}-1}  < \frac{\epsilon^2}{2} \\
    & M_{\ell} \geq 1, \ell \in [L] \ .
	\end{align*}  
	Given Conditions 2 and 3, the optimal solution to this problem has a closed form given by
	\begin{equation}\label{eq:mell_origin}
		M_{\ell} = \ceil[\Big]{ 2 \epsilon^{-2} \sqrt{ \left\| v_{\ell} \right\|_{1} / N_{\ell}} \sum_{\ell^{\prime}=0}^{L} \sqrt{ \left\| v_{\ell^{\prime}} \right\|_{1}N_{\ell^{\prime}}} +1}\ , \quad \forall \ell \in [L] \ .
	\end{equation} 
   For the sake of simplicity, we modify \eqref{eq:mell_origin} to $M_{\ell} = \ceil[\Big]{ 2 \epsilon^{-2} \sqrt{ \left\| v_{\ell} \right\|_{1} / N_{\ell}} \sum_{\ell^{\prime}=0}^{L} \sqrt{ \left\| v_{\ell^{\prime}} \right\|_{1}N_{\ell^{\prime}}}}$ for $\forall \ell \in [L]$.  Notice that this adjustment does not alter the computational cost's order of magnitude. We consider three cases based on the relative magnitudes of $\gamma$ and $2\alpha$. 
   
   \begin{enumerate}
   	\item $\gamma = 2\alpha$: It follows from the Conditions 3 and 4  that $M_{\ell} \leq 2\epsilon^{-2}(L+1)\sigma^2 s^{-\gamma \ell}+1$. Hence,
    \begin{align*}
	C = \sum_{\ell=0}^{L}M_{\ell}N_{\ell} \leq \sum_{\ell=0}^{L} \left(2\epsilon^{-2}(L+1)\sigma^2 s^{-\gamma \ell} +1 \right) \cdot c s^{\gamma \ell}  
	= \sum_{\ell=0}^{L} 2\epsilon^{-2}  (L+1) \sigma^2  c + \sum_{\ell=0}^{L} c \cdot s^{\gamma \ell} \ .
\end{align*}
     In light of \eqref{eq:level_num}, we have 
    \[ \sum_{\ell=0}^{L} s^{\gamma \ell} \leq \frac{s^{\gamma L}}{1 - s^{- \gamma}} \leq \frac{s^{\gamma (\alpha^{-1} \log_{s}(\sqrt{2}b\epsilon^{-1})+1)}}{1 - s^{-\gamma}} =\frac{2s^{\gamma}b^2}{1 - s^{-\gamma}} \cdot \epsilon^{-2} \ .\] 
    It follows that $C \leq 2\epsilon^{-2} \sigma^2 c (L+1)^2 + 2\epsilon^{-2}c(L+1)s^{\gamma} b^2 / (1 - s^{-\gamma})$. For $\epsilon < e^{-1} < 1$,  we have $1 < \log \epsilon^{-1}$ and $\epsilon^{-\gamma/\alpha} = \epsilon^{-2} \leq \epsilon^{-2}(\log \epsilon)^2$,  and it follows that
    \begin{align*}
        \epsilon^{-2}(L+1)^2 \leq \epsilon^{-2} \left( \alpha^{-1} \log_{s} \sqrt{2}b \epsilon^{-1} + 2 \right)^2 = \epsilon^{-2} \left( \frac{\alpha^{-1}}{\log s} \left( \log \sqrt{2}b \epsilon^{-1}  + 2 \right) \right)^{2} = \mathcal{O}\left( \epsilon^{-2}(\log \epsilon^{-1})^2 \right) \ ,
    \end{align*}
   and hence $C =\mathcal{O}\left( \epsilon^{-2}(\log \epsilon^{-1})^2 \right)$.
    
    \item  $\gamma < 2\alpha$: We have $M_{\ell} \leq 2\epsilon^{-2}\sigma^{2} s^{-(2\alpha + \gamma)\ell / 2} \sum_{\ell=0}^{L}s^{(\gamma - 2\alpha)\ell / 2}  + 1$. Notice that $\sum_{\ell=0}^{L}s^{(\gamma - 2\alpha)\ell / 2} = \left(1 - s^{(\gamma - 2\alpha)(L+1)/2} \right) / \left(1 - s^{(\gamma - 2\alpha)/2}\right) \leq \left(1 - s^{(\gamma - 2\alpha)/2} \right)^{-1}$, thus 
   $ M_{\ell} \leq 2\epsilon^{-2}\sigma^2 \left(1 - s^{(\gamma - 2\alpha)/2} \right)^{-1} s^{-(2\alpha + \gamma)\ell / 2} + 1$. 
    It follows that
    \begin{align}
    	C & \leq \sum_{\ell=0}^{L} \left(2\epsilon^{-2}\sigma^2 \left(1 - s^{(\gamma - 2\alpha)/2} \right)^{-1} s^{-(2\alpha + \gamma)\ell / 2} + 1 \right) \cdot c s^{\gamma \ell} \nonumber\\
    	& =2\epsilon^{-2}\sigma^2 c\left(1 - s^{(\gamma - 2\alpha)/2} \right)^{-1} \sum_{\ell=0}^{L}s^{-(2\alpha - \gamma)\ell / 2} + \sum_{\ell=0}^{L}c \cdot s^{\gamma \ell}   \ .\label{eq:gammalesstwoalpha} 
    \end{align}
    Since both terms on the right-hand side of \eqref{eq:gammalesstwoalpha} are $\mathcal{O}(\epsilon^{-2})$, we have $C = \mathcal{O}(\epsilon^{-2})$.

    \item $\gamma > 2\alpha$: Notice that
    \begin{align*}
        \sum_{\ell=0}^{L} s^{\frac{(\gamma -2\alpha)\ell}{2}} = \frac{1 - s^{\frac{\gamma - 2\alpha}{2}(L+1)}}{1 - s^{\frac{\gamma - 2\alpha}{2}}} = \frac{s^{\frac{2\alpha - \gamma}{2}} - s^{\frac{\gamma - 2\alpha}{2}L}}{s^{\frac{2\alpha - \gamma}{2}} - 1} \leq \frac{s^{\frac{\gamma - 2\alpha}{2}L}}{1 - s^{\frac{2\alpha - \gamma}{2}}} \ .
    \end{align*}
    It follows that
    \begin{align*}
        M_{\ell} &\leq 2\epsilon^{-2}\sigma^{2} s^{-(2\alpha + \gamma)\ell / 2} \sum_{\ell=0}^{L}s^{(\gamma - 2\alpha)\ell / 2}  + 1 \leq 2 \epsilon^{-2} \sigma^2 s^{\frac{\gamma - 2\alpha}{2}L} \left(1 - s^{\frac{2\alpha - \gamma}{2}} \right)^{-1} s^{- \frac{(2\alpha + \gamma)\ell}{2}} + 1 \ .
    \end{align*}
    Hence,
    \[C = \sum_{\ell=0}^{L}N_{\ell}M_{\ell} \leq  2 \epsilon^{-2} \sigma^2 s^{\frac{\gamma - 2\alpha}{2}L} \left(1 - s^{\frac{2\alpha - \gamma}{2}} \right)^{-2} + c \sum_{\ell=0}^{L}s^{\gamma \ell} \ .\]
    Since $ s^{(\gamma - 2\alpha)L} \leq s^{(\gamma - 2\alpha) (\alpha^{-1}\log_s\left(\sqrt{2}b \epsilon^{-1}\right) + 1)} = (\sqrt{2}b \epsilon^{-1})^{ (\gamma - 2\alpha)/\alpha} \cdot s^{\gamma - 2\alpha} $ and $\sum_{\ell=0}^{L}s^{\gamma \ell} = \mathcal{O}(\epsilon^{-\gamma/\alpha})$, $C$ is $\mathcal{O}(\epsilon^{-2-(\gamma-2\alpha)/\alpha}) = \mathcal{O}(\epsilon^{-\gamma / \alpha})$. 
   \end{enumerate}

   \end{proof} 
    
   \subsection{Proof of Theorem \ref{thm:computation_SMC}}\label{app:proof_smc}
   
    \begin{proof}
    For the SMC metamodeling estimator, consider it as an MLMC metamodeling estimator derived using one design level $\ell=0$. Based on  Condition 2, we have
    \[\left\|\Var{\Vhatsub{0}{\cdot, M_0}}\right\|_{1} \leq \bar{v} / (M_0-1) < \frac{\epsilon^2}{2} \text{,  and hence } M_0 =\mathcal{O}(\epsilon^{-2}) \ . \]
    To make the bias component less than $\epsilon^2/2$, the number of design points should be the same as that in the finest design level, i.e.,
    $N_{0} = c s^{\gamma L} = c s^{\gamma \log_s\left(\sqrt{2}b \epsilon^{-1}\right)/\alpha} = \mathcal{O}(\epsilon^{-\gamma/\alpha})$. Hence, the total cost  is  $C = M_0 N_0 = \mathcal{O}(\epsilon^{-(2 + \gamma / \alpha)})$. 
\end{proof}

\section{Proofs in Section \ref{sec:clt}}
We first provide some technical results in Subsection \ref{app:lemma}, which will be useful for our later proofs. Following that, we present the proofs for Section \ref{sec:clt} in the remaining subsections.

 We begin by recalling the definition of convergence in probability, a concept that will be used in our analysis (page 56, \cite{durrett2019probability}).

\begin{definition}
	Let $X_1, X_2, \dots, X_n, \ldots$  and $X$ be real-valued random variables. The sequence $\{X_n, n\geq 1\}$ is said to converge to $X$ in probability if, for all $\epsilon > 0$, $\mathbb{P}(|X_n-X| > \epsilon) \rightarrow 0$ as $n \rightarrow \infty$, in which we write as $X_n \converge{p} X$ as $n \rightarrow \infty$.
\end{definition}

\subsection{Auxiliary Lemmas}\label{app:lemma}
Let $\lesssim$ and $\gtrsim$  denote inequalities up to a constant multiple, and  write $a \asymp b$ to indicate that both $a \lesssim b$ and $a \gtrsim b$ hold.
\begin{lemma}\label{lem:Vl_compare}
    For $\forall \ell \in [L_{\epsilon}]$, define $V_{\ell, \epsilon}^{Z}(\theta) \coloneqq M_{\ell,\epsilon} \Var{\deltaZhatsub{\ell}{\ell}}$. Then, $V_{\ell,\epsilon}^{Z}(\theta) \asymp 1$ and $V_{\ell,\epsilon} \asymp 1$ as $\epsilon \rightarrow 0$.
\end{lemma}
\begin{proof}
 Proving Lemma \ref{lem:Vl_compare} is equivalent to showing $\Var{ \DeltaVhatsub{\ell}{\theta, M_{\ell,\epsilon}, \varpi_{\ell}} } \asymp M_{\ell,\epsilon}^{-1}$  and $\Var{\deltaZhatsub{\ell}{\ell}} \asymp M_{\ell,\epsilon}^{-1}$.
We first expand $\Var{ \DeltaVhatsub{\ell}{\theta, M_{\ell,\epsilon}, \varpi_{\ell}} }$ as follows:
 \begin{align*}
 	&\Var{ \DeltaVhatsub{\ell}{\theta, M_{\ell,\epsilon}, \varpi_{\ell}} } = \operatorname{Var}\bigg(\sum_{i=1}^{N_{\ell}}w_{i}^{\ell}(\theta) \cV(\theta_i^{\ell},M_{\ell,\epsilon}, \varpi_{\ell}) - \sum_{i=1}^{N_{\ell-1}}w_{i}^{\ell-1} (\theta)\cV(\theta_i^{\ell-1},M_{\ell,\epsilon}, \varpi_{\ell})\bigg) \\
 	=& \sum_{i=1}^{N_{\ell}} (w_{i}^{\ell}(\theta))^2\Var{ \cV(\theta_i^{\ell},M_{\ell,\epsilon}, \varpi_{\ell})} + \sum_{i=1}^{N_{\ell-1}} (w_{i}^{\ell-1}(\theta))^2\Var{ \cV(\theta_i^{\ell-1},M_{\ell,\epsilon}, \varpi_{\ell})} \\
    &+ \sum_{i=1}^{N_{\ell}}\sum_{j=1, j \neq i}^{N_{\ell}}w_{i}^{\ell}(\theta)w_{j}^{\ell}(\theta)\Cov{ \cV(\theta_i^{\ell},M_{\ell,\epsilon}, \varpi_{\ell})}{ \cV(\theta_j^{\ell},M_{\ell,\epsilon}, \varpi_{\ell})}\\
    &+  \sum_{i=1}^{N_{\ell-1}}\sum_{j=1, j \neq i}^{N_{\ell-1}}w_{i}^{\ell-1}(\theta)w_{j}^{\ell-1}(\theta)\Cov{ \cV(\theta_i^{\ell-1},M_{\ell,\epsilon}, \varpi_{\ell})}{\cV(\theta_j^{\ell-1},M_{\ell,\epsilon}, \varpi_{\ell})}\\
    &- 2 \sum_{i=1}^{N_{\ell}}\sum_{j=1}^{N_{\ell-1}}w_{i}^{\ell}(\theta)w_{j}^{\ell-1}(\theta)\Cov{ \cV(\theta_i^{\ell},M_{\ell,\epsilon}, \varpi_{\ell})}{ \cV(\theta_j^{\ell-1},M_{\ell,\epsilon}, \varpi_{\ell})}  \  .
 \end{align*}
 Similarly, we expand $\Var{\deltaZhatsub{\ell}{\ell}}$ as follows:
  \begin{align*}
 	&\Var{\deltaZhatsub{\ell}{\ell}} = \operatorname{Var}\bigg(\sum_{i=1}^{N_{\ell}}w_{i}^{\ell}(\theta) \overline{Z^2}(\theta_i^{\ell}, \varpi_{\ell}) - \sum_{i=1}^{N_{\ell-1}}w_{i}^{\ell-1}(\theta) \overline{Z^2}(\theta_i^{\ell-1}, \varpi_{\ell}) \bigg) \\
 	=& \sum_{i=1}^{N_{\ell}} (w_{i}^{\ell}(\theta))^2\Var{\overline{Z^2}(\theta_i^{\ell}, \varpi_{\ell})} + \sum_{i=1}^{N_{\ell-1}} (w_{i}^{\ell-1}(\theta))^2\Var{\overline{Z^2}(\theta_i^{\ell-1}, \varpi_{\ell})} \\
    &+ \sum_{i=1}^{N_{\ell}}\sum_{j=1, j \neq i}^{N_{\ell}}w_{i}^{\ell}(\theta)w_{j}^{\ell}(\theta)\Cov{\overline{Z^2}(\theta_i^{\ell}, \varpi_{\ell})}{\overline{Z^2}(\theta_j^{\ell}, \varpi_{\ell})} \\
    &+ \sum_{i=1}^{N_{\ell-1}}\sum_{j=1,  j \neq i}^{N_{\ell-1}}w_{i}^{\ell-1}(\theta)w_{j}^{\ell-1}(\theta)\Cov{\overline{Z^2}(\theta_i^{\ell-1}, \varpi_{\ell})}{\overline{Z^2}(\theta_j^{\ell-1}, \varpi_{\ell})} \\
    &- 2 \sum_{i=1}^{N_{\ell}}\sum_{j=1}^{N_{\ell-1}}w_{i}^{\ell}(\theta)w_{j}^{\ell-1}(\theta)\Cov{\overline{Z^2}(\theta_i^{\ell}, \varpi_{\ell})}{\overline{Z^2}(\theta_j^{\ell-1}, \varpi_{\ell})}   \ .
 \end{align*}
 We next analyze the covariance terms $\Cov{ \cV(\theta_i^{\ell},M_{\ell,\epsilon}, \varpi_{\ell})}{ \cV(\theta_j^{\ell-1},M_{\ell,\epsilon}, \varpi_{\ell})}$. Notice that, for a given design point $\theta_i^{\ell}$, the simulation outputs $\cY_m(\theta_i^{\ell}, \varpi_{\ell})$ and $\cY_n(\theta_i^{\ell}, \varpi_{\ell})$ are independent when $m \neq n$. Furthermore, $\cY_m(\theta_i^{\ell}, \varpi_{\ell})$ is independent of $\cY_n(\theta_j^{\ell-1}, \varpi_{\ell})$ with any design point $\theta_j^{\ell-1}$ on level $\ell-1$. Based on these properties, we can derive the following expansion:
\begin{align*}
	&\Cov{ \cV(\theta_i^{\ell},M_{\ell,\epsilon}, \varpi_{\ell})}{ \cV(\theta_j^{\ell-1},M_{\ell,\epsilon}, \varpi_{\ell})}\\
	 =& \frac{1}{M_{\ell,\epsilon}}  \Cov{\cY_1^2(\theta_i^{\ell}, \varpi_{\ell})}{\cY_2^2(\theta_j^{\ell-1}, \varpi_{\ell})} - \frac{2}{M_{\ell,\epsilon}} \Cov{\cY_1^2(\theta_i^{\ell}, \varpi_{\ell})}{\cY_1(\theta_j^{\ell-1}, \varpi_{\ell})\cY_2(\theta_j^{\ell-1}, \varpi_{\ell})} \\
    &- \frac{2}{M_{\ell,\epsilon}} \Cov{\cY_1^2(\theta_j^{\ell-1}, \varpi_{\ell})}{\cY_1(\theta_i^{\ell}, \varpi_{\ell})\cY_2(\theta_i^{\ell}, \varpi_{\ell})} \\
    & + \frac{2}{M_{\ell,\epsilon}(M_{\ell,\epsilon}-1)}\Cov{\cY_1(\theta_i^{\ell}, \varpi_{\ell})\cY_2(\theta_i^{\ell}, \varpi_{\ell})}{\cY_1(\theta_j^{\ell-1}, \varpi_{\ell})\cY_2(\theta_j^{\ell-1}, \varpi_{\ell})} \\
    & + \frac{4(M_{\ell,\epsilon}-2)}{M_{\ell,\epsilon}(M_{\ell,\epsilon}-1)} \Cov{\cY_1(\theta_i^{\ell}, \varpi_{\ell})\cY_2(\theta_i^{\ell}, \varpi_{\ell})}{\cY_1(\theta_j^{\ell-1}, \varpi_{\ell})\cY_3(\theta_j^{\ell-1}, \varpi_{\ell})} \ .
\end{align*}
Similarly, we expand $\Cov{\overline{Z^2}(\theta_i^{\ell}, \varpi_{\ell})}{\overline{Z^2}(\theta_j^{\ell-1}, \varpi_{\ell})}$ as follows:
\begin{align*}
	&\Cov{\overline{Z^2}(\theta_i^{\ell}, \varpi_{\ell})}{\overline{Z^2}(\theta_j^{\ell-1}, \varpi_{\ell})} \\
	=& \frac{1}{M_{\ell,\epsilon}}\E{\cY_1^2(\theta_i^{\ell}, \varpi_{\ell}) \cY_1^2(\theta_j^{\ell-1}, \varpi_{\ell})} - \frac{1}{M_{\ell,\epsilon}}\E{\cY_1^2(\theta_i^{\ell}, \varpi_{\ell})}\E{\cY_1^2(\theta_j^{\ell-1}, \varpi_{\ell})} \\
 & - \frac{2}{M_{\ell,\epsilon}}\E{\cY_1^2(\theta_i^{\ell}, \varpi_{\ell}) \cY_1(\theta_j^{\ell-1}, \varpi_{\ell})}\E{\cY_1(\theta_j^{\ell-1}, \varpi_{\ell})}  + \frac{2}{M_{\ell,\epsilon}}\E{\cY_1^2(\theta_i^{\ell}, \varpi_{\ell})}\Esup{2}{\cY_1(\theta_j^{\ell-1}, \varpi_{\ell})} \\
 & + \frac{4}{M_{\ell,\epsilon}}\E{\cY_1(\theta_i^{\ell}, \varpi_{\ell}) \cY_1(\theta_j^{\ell-1}, \varpi_{\ell})}\E{\cY_1(\theta_i^{\ell}, \varpi_{\ell})}\E{\cY_1(\theta_j^{\ell-1}, \varpi_{\ell})} \\
 & - \frac{2}{M_{\ell,\epsilon}}\E{\cY_1(\theta_i^{\ell}, \varpi_{\ell}) \cY_1^2(\theta_j^{\ell-1}, \varpi_{\ell})} \E{\cY_1(\theta_i^{\ell}, \varpi_{\ell})} +\frac{2}{M_{\ell,\epsilon}}\Esup{2}{\cY_1(\theta_i^{\ell}, \varpi_{\ell})}\E{\cY_1^2(\theta_j^{\ell-1}, \varpi_{\ell})} \\
 &- \frac{4}{M_{\ell,\epsilon}}\E{\cY_1(\theta_i^{\ell}, \varpi_{\ell})}^2 \Esup{2}{\cY_1(\theta_j^{\ell-1}, \varpi_{\ell})}  \ .
\end{align*}
Therefore, both $\operatorname{Cov}\big( \cV(\theta_i^{\ell},M_{\ell,\epsilon}, \varpi_{\ell}), \cV(\theta_j^{\ell-1},M_{\ell,\epsilon}, \varpi_{\ell}) \big)$ and $\operatorname{Cov}\big(\overline{Z^2}(\theta_i^{\ell}, \varpi_{\ell}), \overline{Z^2}(\theta_j^{\ell-1}, \varpi_{\ell}) \big)$ diminish at a rate of order $ M_{\ell, \epsilon}^{-1}$. Similarly, we can show that $\operatorname{Cov}\big( \cV(\theta_i^{\ell},M_{\ell,\epsilon}, \varpi_{\ell}), \allowbreak  \cV(\theta_j^{\ell},M_{\ell,\epsilon}, \varpi_{\ell}) \big)$, $\operatorname{Cov}\big( \cV(\theta_i^{\ell-1},M_{\ell,\epsilon}, \varpi_{\ell}), \cV(\theta_j^{\ell-1},M_{\ell,\epsilon}, \varpi_{\ell}) \big)$, $\operatorname{Cov}\big(\overline{Z^2}(\theta_i^{\ell}, \varpi_{\ell}), \overline{Z^2}(\theta_j^{\ell}, \varpi_{\ell}) \big)$, and $\operatorname{Cov}\big(\overline{Z^2}(\theta_i^{\ell-1}, \varpi_{\ell}), \allowbreak  \overline{Z^2}(\theta_j^{\ell-1}, \varpi_{\ell}) \big)$ diminish at a rate of order $ M_{\ell, \epsilon}^{-1}$.

For the variance terms, at any design point $\theta_i^{\ell}$ on level $\ell$, we have
\[\Var{\cV(\theta_i^{\ell},M_{\ell,\epsilon}, \varpi_{\ell})} = \frac{\E{\left(\cY_1(\theta_i^{\ell}, \varpi_{\ell}) - \E{\cY_1(\theta_i^{\ell}, \varpi_{\ell})} \right)^4}}{M_{\ell,\epsilon}} -\frac{\Varsup{2}{\cY_1(\theta_i^{\ell}, \varpi_{\ell})}(M_{\ell,\epsilon}-3)}{M_{\ell,\epsilon}(M_{\ell,\epsilon}-1)} \ , \]
and
\[\Var{\overline{Z^2}(\theta_i^{\ell}, \varpi_{\ell})} = \frac{\E{\left(\cY_1(\theta_i^{\ell}, \varpi_{\ell}) - \E{\cY_1(\theta_i^{\ell}, \varpi_{\ell})} \right)^4}}{M_{\ell,\epsilon}} -\frac{\Varsup{2}{\cY_1(\theta_i^{\ell}, \varpi_{\ell})}}{M_{\ell,\epsilon}} \ .\]
Both terms diminish at a rate  of order $M_{\ell, \epsilon}^{-1}$. Hence, $\Var{ \DeltaVhatsub{\ell}{\theta, M_{\ell,\epsilon}, \varpi_{\ell}} } \asymp M_{\ell,\epsilon}^{-1}$ and $\Var{\deltaZhatsub{\ell}{\ell}} \asymp M_{\ell,\epsilon}^{-1}$. 
\end{proof}

\begin{lemma}\label{lem:variance}
 Assume that $\lim_{\epsilon \rightarrow 0} S_{L_{\epsilon}} < \infty$. There exist  $0< c_{lb} \leq c_{ub}$ such that $c_{lb} \leq \lim_{\epsilon \rightarrow 0} \epsilon^{-2}\Var{\zhat} \leq c_{ub}.$
\end{lemma}

\begin{proof}
By Lemma \ref{lem:Vl_compare}, there exists a positive constant $c_{ub} < \infty$ such that $V_{\ell,\epsilon}^{Z} / (2V_{\ell,\epsilon}) \leq c_{ub}$. For any $\epsilon > 0$,  we have
	\[\frac{\Var{\zhat}}{\epsilon^2} = \sum_{\ell=0}^{L_{\epsilon}}\frac{V_{\ell,\epsilon}^{Z}}{\epsilon^2 M_{\ell,\epsilon}} \leq  \sum_{\ell=0}^{L_{\epsilon}}\frac{2c_{ub}V_{\ell,\epsilon}}{\epsilon^2 M_{\ell,\epsilon}} \leq c_{ub}\sum_{\ell=0}^{L_{\epsilon}}\frac{\sqrt{V_{\ell,\epsilon}N_{\ell}}}{S_{L_{\epsilon}}}=c_{ub} \ ,\]
where the last inequality follows by recalling the definition $M_{\ell, \epsilon} \coloneqq \ceil[\Big]{ 2 \epsilon^{-2} \sqrt{ V_{\ell,\epsilon} / N_{\ell}} S_{L_{\epsilon}} }$.
Define $c_{lb} \coloneqq \inf_{\ell} V_{\ell,\epsilon}^{Z} / (2V_{\ell,\epsilon})$. Since $\lim_{\epsilon \rightarrow 0}S_{L_{\epsilon}} < \infty$, there exists $k > 1$ with $\gamma / k < 2\alpha$ such that
	\begin{align*}
		\frac{\Var{\zhat}}{\epsilon^2} &=  \sum_{\ell=0}^{L_{\epsilon}}\frac{V_{\ell,\epsilon}^{Z}}{\epsilon^2 M_{\ell,\epsilon}}  \geq \sum_{\ell=0}^{\lceil L_{\epsilon}/k \rceil}\frac{V_{\ell,\epsilon}^{Z}}{\epsilon^2 M_{\ell,\epsilon}} \geq  \sum_{\ell=0}^{\lceil L_{\epsilon}/k \rceil}\frac{V_{\ell,\epsilon}^{Z}}{2\epsilon^2 + 2\sqrt{V_{\ell,\epsilon}/N_{\ell}}S_{L_{\epsilon}}} \\
		&\geq c_{lb}  \sum_{\ell=0}^{\lceil L_{\epsilon}/k \rceil}\frac{V_{\ell,\epsilon}}{\epsilon^2 + \sqrt{V_{\ell,\epsilon}/N_{\ell}}S_{L_{\epsilon}}} \\
		& \geq c_{lb} \left( \frac{S_{\lceil L_{\epsilon}/k \rceil}}{S_{L_{\epsilon}}}  - \epsilon^2 \sum_{\ell=0}^{\lceil L_{\epsilon}/k \rceil} \frac{N_{\ell}}{S_{L_{\epsilon}}^2} \right) \ .
	\end{align*}
	By the mean value theorem, there exists a constant $C > 0$ such that $\lim_{\epsilon \rightarrow 0}\epsilon^2 \sum_{\ell=0}^{\lceil L_{\epsilon}/k \rceil} N_{\ell}/S_{L_{\epsilon}}^2 \leq \lim_{\epsilon \rightarrow 0} C \epsilon^2 2^{\gamma L_{\epsilon} / k} / S_{L_{\epsilon}}^2 = 0$. And $\lim_{\epsilon \rightarrow 0}S_{\lceil L_{\epsilon}/k \rceil}/S_{L_{\epsilon}} =1 $ leads to $\lim_{\epsilon \rightarrow 0} \Var{\zhat}/\epsilon^2 \geq c_{lb}$.
\end{proof}

\subsection{Proof of Proposition \ref{prop:clt_single_level}}\label{proof:clt_single_level}
\begin{proof}
Consider a fixed $\ell \in [L_{\epsilon}]$. By decomposing the  single level estimator $\Vhatsub{\ell}{\theta, M_{\ell,\epsilon}, \varpi_{\ell}}$, we have
\begin{align*}
	&\Vhatsub{\ell}{\theta, M_{\ell,\epsilon}, \varpi_{\ell}} - \E{\Vhatsub{\ell}{\theta, M_{\ell,\epsilon}, \varpi_{\ell}}} \\ 
	= & \underbrace{\sum_{i=1}^{N_{\ell}}w_{i}^{\ell}(\theta) \frac{1}{M_{\ell,\epsilon}}\cV(\theta_i^{\ell},M_{\ell,\epsilon}, \varpi_{\ell})}_{\text{(B.1)}} 
	- \underbrace{\sum_{i=1}^{N_{\ell}}w_{i}^{\ell}(\theta) (\overline{Z}(\theta_i^{\ell}, \varpi_{\ell}))^2}_{\text{(B.2)}} 
	+ \underbrace{\sum_{i=1}^{N_{\ell}}w_{i}^{\ell}(\theta) \overline{Z^2}(\theta_i^{\ell}, \varpi_{\ell}) 
	-  \E{\Vhatsub{\ell}{\theta, M_{\ell,\epsilon}, \varpi_{\ell}}} }_{\text{(B.3)}} \  .
\end{align*}
Regarding the term (B.1), for any $i \in [N_{\ell}]^{+}$, since $\cV(\theta_i^{\ell},M_{\ell,\epsilon}, \varpi_{\ell}) \converge{p} \V{\theta_i^{\ell}}$ as $\epsilon \rightarrow 0$, we have
\[\sum_{i=1}^{N_{\ell}}w_{i}^{\ell}(\theta) \cdot \cV(\theta_i^{\ell},M_{\ell,\epsilon}, \varpi_{\ell}) \cdot\frac{1}{M_{\ell,\epsilon}} \converge{p}  \sum_{i=1}^{N_{\ell}}w_{i}^{\ell}(\theta) \cdot \V{\theta_i^{\ell}} \cdot 0 =0 \text{ as }  \epsilon \rightarrow 0 \ .\]
Regarding the term (B.2), since $\overline{Z}(\theta_i^{\ell}, \varpi_{\ell}) \converge{p} 0$ as $\epsilon \rightarrow 0$, we have $\sum_{i=1}^{N_{\ell}}w_{i}^{\ell}(\theta) (\overline{Z}(\theta_i^{\ell}, \varpi_{\ell}))^2 \converge{p} 0$ as $\epsilon \rightarrow 0$ by the continuous mapping theorem. 
For the term (B.3), we have
 \begin{align*}
 	\E{\Vhatsub{\ell}{\theta, M_{\ell,\epsilon}, \varpi_{\ell}}} = \sum_{i=1}^{N_{\ell}}w_{i}^{\ell}(\theta) \E{ \cV(\theta_i^{\ell},M_{\ell,\epsilon}, \varpi_{\ell}) } =  \sum_{i=1}^{N_{\ell}}w_{i}^{\ell}(\theta)  \Var{\cZ_1(\theta_i^{\ell}, \varpi_{\ell})} \ ,
 \end{align*}
 and
 \[\sum_{i=1}^{N_{\ell}}w_{i}^{\ell}(\theta)\E{\overline{Z^2}(\theta_i^{\ell}, \varpi_{\ell})} = \sum_{i=1}^{N_{\ell}}w_{i}^{\ell}(\theta)\E{\frac{1}{M_{\ell,\epsilon}}\sum_{m=1}^{M_{\ell,\epsilon}} (\cZ_{m}(\theta_i^{\ell}, \varpi_{\ell}))^2} =\sum_{i=1}^{N_{\ell}}w_{i}^{\ell}(\theta)\Var{\cZ_1(\theta_i^{\ell}, \varpi_{\ell})} \ .\]
 The CLT implies that as $\epsilon \rightarrow 0$,
 \[\sum_{i=1}^{N_{\ell}}w_{i}^{\ell}(\theta) \left(\overline{Z^2}(\theta_i^{\ell}, \varpi_{\ell})\right) -  \E{ \Vhatsub{\ell}{\theta, M_{\ell,\epsilon}, \varpi_{\ell}}}  \Longrightarrow \mathcal{N}\left(0, \operatorname{Var}\bigg(\sum_{i=1}^{N_{\ell}}w_{i}^{\ell}(\theta) \overline{Z^2}(\theta_i^{\ell}, \varpi_{\ell}) \bigg) \right) \ . \]
By Slutsky's theorem, we have  as $\epsilon \rightarrow 0$,
 \[ \Vhatsub{\ell}{\theta, M_{\ell,\epsilon}, \varpi_{\ell}} - \E{ \Vhatsub{\ell}{\theta, M_{\ell,\epsilon}, \varpi_{\ell}}} \Longrightarrow \mathcal{N}\left(0, \operatorname{Var}\bigg(\sum_{i=1}^{N_{\ell}}w_{i}^{\ell}(\theta) \overline{Z^2}(\theta_i^{\ell}, \varpi_{\ell})\bigg) \right)   \ . \]
 Finally, it can be shown that $\Var{\sum_{i=1}^{N_{\ell}}w_{i}^{\ell}(\theta) \overline{Z^2}(\theta_i^{\ell}, \varpi_{\ell})}  = \Var{\sum_{i=1}^{N_{\ell}}w_{i}^{\ell}(\theta) \cY_1^2(\theta_i^{\ell}, \varpi_{\ell})}  / M_{\ell,\epsilon}$ by the definition of $\overline{Z^2}(\theta_i^{\ell}, \varpi_{\ell})$ for $i \in [N_{\ell}]^{+}$. The proof is complete. 
\end{proof}

\subsection{Proof of Proposition \ref{prop:clt_refine}}\label{proof:clt_refine}
\begin{proof}
Consider a fixed $\ell \in [L_{\epsilon}]$.	We have the following decomposition:
\begin{align*}
	&\DeltaVhatsub{\ell}{\theta, M_{\ell,\epsilon}, \varpi_{\ell}}  - \E{\DeltaVhatsub{\ell}{\theta, M_{\ell,\epsilon}, \varpi_{\ell}} } \nonumber \\ = & \underbrace{ \sum_{i=1}^{N_{\ell}}w_{i}^{\ell}(\theta) \frac{1}{M_{\ell,\epsilon}} \cV(\theta_i^{\ell},M_{\ell,\epsilon}, \varpi_{\ell}) - \sum_{i=1}^{N_{\ell-1}}w_{i}^{\ell-1}(\theta) \frac{1}{M_{\ell,\epsilon}} \cV(\theta_i^{\ell-1},M_{\ell,\epsilon}, \varpi_{\ell})}_{\text{(B.4)}} \\
	& - \underbrace{ \left(\sum_{i=1}^{N_{\ell}}w_{i}^{\ell}(\theta) (\overline{Z}(\theta_i^{\ell}, \varpi_{\ell}))^2 - \sum_{i=1}^{N_{\ell-1}}w_{i}^{\ell-1}(\theta) (\overline{Z}(\theta_i^{\ell-1}, \varpi_{\ell}))^2 \right) }_{\text{(B.5)}} \\
	& + \underbrace{ \sum_{i=1}^{N_{\ell}}w_{i}^{\ell}(\theta) \overline{Z^2}(\theta_i^{\ell}, \varpi_{\ell}) - \sum_{i=1}^{N_{\ell-1}}w_{i}^{\ell-1}(\theta) \overline{Z^2}(\theta_i^{\ell-1}, \varpi_{\ell}) -  \E{\DeltaVhatsub{\ell}{\theta, M_{\ell, \epsilon}, \varpi_{\ell}} } }_{\text{(B.6)}}  \ .
\end{align*}
Based on the proof of Proposition \ref{prop:clt_single_level} in Appendix \ref{proof:clt_single_level}, we can easily show that both  (B.4) and (B.5) converge to 0 in probability as $\epsilon \rightarrow 0$.
For the term (B.6),   we have
\begin{align*}
	\E{\DeltaVhatsub{\ell}{\theta, M_{\ell,\epsilon}, \varpi_{\ell}} } &= \sum_{i=1}^{N_{\ell}}w_{i}^{\ell}(\theta) \E{ \cV(\theta_i^{\ell},M_{\ell,\epsilon}, \varpi_{\ell}) } - \sum_{i=1}^{N_{\ell-1}}w_{i}^{\ell-1}(\theta) \E{ \cV(\theta_i^{\ell-1},M_{\ell,\epsilon}, \varpi_{\ell}) } \\
	&= \sum_{i=1}^{N_{\ell}}w_{i}^{\ell}(\theta) \Var{\cZ_1(\theta_i^{\ell}, \varpi_{\ell})  } - \sum_{i=1}^{N_{\ell-1}}w_{i}^{\ell-1}(\theta) \Var{ \cZ_1(\theta_i^{\ell-1}, \varpi_{\ell}) } \ .
\end{align*}
It follows that
\begin{align*}
	&\E{ \sum_{i=1}^{N_{\ell}}w_{i}^{\ell}(\theta) \overline{Z^2}(\theta_i^{\ell}, \varpi_{\ell}) - \sum_{i=1}^{N_{\ell-1}}w_{i}^{\ell-1}(\theta) \overline{Z^2}(\theta_i^{\ell-1}, \varpi_{\ell})} \\ = &\sum_{i=1}^{N_{\ell}}w_{i}^{\ell}(\theta) \Var{\cZ_1(\theta_i^{\ell}, \varpi_{\ell})}  - \sum_{i=1}^{N_{\ell-1}}w_{i}^{\ell-1}(\theta) \Var{ \cZ_1(\theta_i^{\ell-1}, \varpi_{\ell}) } = \E{\DeltaVhatsub{\ell}{\theta, M_{\ell,\epsilon}, \varpi_{\ell}} } \ .
\end{align*}
 Hence, the CLT implies that  as $\epsilon \rightarrow 0$,
 \begin{align*}
 	&\sum_{i=1}^{N_{\ell}}w_{i}^{\ell}(\theta) \overline{Z^2}(\theta_i^{\ell}, \varpi_{\ell}) - \sum_{i=1}^{N_{\ell-1}}w_{i}^{\ell-1}(\theta) \overline{Z^2}(\theta_i^{\ell-1}, \varpi_{\ell}) -  \E{\DeltaVhatsub{\ell}{\theta, M_{\ell,\epsilon}, \varpi_{\ell}} } \\  
 	&\qquad \Longrightarrow \mathcal{N}\bigg(0,  \underbrace{ \Var{\sum_{i=1}^{N_{\ell}}w_{i}^{\ell}(\theta) \overline{Z^2}(\theta_i^{\ell}, \varpi_{\ell}) - \sum_{i=1}^{N_{\ell-1}}w_{i}^{\ell-1}(\theta) \overline{Z^2}(\theta_i^{\ell-1}, \varpi_{\ell})} }_{\text{(B.7)}} \bigg)  \ . 
 \end{align*}
 By the definitions of $\overline{Z^2}(\theta_i^{\ell}, \varpi_{\ell})$ and $\overline{Z^2}(\theta_i^{\ell-1}, \varpi_{\ell})$, the asymptotic variance given in (B.7) can be rewritten as
 \[\Var{\sum_{i=1}^{N_{\ell}}w_{i}^{\ell}(\theta) \cY_1^2(\theta_i^{\ell}, \varpi_{\ell}) - \sum_{i=1}^{N_{\ell-1}}w_{i}^{\ell-1}(\theta) \cY_1^2(\theta_i^{\ell}, \varpi_{\ell}) } / M_{\ell, \epsilon} \ .\]
Hence, as $\epsilon \rightarrow 0$,
 \begin{align*}
 	&\sqrt{M_{\ell, \epsilon}} \left( \DeltaVhatsub{\ell}{\theta, M_{\ell,\epsilon}, \varpi_{\ell}}  - \E{\DeltaVhatsub{\ell}{\theta, M_{\ell,\epsilon}, \varpi_{\ell}} } \right) \\
 	&\qquad \Longrightarrow \mathcal{N}\left(0, \Var{\sum_{i=1}^{N_{\ell}}w_{i}^{\ell}(\theta) \cY_1^2(\theta_i^{\ell}, \varpi_{\ell}) - \sum_{i=1}^{N_{\ell-1}}w_{i}^{\ell-1}(\theta) \cY_1^2(\theta_i^{\ell-1}, \varpi_{\ell}) } \right) \ . 
 \end{align*}

\end{proof}

\subsection{Proof of Proposition \ref{prop:clt_z_v}}\label{proof:clt_z_v}
\begin{proof}
	To show the asymptotic normality of MLMC metamodeling estimator $\Vhat{\theta}$, we first note that 
 
 \begin{align*}
 	\frac{\Vhat{\theta} - \E{\Vhat{\theta}}}{\sqrt{\Var{\zhat}}}  =& \frac{ \sum_{\ell=0}^{L_{\epsilon}}\left(\DeltaVhatsub{\ell}{\theta, M_{\ell,\epsilon}, \varpi_{\ell}}  -\E{\DeltaVhatsub{\ell}{\theta, M_{\ell,\epsilon}, \varpi_{\ell}} } \right)}{\sqrt{\Var{\zhat}}} \nonumber\\
		=& \underbrace{ \Var{\zhat}^{-\frac{1}{2}}\sum_{\ell=0}^{L_{\epsilon}} \frac{1}{M_{\ell,\epsilon}} \DeltaVhatsub{\ell}{\theta, M_{\ell,\epsilon}, \varpi_{\ell}} }_{\text{(B.8)}} \\ 
	& - \underbrace{ \Var{\zhat}^{-\frac{1}{2}}  \sum_{\ell=0}^{L_{\epsilon}}\left(\sum_{i=1}^{N_{\ell}}w_{i}^{\ell}(\theta) (\overline{Z}(\theta_i^{\ell}, \varpi_{\ell}))^2 - \sum_{i=1}^{N_{\ell-1}}w_{i}^{\ell-1}(\theta) (\overline{Z}(\theta_i^{\ell-1}, \varpi_{\ell-1}))^2 \right) }_{\text{(B.9)}}\\
	& + \underbrace{ \Var{\zhat}^{-\frac{1}{2}} \left(\zhat - \E{\zhat} \right) }_{\text{(B.10)}} \ .
  \end{align*}
  We first show that the term (B.8) $ \converge{p} 0$ as $\epsilon \rightarrow 0$. By the continuous mapping theorem, it is sufficient to show that
  	\[ \frac{\left( \sum_{\ell=0}^{L_{\epsilon}}\frac{1}{M_{\ell,\epsilon}}  \DeltaVhatsub{\ell}{\theta, M_{\ell,\epsilon}, \varpi_{\ell}} \right)^2}{\sum_{\ell=0}^{L_{\epsilon}} \frac{1}{M_{\ell,\epsilon}} \Var{  \deltaZhatsubshort{(1)}{\ell}{\ell}  }} \converge{p} 0 \text{ as } \epsilon \rightarrow 0 \ .\]
  By Titu's lemma, we have
  	\begin{align*}
  		\frac{\left( \sum_{\ell=0}^{L_{\epsilon}}\frac{1}{M_{\ell,\epsilon}}  \DeltaVhatsub{\ell}{\theta, M_{\ell,\epsilon}, \varpi_{\ell}} \right)^2}{\sum_{\ell=0}^{L_{\epsilon}} \frac{1}{M_{\ell,\epsilon}} \Var{  \deltaZhatsubshort{(1)}{\ell}{\ell} }} &\leq \sum_{\ell=0}^{L_{\epsilon}} \frac{\left( \frac{1}{M_{\ell,\epsilon}}  \DeltaVhatsub{\ell}{\theta, M_{\ell,\epsilon}, \varpi_{\ell}} \right)^2}{\frac{1}{M_{\ell,\epsilon}} \Var{  \deltaZhatsubshort{(1)}{\ell}{\ell} }} \converge{p} 0 \text{ as } \epsilon \rightarrow 0 \ .
  	\end{align*}
To see how the last step follows, we note that for any $\delta > 0$, as $\epsilon \rightarrow 0$, it follows that
   \begin{align*}
       &\mathbb{P}\left\{\left| \sum_{\ell=0}^{L_{\epsilon}} \frac{1}{M_{\ell,\epsilon}}  \frac{\left(   \DeltaVhatsub{\ell}{\theta, M_{\ell,\epsilon}, \varpi_{\ell}} \right)^2}{ \Var{  \deltaZhatsubshort{(1)}{\ell}{\ell} }} \right| > \delta \right\} \leq \frac{1}{\delta} \cdot \E{ \sum_{\ell=0}^{L_{\epsilon}} \frac{1}{M_{\ell,\epsilon}}  \frac{\left(  \DeltaVhatsub{\ell}{\theta, M_{\ell,\epsilon}, \varpi_{\ell}}  \right)^2}{ \Var{  \deltaZhatsubshort{(1)}{\ell}{\ell} }}}  \\
      = &  \frac{1}{\delta}  \sum_{\ell=0}^{L_{\epsilon}}\frac{1}{M_{\ell,\epsilon}} \cdot \frac{ \left( \E{ \DeltaVhatsub{\ell}{\theta, M_{\ell,\epsilon}, \varpi_{\ell}}} \right)^2 + \Var{ \DeltaVhatsub{\ell}{\theta, M_{\ell,\epsilon}, \varpi_{\ell}} }}{ \Var{ \deltaZhatsubshort{(1)}{\ell}{\ell} } }  \\
      = & \frac{1}{\delta}   \underbrace{ \sum_{\ell=0}^{L_{\epsilon}}\frac{1}{M_{\ell,\epsilon}} \cdot \frac{ \left( \E{ \DeltaVhatsub{\ell}{\theta, M_{\ell,\epsilon}, \varpi_{\ell}}} \right)^2}{ \Var{ \deltaZhatsubshort{(1)}{\ell}{\ell} } } }_{\mbox{(B.11)}} +  \frac{1}{\delta}   \underbrace{ \sum_{\ell=0}^{L_{\epsilon}}\frac{1}{M_{\ell,\epsilon} } \cdot \frac{(M_{\ell,\epsilon}-1)^{-1} v_{\ell, \epsilon} }{ \Var{ \deltaZhatsubshort{(1)}{\ell}{\ell} } } }_{\mbox{(B.12)}} \  .
   \end{align*} 
   For the term (B.11), we first note that $\left( \E{ \DeltaVhatsub{\ell}{\theta, M_{\ell,\epsilon}, \varpi_{\ell}}} \right)^2 / \Var{ \deltaZhatsubshort{(1)}{\ell}{\ell} }$ is a constant. Furthermore, since $\Vhat{\theta}$ is an \textit{$\epsilon^2$-estimator}, we have 
   $L_{\epsilon} \leq \alpha^{-1} \log_s(\sqrt{2}b\epsilon^{-1}) + 1 \mbox{ and } M_{\ell, \epsilon} \geq \epsilon^{-2} \cdot C_1$, 
   with $C_1 > 0$ being some constant. Hence, as $\epsilon \rightarrow 0$,
   \begin{align*}
   	\sum_{\ell=0}^{L_{\epsilon}} M_{\ell,\epsilon}^{-1} \leq \sum_{\ell=0}^{L_{\epsilon}} C_1^{-1}\epsilon^2 \leq 2b^2 C_1^{-1} \cdot \sum_{\ell=0}^{L_{\epsilon}}  s^{-2\alpha (L_{\epsilon}-1)} = 2b^2 C_1^{-1} \cdot (L_{\epsilon}+1) \cdot s^{-2\alpha (L_{\epsilon}-1)}  \rightarrow 0 \ .
   \end{align*}
   It follows that the term (B.11) converges to $0$ as $\epsilon \rightarrow 0$. Regarding the term (B.12), we note that $v_{\ell, \epsilon} = \cO(s^{-2\alpha \ell})$. By applying the same analytical approach used for (B.11), we can show that  (B.12) converges to $0$  as $\epsilon \rightarrow 0$.
   
 Next, we show that the term (B.9) $ \converge{p} 0$ as $\epsilon \rightarrow 0$. We first show that $\Var{\zhat}^{-1/2}\sum_{\ell=0}^{L_{\epsilon}}  \sum_{i=1}^{N_{\ell}}w_{i}^{\ell}(\theta) (\overline{Z}(\theta_i^{\ell}, \varpi_{\ell}))^2 \converge{p} 0$ as $\epsilon \rightarrow 0$.
  	 Notice that
  	 \begin{align*}
  	 	&\frac{\sum_{\ell=0}^{L_{\epsilon}}  \sum_{i=1}^{N_{\ell}}w_{i}^{\ell}(\theta) (\overline{Z}(\theta_i^{\ell}, \varpi_{\ell}))^2   }{\sqrt{\Var{\zhat}}} = \sqrt{\frac{ \left( \sum_{\ell=0}^{L_{\epsilon}}  \sum_{i=1}^{N_{\ell}}w_{i}^{\ell}(\theta) (\overline{Z}(\theta_i^{\ell}, \varpi_{\ell}))^2 \right)^2  }{\sum_{\ell=0}^{L_{\epsilon}} \frac{1}{M_{\ell, \epsilon}} \Var{  \deltaZhatsubshort{(1)}{\ell}{\ell} }}} \\
  	 	&\leq \sqrt{\sum_{\ell=0}^{L_{\epsilon}} \frac{M_{\ell,\epsilon} \left( \sum_{i=1}^{N_{\ell}}w_{i}^{\ell}(\theta) (\overline{Z}(\theta_i^{\ell}, \varpi_{\ell}))^2 \right)^2  }{\Var{  \deltaZhatsubshort{(1)}{\ell}{\ell} }}} \leq \sum_{\ell=0}^{L_{\epsilon}} \sqrt{\frac{M_{\ell,\epsilon} \left( \sum_{i=1}^{N_{\ell}}w_{i}^{\ell}(\theta) (\overline{Z}(\theta_i^{\ell}, \varpi_{\ell}))^2 \right)^2  }{\Var{  \deltaZhatsubshort{(1)}{\ell}{\ell} }}} \\
  	 	&=\sum_{\ell=0}^{L_{\epsilon}} \frac{\sqrt{M_{\ell,\epsilon}} \sum_{i=1}^{N_{\ell}}w_{i}^{\ell}(\theta) (\overline{Z}(\theta_i^{\ell}, \varpi_{\ell}))^2 }{\sqrt{\Var{  \deltaZhatsubshort{(1)}{\ell}{\ell} }}} \ .
  	 \end{align*}
Define $g_{\ell} :=\sum_{i=1}^{N_{\ell}}w_{i}^{\ell}(\theta) \Var{\cZ_1(\theta_i^{\ell}, \varpi_{\ell})} / \sqrt{\Var{  \deltaZhatsubshort{(1)}{\ell}{\ell} }}$, which is a constant for any fixed $\ell$. It follows that
    \begin{align*}
       & \mathbb{P}\left\{\left| \sum_{\ell=0}^{L_{\epsilon}} \frac{\sqrt{M_{\ell,\epsilon}} \sum_{i=1}^{N_{\ell}}w_{i}^{\ell}(\theta) (\overline{Z}(\theta_i^{\ell}, \varpi_{\ell}))^2 }{\sqrt{\Var{  \deltaZhatsubshort{(1)}{\ell}{\ell} }}}   \right| > \delta \right\} 
       \leq  \frac{1}{\delta} \cdot \E{ \sum_{\ell=0}^{L_{\epsilon}} \frac{\sqrt{M_{\ell,\epsilon}} \sum_{i=1}^{N_{\ell}}w_{i}^{\ell}(\theta) (\overline{Z}(\theta_i^{\ell}, \varpi_{\ell}))^2 }{\sqrt{\Var{  \deltaZhatsubshort{(1)}{\ell}{\ell} }}}} \\
       =& \frac{1}{\delta} \sum_{\ell=0}^{L_{\epsilon}}\frac{1}{\sqrt{M_{\ell,\epsilon}}} \cdot \frac{\sum_{i=1}^{N_{\ell}}w_{i}^{\ell}(\theta) M_{\ell,\epsilon}\E{  (\overline{Z}(\theta_i^{\ell}, \varpi_{\ell}))^2}}{\sqrt{\Var{  \deltaZhatsubshort{(1)}{\ell}{\ell} }}} 
       =\frac{1}{\delta} \sum_{\ell=0}^{L_{\epsilon}}\frac{1}{\sqrt{M_{\ell,\epsilon}}} \cdot \frac{\sum_{i=1}^{N_{\ell}}w_{i}^{\ell}(\theta) \Var{\cZ_1(\theta_i^{\ell}, \varpi_{\ell})}}{\sqrt{\Var{  \deltaZhatsubshort{(1)}{\ell}{\ell} }}} \\
      = & \frac{1}{\delta} \cdot \sum_{\ell=0}^{L_{\epsilon}}\frac{g_{\ell}}{\sqrt{M_{\ell,\epsilon}}} \longrightarrow 0 \text{ as } \epsilon \rightarrow 0 \ .
   \end{align*}
  	This implies that
  	 \[ \sum_{\ell=0}^{L_{\epsilon}} \frac{\sqrt{M_{\ell,\epsilon}} \sum_{i=1}^{N_{\ell}}w_{i}^{\ell}(\theta) (\overline{Z}(\theta_i^{\ell}, \varpi_{\ell}))^2 }{\sqrt{\Var{  \deltaZhatsubshort{(1)}{\ell}{\ell} }}} \converge{p} 0 \text{ as } \epsilon \rightarrow 0 \ .\]
  	 Similarly, we have 
  	  \[\frac{\sum_{\ell=0}^{L_{\epsilon}}  \sum_{i=1}^{N_{\ell-1}}w_{i}^{\ell-1}(\theta) (\overline{Z}(\theta_i^{\ell-1}, \varpi_{\ell}))^2   }{\sqrt{\Var{\zhat}}} \converge{p} 0 \text{ as } \epsilon \rightarrow 0 \ .\]
  	 Hence, the term (B.9) $\converge{p} 0$ as $\epsilon \rightarrow 0$. The proof is completed by applying   Slutsky's theorem.
\end{proof}

\subsection{Proof of Proposition \ref{thm:clt}}\label{proof:clt_corollary}
\begin{proof} 
	For any $\nu>0$, it follows from \eqref{eq:triangle} and \eqref{eq:norm_var_estimator} that
    \begin{align*}
        &\lim_{\epsilon \rightarrow 0} \sum_{n=1}^{M_{\epsilon}} \E{|Z_{\epsilon, n}|^2 \bfone{|Z_{\epsilon, n}|> \nu}}  \\
        &=  \lim_{\epsilon \rightarrow 0} \sum_{\ell=0}^{L_{\epsilon}} \sum_{m=1}^{M_{\ell,\epsilon}} \E{\frac{\left|\deltaZhatsubshort{(1)}{\ell}{\ell} - \E{\deltaZhatsubshort{(1)}{\ell}{\ell}} \right|^2}{ M_{\ell,\epsilon}^2 \Var{\zhat}} \bfone{\frac{\left|\deltaZhatsubshort{(1)}{\ell}{\ell} - \E{\deltaZhatsubshort{(1)}{\ell}{\ell}} \right|^2}{ M_{\ell,\epsilon}^2 \Var{\zhat}} > \nu^2}} \\
        &= \lim_{\epsilon \rightarrow 0}\sum_{\ell=0}^{L_{\epsilon}} \frac{V_{\ell,\epsilon}(\theta)}{\Var{\zhat}M_{\ell,\epsilon}} \E{\frac{\left|\deltaZhatsubshort{(1)}{\ell}{\ell} - \E{\deltaZhatsubshort{(1)}{\ell}{\ell}}
       \right|^2}{V_{\ell,\epsilon}(\theta)} \times \right. \\ &
        \left.  \qquad \bfone{ \frac{\left|\deltaZhatsubshort{(1)}{\ell}{\ell} - \E{\deltaZhatsubshort{(1)}{\ell}{\ell}} \right|^2}{V_{\ell,\epsilon}(\theta)} > \frac{\Var{\zhat}M_{\ell,\epsilon}^2}{V_{\ell,\epsilon}(\theta)} \nu } } \ . 
    \end{align*}
\end{proof}

\subsection{Proof of Proposition \ref{thm:linde_finite}}\label{proof:linde_finite}
\begin{proof}

By Lemma \ref{lem:variance} in Subsection \ref{app:lemma}, for any $\nu > 0$, we have
{\small
 \begin{align*}
	\lim_{\epsilon \rightarrow 0}\sum_{\ell=0}^{L_{\epsilon}} \E{\frac{\left|\deltaZhatsubshort{(1)}{\ell}{\ell} - \E{\deltaZhatsubshort{(1)}{\ell}{\ell}} \right|^2}{\Var{\zhat}M_{\ell,\epsilon}} \cdot \bfone{ \frac{\left|\deltaZhatsubshort{(1)}{\ell}{\ell} - \E{\deltaZhatsubshort{(1)}{\ell}{\ell}} \right|^2}{V_{\ell,\epsilon}} > \frac{\Var{\zhat}M_{\ell,\epsilon}^2}{V_{\ell,\epsilon}} \nu } } \\
	 \geq  \frac{1}{c_{ub}} \lim_{\epsilon \rightarrow 0}\sum_{\ell=0}^{L_{\epsilon}} \E{\frac{\left|\deltaZhatsubshort{(1)}{\ell}{\ell} - \E{\deltaZhatsubshort{(1)}{\ell}{\ell}} \right|^2}{\epsilon^2 M_{\ell,\epsilon}} \cdot \bfone{ \frac{\left|\deltaZhatsubshort{(1)}{\ell}{\ell} - \E{\deltaZhatsubshort{(1)}{\ell}{\ell}} \right|^2}{V_{\ell,\epsilon}} > \frac{c_{ub} \epsilon^2 M_{\ell,\epsilon}^2}{V_{\ell,\epsilon}} \nu } }  
\end{align*}
}
and
{\small
 \begin{align*}
	\lim_{\epsilon \rightarrow 0}\sum_{\ell=0}^{L_{\epsilon}} \E{\frac{\left|\deltaZhatsubshort{(1)}{\ell}{\ell} - \E{\deltaZhatsubshort{(1)}{\ell}{\ell}} \right|^2}{\Var{\zhat}M_{\ell,\epsilon}} \cdot \bfone{ \frac{\left|\deltaZhatsubshort{(1)}{\ell}{\ell} - \E{\deltaZhatsubshort{(1)}{\ell}{\ell}} \right|^2}{V_{\ell,\epsilon}} > \frac{\Var{\zhat}M_{\ell,\epsilon}^2}{V_{\ell,\epsilon}} \nu } } \\
	 \leq \frac{1}{c_{lb}}\lim_{\epsilon \rightarrow 0}\sum_{\ell=0}^{L_{\epsilon}} \E{\frac{\left|\deltaZhatsubshort{(1)}{\ell}{\ell} - \E{\deltaZhatsubshort{(1)}{\ell}{\ell}} \right|^2}{\epsilon^2 M_{\ell,\epsilon}} \cdot \bfone{ \frac{\left|\deltaZhatsubshort{(1)}{\ell}{\ell} - \E{\deltaZhatsubshort{(1)}{\ell}{\ell}} \right|^2}{V_{\ell,\epsilon}} > \frac{c_{lb} \epsilon^2 M_{\ell,\epsilon}^2}{V_{\ell,\epsilon}} \nu } } \ .
\end{align*}
}
Hence, the Lindeberg's condition in \eqref{eq:lindeberg} is equivalent to 
{\small
\begin{equation}\label{eq:lindeberg_intermediate}
	\lim_{\epsilon \rightarrow 0}\sum_{\ell=0}^{L_{\epsilon}} \E{\frac{\left|\deltaZhatsubshort{(1)}{\ell}{\ell} - \E{\deltaZhatsubshort{(1)}{\ell}{\ell}} \right|^2}{\epsilon^2 M_{\ell,\epsilon}} \cdot \bfone{ \frac{\left|\deltaZhatsubshort{(1)}{\ell}{\ell} - \E{\deltaZhatsubshort{(1)}{\ell}{\ell}} \right|^2}{V_{\ell,\epsilon}} > \frac{\epsilon^2 M_{\ell,\epsilon}^2}{V_{\ell,\epsilon}} \nu } } = 0 \ . \tag{B.13}
\end{equation}
}
In light of the fact that $M_{\ell,\epsilon} = \ceil[\Big]{ 2 \epsilon^{-2} \sqrt{  V_{\ell,\epsilon} / N_{\ell}} \sum_{\ell^{\prime}=0}^{L_{\epsilon}} \sqrt{ V_{\ell^{\prime},\epsilon} N_{\ell^{\prime}}} }$, the term \eqref{eq:lindeberg_intermediate} can be further written as 
\begin{align*}
	\lim_{\epsilon \rightarrow 0}\sum_{\ell=0}^{L_{\epsilon}} \bfone{V_{\ell,\epsilon} > 0}\sqrt{\frac{N_{\ell}}{V_{\ell,\epsilon}}} \E{\left|\deltaZhatsubshort{(1)}{\ell}{\ell} - \E{\deltaZhatsubshort{(1)}{\ell}{\ell}} \right|^2 \times \right. \nonumber \\ \left. \  \bfone{ \left|\deltaZhatsubshort{(1)}{\ell}{\ell} - \E{\deltaZhatsubshort{(1)}{\ell}{\ell}} \right|^2 > \epsilon^2 M_{\ell,\epsilon}^2 \nu } }=0\ .	
\end{align*}
For any $\nu > 0$ and $\ell \in [L_{\epsilon}]$, we have 
{\small
\[ \E{\left|\deltaZhatsubshort{(1)}{\ell}{\ell} - \E{\deltaZhatsubshort{(1)}{\ell}{\ell}} \right|^2 \cdot   \bfone{ \left|\deltaZhatsubshort{(1)}{\ell}{\ell} - \E{\deltaZhatsubshort{(1)}{\ell}{\ell}} \right|^2 > \epsilon^2 M_{\ell,\epsilon}^2 \nu } } \leq V_{\ell,\epsilon}^Z < \infty \ , \]
}
where the last step follows from Lemma \ref{lem:Vl_compare}. By the dominated convergence theorem, we have
\begin{align*}
	\lim_{\epsilon \rightarrow 0}\sum_{\ell=0}^{L_{\epsilon}} \bfone{V_{\ell,\epsilon} > 0}\sqrt{\frac{N_{\ell}}{V_{\ell,\epsilon}}} \E{\left|\deltaZhatsubshort{(1)}{\ell}{\ell} - \E{\deltaZhatsubshort{(1)}{\ell}{\ell}} \right|^2 \times \right. \nonumber \\ \left. \  \bfone{ \left|\deltaZhatsubshort{(1)}{\ell}{\ell} - \E{\deltaZhatsubshort{(1)}{\ell}{\ell}} \right|^2 > \epsilon^2 M_{\ell,\epsilon}^2 \nu } } \\
	=\sum_{\ell=0}^{L_{\epsilon}} \bfone{V_{\ell,\epsilon} > 0}\sqrt{\frac{N_{\ell}}{V_{\ell,\epsilon}}} \lim_{\epsilon \rightarrow 0}\E{\left|\deltaZhatsubshort{(1)}{\ell}{\ell} - \E{\deltaZhatsubshort{(1)}{\ell}{\ell}} \right|^2 \times \right. \nonumber \\ \left. \  \bfone{ \left|\deltaZhatsubshort{(1)}{\ell}{\ell} - \E{\deltaZhatsubshort{(1)}{\ell}{\ell}} \right|^2 > \epsilon^2 M_{\ell,\epsilon}^2 \nu } } \ .
\end{align*}
 Since for any $\ell \in [L_{\epsilon}]$, $\lim_{\epsilon \rightarrow 0} \epsilon^2 M_{\ell,\epsilon}^2 \geq \lim_{\epsilon \rightarrow 0} \epsilon^{-2}\frac{V_{\ell,\epsilon}}{N_{\ell}}S^2_{L_{\epsilon}} = \infty$. Therefore, we have
 \begin{align*}
 \begin{split}
 	\bfone{V_{\ell,\epsilon} > 0}\sqrt{\frac{N_{\ell}}{V_{\ell,\epsilon}}} \lim_{\epsilon \rightarrow 0}\E{\left|\deltaZhatsubshort{(1)}{\ell}{\ell} - \E{\deltaZhatsubshort{(1)}{\ell}{\ell}} \right|^2 \times \right. \nonumber \\ \left. \  \bfone{ \left|\deltaZhatsubshort{(1)}{\ell}{\ell} - \E{\deltaZhatsubshort{(1)}{\ell}{\ell}} \right|^2 > \epsilon^2 M_{\ell,\epsilon}^2 \nu } } =0 \ . 
 \end{split}
 \end{align*} 
 \end{proof}

\subsection{Proof of Proposition \ref{thm:linde_infinite}}\label{proof:linde_infinite}
\begin{proof}
The proof is inspired by the proof of Theorem 2.6 in \cite{hoel2019central}. We first establish the asymptotic normality of the normalized estimator given in \eqref{eq:norm_var_estimator} under Assumption \ref{asm:basic_convergence}. We have
	\begin{align*}
		&\sum_{\ell=0}^{L_{\epsilon}} \frac{V_{\ell,\epsilon}}{\Var{\zhat}M_{\ell,\epsilon}} \E{ V_{\ell,\epsilon}^{-1} \left|\deltaZhatsubshort{(1)}{\ell}{\ell} - \E{\deltaZhatsubshort{(1)}{\ell}{\ell}} \right|^2  \right. \\ &\left. \quad \times \bfone{ V_{\ell,\epsilon}^{-1} \left|\deltaZhatsubshort{(1)}{\ell}{\ell} - \E{\deltaZhatsubshort{(1)}{\ell}{\ell}} \right|^2  >  V_{\ell,\epsilon}^{-1} \Var{\zhat}M_{\ell,\epsilon}^2 \nu }} \\
		\leq & \sum_{\ell=0}^{L_{\epsilon}} \frac{\sqrt{V_{\ell,\epsilon}N_{\ell}}}{\epsilon^{-2}\Var{\zhat}S_{L_{\epsilon}}} \E{ V_{\ell,\epsilon}^{-1}  \left|\deltaZhatsubshort{(1)}{\ell}{\ell} - \E{ \deltaZhatsubshort{(1)}{\ell}{\ell}} \right|^2  \right. \\ &\left. \quad \times \bfone{V_{\ell,\epsilon}^{-1} \left|\deltaZhatsubshort{(1)}{\ell}{\ell} - \E{\deltaZhatsubshort{(1)}{\ell}{\ell}} \right|^2 > V_{\ell,\epsilon}^{-1} \Var{\zhat}M_{\ell,\epsilon}^2 \nu } } \\
		=& \frac{\epsilon^2}{\Var{\zhat}} \sum_{\ell=0}^{L_{\epsilon}} \frac{\sqrt{V_{\ell,\epsilon}N_{\ell}}}{S_{L_{\epsilon}}} \E{V_{\ell,\epsilon}^{-1} \left|\deltaZhatsubshort{(1)}{\ell}{\ell} - \E{\deltaZhatsubshort{(1)}{\ell}{\ell}} \right|^2 \right. \\ &\left. \quad \times \bfone{ V_{\ell,\epsilon}^{-1} \left|\deltaZhatsubshort{(1)}{\ell}{\ell} - \E{\deltaZhatsubshort{(1)}{\ell}{\ell}} \right|^2  > V_{\ell,\epsilon}^{-1} \Var{\zhat}M_{\ell,\epsilon}^2 \nu } } \ ,
	\end{align*}
where the last inequality follows by recalling the definition $M_{\ell, \epsilon} \coloneqq \ceil[\Big]{ 2 \epsilon^{-2} \sqrt{ V_{\ell,\epsilon} / N_{\ell}} S_{L_{\epsilon}} }$.
Let $\tilde{L}_{\epsilon}$ be a monotonically decreasing function of $\epsilon$ satisfying $\lim_{\epsilon \rightarrow 0}\tilde{L}_{\epsilon} = \infty \text{ and } \lim_{\epsilon \rightarrow 0} S_{\tilde{L}_{\epsilon}}/S_{L_{\epsilon}} = 0$.
	It is easy to verify that $\tilde{L}_{\epsilon} = \min \left\{\ell \in \mathbb{N} \mid S_{\ell+1} \geq \sqrt{S_{L_{\epsilon}}}\right\}$ satisfies these conditions. Then, we have
	
		\begin{align*}
		&\sum_{\ell=0}^{L_{\epsilon}} \frac{\sqrt{V_{\ell,\epsilon}N_{\ell}}}{S_{L_{\epsilon}}} \E{ V_{\ell,\epsilon}^{-1} \left|\deltaZhatsubshort{(1)}{\ell}{\ell} - \E{\deltaZhatsubshort{(1)}{\ell}{\ell}} \right|^2 \right. \\ &\left. \quad \times \bfone{ V_{\ell,\epsilon}^{-1} \left|\deltaZhatsubshort{(1)}{\ell}{\ell} - \E{\deltaZhatsubshort{(1)}{\ell}{\ell}} \right|^2 > V_{\ell,\epsilon}^{-1} \Var{\zhat}M_{\ell,\epsilon}^2 \nu } } \\
		\leq& \sum_{\ell=0}^{\tilde{L}_{\epsilon}}  \frac{\sqrt{V_{\ell,\epsilon}N_{\ell}}}{S_{L_{\epsilon}}} + \sum_{\ell=\tilde{L}_{\epsilon}+1}^{L_{\epsilon}} \frac{\sqrt{V_{\ell,\epsilon}N_{\ell}}}{S_{L_{\epsilon}}} \E{ V_{\ell,\epsilon}^{-1} \left|\deltaZhatsubshort{(1)}{\ell}{\ell} - \E{\deltaZhatsubshort{(1)}{\ell}{\ell}} \right|^2 \right. \\ &\left. \quad \times \bfone{ V_{\ell,\epsilon}^{-1} \left|\deltaZhatsubshort{(1)}{\ell}{\ell} - \E{\deltaZhatsubshort{(1)}{\ell}{\ell}} \right|^2 > V_{\ell,\epsilon}^{-1} \Var{\zhat}M_{\ell,\epsilon}^2 \nu } } \\
		\leq & \frac{S_{\tilde{L}_{\epsilon}}}{S_{L_{\epsilon}}} + \frac{S_{L_{\epsilon}}-S_{\tilde{L}_{\epsilon}}}{S_{L_{\epsilon}}} \sup_{\ell > \tilde{L}_{\epsilon}} \E{V_{\ell,\epsilon}^{-1} \left|\deltaZhatsubshort{(1)}{\ell}{\ell} - \E{\deltaZhatsubshort{(1)}{\ell}{\ell}} \right|^2 \right. \\ &\left. \quad \times \bfone{ V_{\ell,\epsilon}^{-1} \left|\deltaZhatsubshort{(1)}{\ell}{\ell} - \E{\deltaZhatsubshort{(1)}{\ell}{\ell}} \right|^2 > V_{\ell,\epsilon}^{-1} \Var{\zhat}M_{\ell,\epsilon}^2 \nu } }  \ .
	\end{align*}
Assumption \ref{asm:basic_convergence} implies that
	\begin{align*}
		\lim_{\epsilon \rightarrow 0} \Var{\zhat} M_{\ell,\epsilon}^2 &\geq \lim_{\epsilon \rightarrow 0} \Var{\zhat} \epsilon^{-4} \frac{V_{\ell,\epsilon}}{N_{\ell}}S_{L_{\epsilon}}^2  = \lim_{\epsilon \rightarrow 0} \left(\epsilon^{-2}\Var{\zhat} \sqrt{\frac{V_{\ell,\epsilon}}{N_{\ell}}} S_{L_{\epsilon}} \right) \cdot \left(\epsilon^{-2} \sqrt{\frac{V_{\ell,\epsilon}}{N_{\ell}}}  S_{L_{\epsilon}} \right) = \infty .
	\end{align*}
It follows that
	\begin{align*}
		&\lim_{\epsilon \rightarrow 0}\sum_{\ell=0}^{L_{\epsilon}} \frac{\sqrt{V_{\ell,\epsilon}N_{\ell}}}{S_{L_{\epsilon}}} \E{V_{\ell,\epsilon}^{-1} \left|\deltaZhatsubshort{(1)}{\ell}{\ell} - \E{\deltaZhatsubshort{(1)}{\ell}{\ell}} \right|^2 \right. \\ &\left. \quad \times \bfone{ V_{\ell,\epsilon}^{-1} \left|\deltaZhatsubshort{(1)}{\ell}{\ell} - \E{\deltaZhatsubshort{(1)}{\ell}{\ell}} \right|^2 > V_{\ell,\epsilon}^{-1} \Var{\zhat}M_{\ell,\epsilon}^2 \nu } } \\
		\leq  &\lim_{\epsilon \rightarrow 0}\frac{S_{\tilde{L}_{\epsilon}}}{S_{L_{\epsilon}}}  + \limsup_{\ell \rightarrow \infty} \E{ V_{\ell,\epsilon}^{-1} \left|\deltaZhatsubshort{(1)}{\ell}{\ell} - \E{\deltaZhatsubshort{(1)}{\ell}{\ell}} \right|^2 \right. \\ &\left. \quad \times \bfone{V_{\ell,\epsilon}^{-1} \left|\deltaZhatsubshort{(1)}{\ell}{\ell} - \E{\deltaZhatsubshort{(1)}{\ell}{\ell}} \right|^2 > V_{\ell,\epsilon}^{-1} \Var{\zhat}M_{\ell,\epsilon}^2 \nu }} =0 \ .
	\end{align*}
Hence, the Lindeberg's condition \eqref{eq:lindeberg} is satisfied. To show that under Assumption \ref{asm:Vl_Sl_order}, the normalized estimator given in \eqref{eq:norm_var_estimator} is asymptotically normal, we only need to show that Assumption \ref{asm:Vl_Sl_order} is a sufficient condition for Assumption \ref{asm:basic_convergence}. Define $p \coloneqq \inf_{\ell} V_{\ell,\epsilon}^{Z} / V_{\ell,\epsilon}$. Then  $p$ is positive and finite by Lemma \ref{lem:Vl_compare}. It follows that
\begin{align*}
	 \frac{\Var{\zhat}}{\epsilon^2} = \sum_{\ell=0}^{L_{\epsilon}} \frac{V_{\ell,\epsilon}^{Z}}{\epsilon^2 M_{\ell,\epsilon}} \geq \sum_{\ell=0}^{L_{\epsilon}} \frac{V_{\ell,\epsilon}^{Z}}{\sqrt{\frac{V_{\ell,\epsilon}}{N_{\ell}}} S_{L_{\epsilon}} + \epsilon^2} 
  = \sum_{\ell=0}^{L_{\epsilon}} \frac{V_{\ell,\epsilon}^{Z} / V_{\ell,\epsilon}}{S_{L_{\epsilon}}/\sqrt{V_{\ell,\epsilon} N_{\ell}} + \epsilon^2 / V_{\ell,\epsilon}} \geq p \sum_{\ell=0}^{L_{\epsilon}} \frac{\sqrt{V_{\ell,\epsilon} N_{\ell}}/S_{L_{\epsilon}}}{1 + \sqrt{N_{\ell} / V_{\ell,\epsilon}}\epsilon^2/S_{L_{\epsilon}}} \ .
\end{align*}
By Assumption \ref{asm:Vl_Sl_order}, $\lim_{\epsilon \rightarrow 0} \sqrt{N_{\ell} / V_{\ell,\epsilon}}\epsilon^2/S_{L_{\epsilon}} \leq \lim_{\epsilon \rightarrow 0} \epsilon^{2 - \gamma/(2\alpha)} /(\sqrt{V_{L_{\epsilon}}}S_{L_{\epsilon}}) < \infty$. Hence, there exists $q < \infty$ such that $q = \max_{\epsilon} \sqrt{N_{L_{\epsilon}} / V_{L_{\epsilon},\epsilon}}\epsilon^2/S_{L_{\epsilon}} $, and the following inequality holds:
\begin{align*}
	\frac{\Var{\zhat}}{\epsilon^2} \geq \frac{p}{1 + q} \sum_{\ell=0}^{L_{\epsilon}}\frac{\sqrt{V_{\ell,\epsilon} N_{\ell}}}{S_{L_{\epsilon}}} = \frac{p}{1+q}.
\end{align*}
Since $p > 0$ and $q > 0$, Assumption \ref{asm:Vl_Sl_order} is a sufficient condition for Assumption \ref{asm:basic_convergence}. The proof is complete.
\end{proof}

\section{Estimating the Integrated Variance via Bootstrapping}\label{app:bootstrap}

\begin{algorithm}
\caption{Estimating the variance of $\DeltaVhatsub{\ell}{\theta, M_{\ell}, \varpi_{\ell}}$ for $\forall \ell \in [L]$ via bootstrapping}\label{alg:bootstrap}
\begin{algorithmic}[1]

\State \textbf{Input:}  The  set of simulation outputs  $\{\cY(\theta, \omega_i), \theta \in \mathcal{T}_{\ell}, i \in [M_{\ell}]^{+}\}$ with the $\omega_i$'s   drawn using the random number stream $\varpi_{\ell}$, the prediction-point set $\mathcal{P}$, and the bootstrap sample size $B$
\State \textbf{Output:} The variance estimator $V_{B}(\theta, M_{\ell}, \varpi_{\ell})$

\State $b \gets 1$; \Comment{Initialize the bootstrap index}
\While{$b \leq B$}
\State	\algmultiline{Randomly and independently draw $M_{\ell}$ observations with replacement from $\{\cY(\theta, \omega_i), i \in [M_{\ell}]^{+}\}$ at each $ \theta \in \mathcal{T}_{\ell}$  and obtain the $b$th  bootstrap sample of outputs $\mathbb{Y}_{b} \coloneqq \{\cY^{\ast}_{i,b}(\theta), \theta \in \mathcal{T}_{\ell}, i \in [M_{\ell}]^{+}\}$; }
\State	\algmultiline{Build the metamodel-based estimators $\Vhatsub{\ell, b}{\theta, M_{\ell}, \varpi_{\ell}}$ and $\Vhatsub{\ell-1, b}{\theta, M_{\ell}, \varpi_{\ell}}$ according to \eqref{eq:metamodel} based on $\mathbb{Y}_{b}$;}
\State	$b \gets b + 1$;
\EndWhile
    
\State	$\Vbarsub{\ell, b}{\theta, M_{\ell}, \varpi_{\ell}} \gets B^{-1}\sum_{b=1}^{B}\widehat{\mathbb{V}}_{\ell, b}(\theta, M_{\ell}, \varpi_{\ell}) $; \Comment{Calculate the bootstrap sample mean of level $\ell$}
\State	$\Vbarsub{\ell-1, b}{\theta, M_{\ell}, \varpi_{\ell}} \gets B^{-1}\sum_{b=1}^{B}\widehat{\mathbb{V}}_{\ell-1, b}(\theta, M_{\ell}, \varpi_{\ell})$; \Comment{Calculate the bootstrap sample mean of level $\ell-1$}
\State	$V_{B}(\theta, M_{\ell}, \varpi_{\ell}) \gets (B-1)^{-1} \sum_{b=1}^{B}\left( \Vhatsub{\ell, b}{\theta, M_{\ell}, \varpi_{\ell}} - \Vhatsub{\ell-1, b}{\theta, M_{\ell}, \varpi_{\ell}} -  \Delta\bar{\mathbb{V}}_{\ell, b}(\theta, M_{\ell}, \varpi_{\ell}) \right)^2$ for $\forall \theta \in \mathcal{P}$, where $\Delta\bar{\mathbb{V}}_{\ell, b}(\theta, M_{\ell}, \varpi_{\ell}) = \Vbarsub{\ell, b}{\theta, M_{\ell}, \varpi_{\ell}} - \Vbarsub{\ell-1, b}{\theta, M_{\ell}, \varpi_{\ell}}$.
\end{algorithmic}
\end{algorithm}

\bibliographystyle{siamplain}

\bibliography{Reference.bib}

\end{document}